\documentclass[orivec, conference]{IEEEtran}
\usepackage{graphicx, cite}
\usepackage[usenames,dvipsnames]{xcolor}
\usepackage{float}
\usepackage{url}
\usepackage{enumerate}
\usepackage{xspace}
\usepackage{pifont} 
\usepackage{amssymb,amstext,mathtools,amsmath,amsthm}
\usepackage{thmtools,thm-restate}
\usepackage{boxedminipage}
\usepackage{paralist}
\usepackage{minibox}
\usepackage{bbm, models}
\usepackage{tikz, tkz-graph, misc}
\usetikzlibrary{arrows, shapes, positioning, calc}

\makeatletter
\renewenvironment{proof}[1][\proofname]{\par
\pushQED{\qed}%
\normalfont \topsep6\p@\@plus6\p@\relax
\trivlist
\item\relax
{\bf
#1\@addpunct{.}}\hspace\labelsep\ignorespaces
}{%
\popQED\endtrivlist\@endpefalse
}
\makeatother
\usepackage{misc}

\newtheorem{example}{Example}
\newtheorem{lemma}{Lemma}
\newtheorem{theorem}{Theorem}
\newtheorem{proposition}{Proposition}
\newtheorem{definition}{Definition}
\newtheorem{corollary}{Corollary}

\newcommand{\mc}[1]{\ensuremath{\mathcal{#1}}\xspace}
\newcommand{\mi}[1]{\ensuremath{\mathit{#1}}\xspace}

\renewcommand{\EL}{\mc{E\kern-0.1emL}}
\newcommand{\Ind}{\mn{dom}}

\renewcommand{\phi}{\varphi}
\renewcommand{\epsilon}{\varepsilon}

\newcommand{\adom}{\ensuremath{\mathsf{dom}}}

\begin{document}

\title{When is Ontology-Mediated Querying Efficient?}

\author{\IEEEauthorblockN{Pablo Barcel\'o}
\IEEEauthorblockA{DCC, U of Chile \& IMFD Chile\\
pbarcelo@dcc.uchile.cl}
\and
\IEEEauthorblockN{Cristina Feier}
\IEEEauthorblockA{University of Bremen\\
feier@uni-bremen.de}
\and
\IEEEauthorblockN{Carsten Lutz}
\IEEEauthorblockA{University of Bremen\\
clu@uni-bremen.de}
\and
\IEEEauthorblockN{Andreas Pieris}
\IEEEauthorblockA{University of Edinburgh\\
apieris@inf.ed.ac.uk }

}

\IEEEoverridecommandlockouts
  \IEEEpubid{\makebox[\columnwidth]{978-1-7281-3608-0/19/\$31.00~
  \copyright2019 IEEE \hfill} \hspace{\columnsep}\makebox[\columnwidth]{ }}
\maketitle


\begin{abstract}
  In ontology-mediated querying, description logic (DL) ontologies are
  used to enrich incomplete data with domain knowledge which results
  in more complete answers to queries. However, the evaluation of
  ontology-mediated queries (OMQs) over relational databases is
  computationally hard.
  This raises the question when OMQ evaluation is efficient, in the
  sense of being tractable in combined complexity or fixed-parameter
  tractable. We study this question for a range of ontology-mediated
  query languages based on several important and widely-used DLs,
  using unions of conjunctive queries as the actual queries. For the
  DL $\ELHIbot$, we provide a characterization of the classes of OMQs
  that are fixed-parameter tractable. For its fragment $\ELpoly$,
  which restricts the use of inverse roles, we provide a
  characterization of the classes of OMQs that are tractable in
  combined complexity.  Both results are in terms of equivalence to
  OMQs of bounded tree width and rest on a reasonable assumption from
  parameterized complexity theory. They are similar in spirit to
  Grohe's seminal characterization of the tractable classes of
  conjunctive queries over relational databases. We further study the
  complexity of the meta problem of deciding whether a given OMQ is
  equivalent to an OMQ of bounded tree width, providing several
  completeness results that range from \NPclass~to \TwoExpTime, depending
  on the DL used. We also consider the DL-Lite family of DLs,
  including members that, unlike $\ELHIbot$, admit functional roles.
\end{abstract}

\setlength{\textfloatsep}{2mm}

\section{Introduction}

An ontology-mediated query (OMQ) is a database query enriched with an
ontology that contains domain knowledge
\cite{DBLP:conf/rweb/CalvaneseGLLPRR09,DBLP:journals/tods/BienvenuCLW14,DBLP:conf/rweb/BienvenuO15}. Adding
the ontology serves the purpose of delivering more complete answers to
queries and of enriching the vocabulary available for
querying. Ontologies are often formulated in description logics (DLs),
a family of ontology languages that has emerged from artificial
intelligence, underlies the OWL 2 recommendation for ontology
languages on the web, and whose members can be seen as decidable
fragments of (two-variable guarded) first-order logic
\cite{DBLP:books/daglib/0041477}. The actual queries in OMQs are
typically conjunctive queries (CQs), unions of CQs (UCQs), or
fragments thereof, query languages that are at the heart of relational
databases~\cite{AHV}.

An \emph{OMQ language} is a pair $(\Lmc,\Qmc)$ with \Lmc an ontology
language and \Qmc a query language
\cite{DBLP:journals/tods/BienvenuCLW14}. Depending on the OMQ language
chosen, the computational cost of evaluating OMQs can be high. Both
the combined complexity and the data complexity of OMQ evaluation have
received considerable interest in the literature, where data
complexity means that the OMQ is fixed while the database is treated
as an input, in line with the standard setup from database theory. The
combined complexity ranges from \PTime
\cite{DBLP:conf/ijcai/BienvenuOSX13,DBLP:journals/jacm/BienvenuKKPZ18,DBLP:conf/pods/BienvenuKKPRZ17}
to at least \TwoExpTime
\cite{DBLP:conf/cade/Lutz08,DBLP:conf/ijcai/EiterLOS09,DBLP:conf/kr/NgoOS16}. Regarding
the data complexity, there is an important divide between DLs that
include negation or disjunction and induce {\sc coNP}-hardness, and
DLs that do not
\cite{DBLP:journals/jar/HustadtMS07,DBLP:conf/jelia/EiterGOS08,DBLP:journals/ai/CalvaneseGLLR13,DBLP:conf/lpar/KrisnadhiL07}. Studying
the data complexity of the former has turned out to be closely related
to the complexity of constraint satisfaction problems
(CSPs)~\cite{DBLP:journals/tods/BienvenuCLW14}. 

In this paper, we explore the frontiers of two important notions of
OMQ tractability, \PTime combined complexity and fixed-parameter
tractability (\FPTclass) where the parameter is the size of the
OMQ. We believe these to be more realistic than \PTime data complexity
given that ontologies can get large. In fact, hundreds or thousands of
logical statements are not unusual in real world ontologies, and this
can even go up to hundreds of thousands in extreme but important cases
such as {\sc snomed CT}~\cite{DBLP:conf/amia/SpackmanCC97}.
However, there are only few OMQ languages that have \PTime combined
complexity or are FPT without imposing serious restrictions on the
shape of the query or the ontology. An example for the former is
$(\mathcal{ELH}^{dr}_\bot,\text{AQ})$ where $\cdot^{dr}$ stands for
domain and range restrictions and AQ refers to the class of atomic
queries of the form $A(x)$, $A$ a concept name; this result is
implicit in \cite{LutzTomanWolter-IJCAI09}. An example for FPT is
$(\mathcal{ELHI}_\bot,\text{AQ})$; we are not aware of this being
stated explicitly anywhere, but it is not too hard to prove using
standard means. Both $\mathcal{ELH}^{dr}_\bot$ and
$\mathcal{ELHI}_\bot$ are widely used DLs that underpin profiles of
the OWL 2 recommendation \cite{owl2-overview}.  One should think of
the former as an important DL without negation and disjunction, and of
the latter as an important DL in which basic reasoning problems such
as subsumption are still in \PTime.  Note that the unrestricted use of
CQs and UCQs rules out both of the considered complexities since
(U)CQ-evaluation is \NPclass-complete in complexity and
$\Wclass[1]$-hard, thus most likely not fixed-parameter
tractable~\cite{GroheBook}.


A remarkable
result by Grohe precisely characterizes the (recursively enumerable)
classes of CQs over schemas of bounded arity that can be evaluated in
\PTime~\cite{DBLP:journals/jacm/Grohe07}: this is the case if and only
if for some $k$, every CQ in the class is equivalent to a CQ of tree
width $k$, unless the assumption from parameterized complexity theory
that $\FPTclass \neq \Wclass[1]$ fails. Grohe's result also
establishes that \PTime complexity and \FPTclass coincide for
evaluating CQs.  
A generalization to UCQs is in \cite{DBLP:journals/tocl/Chen14}, more
details are given in Sections~\ref{sect:prelims} and~\ref{sect:eval}.

Our main contribution is to establish the following, under the widely-held
assumption that $\FPTclass \neq \Wclass[1]$:
\begin{enumerate}

\item a precise characterization of classes of OMQs from
  $(\ELHIbot,\text{UCQ})$ for which evaluation is in FPT as those
  classes in which each OMQ is equivalent to an OMQ of bounded tree
  width (Section~\ref{sect:fpt});

\item a precise characterization of classes of OMQs from $(\ELpoly,\text{UCQ})$ that admit \PTime evaluation as
  those classes in which each OMQ is equivalent to an OMQ of bounded
  tree width  (Section~\ref{sect:ptimeupper}),

\end{enumerate}
where an OMQ has bounded tree width if the actual query in it
has. 
%
Regarding Point~1, we also observe that the runtime of the FPT
algorithm can be made single exponential in the parameter.  In
Point~2, we work under the assumption that the ontology does not
introduce relations beyond those admitted in the database; such
additional relations are introduced to enrich the vocabulary available
for querying, bearing a similarity to views in relational
databases~\cite{Levy95}.
%
Given that \ELpoly is a fragment of \ELHIbot, Points~1 and~2 imply that \PTime
complexity and \FPTclass coincide in $(\ELpoly,\text{UCQ})$. To prove
the
`upper bound' of Point~2, we use
existential pebble
games adapted in a careful way to OMQs. For the rather non-trivial
`lower bound', we build on Grohe's result. Here, the fact that
OMQs can introduce additional relations results in serious challenges; in fact, several fundamental techniques that are standard in relational databases must be replaced by more subtle ones.

Related to our second main result, it has been shown in
in~\cite{DBLP:conf/ijcai/BienvenuOSX13} that whenever \Qmc is a class
of CQs that can be evaluated in \PTime, then the same is true for OMQs
from $(\mathcal{ELH},\Qmc)$.  In particular, \Qmc might be the class
of CQs of tree width bounded by some~$k$. Our tractability results are
stronger than this: adding an ontology can lower the complexity of a
(U)CQ and it is in fact not hard to see that there are classes of OMQs
from $(\mathcal{EL},\text{CQ})$ that can be evaluated in \PTime, but
the class of CQs used in them cannot. In our characterizations,
equivalence to an OMQ $Q$ of bounded tree width includes the case
that $Q$ uses a different ontology than the original OMQ. We also
show, however, that in most of the studied cases there is no benefit
in changing the ontology. More loosely related studies of the combined
complexity of OMQs in which the ontology is formulated in fragments of
$\ELHIbot$ such as DL-Lite$^\Rmc$ and DL-Lite$_{\mn{horn}}^\Rmc$,
important in ontology-based data integration, and where the queries
have bounded tree width are in
\cite{DBLP:conf/pods/BienvenuKKPRZ17,DBLP:journals/jacm/BienvenuKKPZ18}.

We further study the complexity of the meta problem of deciding
whether a given OMQ is equivalent to an OMQ of bounded tree width
(Section~\ref{sect:deciding}).  Decidability is needed for the
characterizations described above, but we also consider this question
interesting in its own right. Our results range from $\Pi^p_2$ between
$(\text{DL-Lite}^\Rmc,\text{CQ})$ and
$(\text{DL-Lite}^\Rmc_{\mn{horn}},\text{UCQ})$ via \ExpTime between
$(\EL,\text{CQ})$ and $(\ELpoly,\text{UCQ})$ to \TwoExpTime between
$(\ELI,\text{CQ})$ and $(\ELHIbot,\text{UCQ})$; all these are
completeness results. As an important
special case, we consider the full database schema, meaning that the
ontology cannot introduce additional relations. There, the complexity
drops considerably, to \NPclass, \NPclass, and \ExpTime,
respectively. The case of the full schema is also interesting because
it admits constructions that are closer to the case of relational
databases. We remark
that when the schema is full, the problems studied here are closely
related to the evaluation of (U)CQs of bounded tree width over
relational databases with integrity
constraints~\cite{PabloAndreasPods}. However, the constraints
languages considered there are different from ontology languages and
the connection breaks when the schema is not full.

Finally, we take a first glimpse at ontology languages that include a
form of counting, more precisely at DL-Lite$^\Fmc$, in which some
binary relations can be declared to be partial functions
(Section~\ref{sec:dl-lite-f}). This turns out to be closely related to
the evaluation of UCQs over relational databases 
in the presence of key dependencies, as studied by
Figueira~\cite{Figueira16}. We show that evaluating OMQs that are
equivalent to an OMQ of tree width bounded by some $k$ is in FPT and
even in PTime when $k=1$, and that the meta problem of deciding
whether an OMQ belongs to this class is decidable in \ThreeExpTime and
\NPclass-complete when $k=1$.  In this part, we assume the full
database schema and Boolean queries. For the case $k>1$, we
additionally assume that the ontology cannot be changed.

Most of the proof details are deferred to the appendix, available at
\url{http://www.informatik.uni-bremen.de/tdki}.

\section{Preliminaries}
\label{sect:prelims}

\subsection{Databases and Queries}

{\bf {\em Databases.}}  Let \NC, \NR, and \Cbf be countably infinite sets of
\emph{concept names}, \emph{role names}, and \emph{constants}, respectively.  A
\emph{database} \Dmc is a finite set of \emph{facts} of the form
$A(a)$ and $r(a,b)$ where, here and in the remainder of the paper, $A$
ranges over \NC, $r$ ranges over \NR, and $a,b$ range
over \Cbf. We denote by $\mn{dom}(\Dmc)$ the set of constants
used in \Dmc and sometimes write $r^-(a,b) \in \Dmc$ in place of
$r(b,a) \in \Dmc$.  A \emph{schema} $\Sbf$ is a set of concept and
role names. An \emph{$\Sbf$-database} is a database that uses only
concept and role names from  $\Sbf$. Note that as usual in
the context of DLs, databases can only refer to unary and binary
relations, i.e., concept and role names, but not to relations of
higher (or lower) arity. We shall sometimes consider also infinite databases and
then say so explicitly.

A \emph{homomorphism} from database $\Dmc_1$ to database $\Dmc_2$ is a
function $h :\mn{dom}(\Dmc_1) \to \mn{dom}(\Dmc_2)$ such that $A(h(a))
\in \Dmc_2$ for every $A(a) \in \Dmc_1$, and $r(h(a),h(b)) \in \Dmc_2$
for every $r(a,b) \in \Dmc_1$.  We write $\Dmc_1 \rightarrow \Dmc_2$
if there is a homomorphism from $\Dmc_1$ to $\Dmc_2$.
For a database \Dmc and tuple (or set) $\abf$ of constants,
we write $\Dmc|_\abf$ to denote the restriction of \Dmc to facts that
involve only constants from \abf.

\smallskip

{\bf {\em Conjunctive Queries.}}  A \emph{conjunctive query (CQ)} is of the
form $q = \exists \ybf\,\varphi(\xbf,\ybf), $ where \xbf and \ybf are
tuples of variables and $\varphi(\xbf,\ybf)$ is a conjunction of
\emph{atoms} of the form
$A(x)$ and $r(x,y)$ with $x,y$ variables.  We call \xbf the
\emph{answer variables} of~$q$, \ybf the \emph{quantified variables},
and use $\mn{var}(q)$ to denote $\xbf \cup \ybf$.
We take the liberty
to write $\alpha \in q$ to indicate that $\alpha$ is an atom in $q$
and sometimes write $r^-(x,y) \in q$ in place of $r(y,x) \in q$. We
neither admit equality atoms nor constants in CQs, but all results in
this paper remain valid when both are admitted.

Every CQ $q$ can be seen as a database $\Dmc_q$ by dropping the
existential quantifier prefix and viewing variables as constants.  A
\emph{homomorphism} from $q$ to a database \Dmc is a homomorphism from
$\Dmc_q$ to $\Dmc$. We write $\Dmc \models q(\abf)$ and call the
tuple of constants \abf an \emph{answer to $q$ on \Dmc} if there is a
homomorphism $h$ from $q$ to $\Dmc$ with $h(\xbf) = \abf$, \xbf the
answer variables of $q$. Moreover, $q(\Dmc)$ denotes the set of all
answers to $q$ on \Dmc.

A \emph{union of conjunctive queries (UCQ)} $q$ is a disjunction of one
or more CQs that all have the same answer
variables.
\emph{Answers} to a UCQ $q$ are defined in the expected
way, and so is $q(\Dmc)$.
The \emph{arity} of a (U)CQ $q$ is defined as the number of answer variables in it, and we use the term \emph{Boolean} interchangeably with
`arity zero'.

A \emph{homomorphism} from a CQ $q_1(\xbf)$ to a CQ $q_2(\xbf)$
 is a homomorphism $h$
from $\Dmc_{q_1}$ to $\Dmc_{q_2}$ such that $h(\xbf)=\xbf$.  We write
  $q_1 \rightarrow q_2$ if such a homomorphism exists. For a CQ
$q$ and tuple (or set) \zbf of variables, 
the restriction $q|_\zbf$ of $q$ to the variables in \zbf is defined
in the expected way (it might involve a change of arity).

\smallskip

{\bf {\em Tree Width.}} The {\em evaluation problem for CQs} takes as input a CQ $q(\xbf)$, a database $\Dmc$, and a candidate
answer $\abf$, and asks whether $\abf \in q(\Dmc)$.
While in general the evaluation problem for CQs is {\rm NP}-complete, it becomes tractable
for CQs of {\em bounded tree width} \cite{DKV02}. The notion of tree width of a CQ, which is central to our work, is introduced next.

A \emph{tree decomposition} of an undirected graph
$G=(V,E)$ is a triple $D=(V_D,E_D,\mu)$ where $(V_D,E_D)$ is an
undirected tree and $\mu:V_D \rightarrow 2^V$ a function such that
\begin{itemize}
\item $\bigcup_{t \in V_D} \mu(t)=V$;
\item if $\{v_1,v_2\} \in E$, then $v_1,v_2 \in \mu(t)$ for some $t \in V_D$;
\item for every $v \in V$, the subgraph of $(V_D,E_D)$ induced by
   the vertex set $\{t \in V_D \mid v \in \mu(t)\}$ is connected.
\end{itemize}
The \emph{width} of $D$ is $\max_{t \in V_D}
|\mu(t)|-1$ and the \emph{tree width} of~$G$ is the smallest
width of any tree decomposition of $G$.

Each database \Dmc is associated with an undirected graph (without
self loops) $G_\Dmc$, its \emph{Gaifman graph}, defined as follows:
the nodes of $G_\Dmc$ are the constants in \Dmc and there is an edge
$\{a,b\}$ iff $\Dmc$ contains a fact that involves both $a$ and
$b$. Each CQ $q$ is associated with the directed graph
$G_q=G_{\Dmc_q}$. We can thus use standard terminology from graph
theory for databases and CQs, e.g.\ saying that a database \Dmc is
connected and speaking about the tree width of \Dmc (when \Dmc
contains no binary facts, we define its tree width to be~1). There is
an exception, though. 
The \emph{tree width of a CQ} $q=\exists \ybf \, \vp(\xbf,\ybf)$ is
defined in a more liberal way, namely as the tree width of
$G_{q|_\ybf}$ (and 1 if $q|_\ybf$ contains no binary atoms).
%
%
%
For every $k \geq 0$, a \emph{CQ$_k$} is a CQ of tree width at most
$k$ and a \emph{UCQ$_k$} is a union of CQ$_k$s with the same answer
variables.


\subsection{Description Logics and Ontology-Mediated Queries}

{\bf {\em Concepts and Ontologies.}} We introduce several widely used
description logics, see~\cite{DBLP:books/daglib/0041477} for more
details. An \emph{\ELIbot-concept} is formed according to the syntax rule
$$C,D ::= A
\mid \top \mid \bot \mid C \sqcap D \mid \exists r. C \mid \exists r^-
. C.
$$
An expression $r^-$ is an \emph{inverse role} and a
\emph{role} is a role name or an inverse role. As usual, we identify
$(r^-)^-$ with~$r$.  An \emph{$\EL_\bot$-concept} is an
$\ELI_\bot$-concept with no inverse roles.

An \emph{$\mathcal{ELHI}_\bot$-ontology} is a finite set of
\emph{$\ELI_\bot$-concept inclusions} of the form $C \sqsubseteq D$, with $C,D$
\ELIbot-concepts, and \emph{role inclusions} of the form $r \sqsubseteq s$, with
$r,s$ roles. In the name $\mathcal{ELHI}_\bot$, the letter \Hmc
indicates that role inclusions are admitted, \Imc indicates that
inverse roles are admitted, and $\cdot_\bot$ indicates that the
$\bot$-concept may be used. It should thus also be clear what we mean
by an \emph{$\mathcal{EL}$-ontology}, an
\emph{$\mathcal{ELH}_\bot$-ontology}, and so on. An
\emph{$\ELpoly$-ontology} is an $\mathcal{ELH}_\bot$-ontology that
additionally admits \emph{range restrictions} $\exists r^- . \top
\sqsubseteq C$ with $r$ a role name and $C$ an $\EL_\bot$-concept.
Note that $\cdot^{dr}$ stands for domain and range restrictions, where
domain restrictions are simply \ELbot-concept inclusions of the form
$\exists r . \top \sqsubseteq C$. We assume without loss of generality, and without further notice, that the $\bot$-concept occurs only in
concept inclusions of the form $C \sqsubseteq \bot$, where $C$ does not contain $\bot$.

\smallskip

{\bf {\em Semantics.}}   The semantics of
ontologies is defined based on interpretations, relational structures
that interpret only relations of arity one and two. We choose a
presentation here that is slightly nonstandard, but equivalent to the
usual one~\cite{DBLP:books/daglib/0041477}: an \emph{interpretation}
is a finite or infinite database \Imc with \mbox{$\mn{dom}(\Imc) \neq
  \emptyset$}. Each \ELIbot-concept $C$ and role $r$ is associated
with an \emph{extension} $C^\Imc$, resp.\ $r^\Imc$, according to
Figure~\ref{fig:semantics}.
\begin{figure}[t]
 \begin{boxedminipage}[t]{\columnwidth}
 \vspace{-\abovedisplayskip}
   \begin{center}
$$
\begin{array}{r@{\,}c@{\,}l@{}c@{}r@{\,}c@{\,}l}
  \top^\Imc &=& \mn{dom}(\Imc) &&
  \bot^\Imc &=& \emptyset \\[1mm]
  A^\Imc &=& \{ d \mid A(d) \in \Imc \} &&
  r^\Imc &=& \{ (d,e) \mid r(d,e) \in \Imc \} \\[1mm]
  (C \sqcap D)^\Imc &=& C^\Imc \cap D^\Imc  &&
  (r^-)^\Imc &=& \{ (e,d) \mid r(d,e) \in \Imc \} \\[1mm]
  (\exists r . C)^\Imc &=& \multicolumn{5}{@{}l}{\{ d \in \mn{dom}(\Imc) \mid \exists e:
  (d,e) \in r^\Imc \wedge e \in C^\Imc \}}
\end{array}
$$
   \end{center}
 \end{boxedminipage}
  \caption{Semantics of $\ELIbot$-concepts}
  \label{fig:semantics}
\end{figure}

An interpretation \Imc \emph{satisfies} a concept inclusion $C
\sqsubseteq D$ if $C^\Imc \subseteq D^\Imc$, and a role inclusion $r
\sqsubseteq s$ if $r^\Imc \subseteq s^\Imc$. It is a \emph{model} of
an $\mathcal{ELHI}_\bot$-ontology \Omc if it satisfies all the
inclusions in~\Omc, and of a database \Dmc if $\Dmc \subseteq \Imc$. A
database \Dmc is \emph{consistent with} \Omc if \Dmc and \Omc have a
common model.

Note that standard reasoning tasks, e.g., the consistency
of a given database with a given ontology, are in \PTime in \ELpoly and
\ExpTime-complete between $\ELI_\bot$ and
\ELHIbot~\cite{DBLP:books/daglib/0041477}. Note also that all the DLs defined above can be translated into (two-variable guarded) first-order logic in a standard way~\cite{DBLP:books/daglib/0041477}.


\smallskip

{\bf {\em Ontology-Mediated Queries.}}
An \emph{ontology-mediated query (OMQ)} takes the form
$Q=(\Omc,\Sbf,q)$ with \Omc an ontology,
$\Sbf$ 
a schema (which indicates that $Q$ will be evaluated over $\Sbf$-databases), and $q$ a query. The \emph{arity} of $Q$ is the
arity of~$q$.  
We write $Q(\xbf)$ to emphasize that the answer variables of $q$
are~\xbf.  When $\Sbf = \NC \cup \NR$, then we denote it
with $\Sbf_{\mn{full}}$ and speak of the \emph{full schema}.
It makes perfect sense to use a non-full schema \Sbf while
referring to concept and role names from outside \Sbf in both the
ontology~\Omc and query $q$. In fact, enriching the schema
with additional symbols is one main application of ontologies in
querying \cite{jair-data-schema}. This is similar to the distinction between
extensional and intensional relations in Datalog \cite{AHV}.

Consider an OMQ $Q(\xbf)$ and an $\Sbf$-database \Dmc. A tuple $\abf
\in \mn{dom}(\Dmc)^{|\xbf|}$ is an \emph{answer} to $Q$ on \Dmc,
written $\Dmc \models Q(\abf)$, if $\Imc \models q(\abf)$ for all
models \Imc of \Dmc and~\Omc. We write $Q(\Dmc)$ for the set of
answers to $Q$ on \Dmc.
%
%

An OMQ $Q=(\Omc,\Sbf,q)$ is \emph{empty} if, for all
$\Sbf$-databases~\Dmc consistent with \Omc, there is no answer to $Q$
on \Dmc, i.e., $Q(\Dmc)$ is empty.  Let $Q_1,Q_2$ be OMQs,
$Q_i=(\Omc_i,\Sbf,q_i)$ for $i \in \{1,2\}$.  Then $Q_1$ is
\emph{contained} in $Q_2$, written $Q_1 \subseteq Q_2$, if $Q_1(\Dmc)
\subseteq Q_2(\Dmc)$ for all $\Sbf$-databases \Dmc. Further, $Q_1$ and
$Q_2$ are \emph{equivalent}, written $Q_1 \equiv Q_2$, if $Q_1
\subseteq Q_2$ and $Q_2 \subseteq Q_1$. 
We use $(\Lmc,\Qmc)$ to refer to the \emph{OMQ language} in which the
ontology is formulated in \Lmc and where the actual queries are from~\Qmc,
e.g., $(\ELbot,\text{CQ})$ and $(\ELIHbot,\text{UCQ})$.
As usual, we write $|O|$ for the \emph{size} of a syntactic object $O$
such as an OMQ, an ontology, or a conjunctive query, that is, the
number of symbols needed to write $O$ where concept names, role names,
variables names, and the like count as one.

\smallskip

{\bf {\em The Chase.}} The chase is a widely used tool in database
theory that allows us, whenever a database is consistent with an
$\mathcal{ELHI}_\bot$-ontology \Omc, to construct a {\em universal
  model} of $\Dmc$ and \Omc that enjoys many good properties; cf.,
\cite{beeri-chase,AHV}.


Let \Omc be an $\mathcal{ELHI}_\bot$-ontology.
Intuitively,
{\em the chase of  $\Dmc$ with respect to \Omc},
denoted $\mn{ch}_\Omc(\Dmc)$,
is the potentially infinite interpretation \Imc that is obtained in
the limit of recursively applying the following two rules on
$\Dmc$, based on the inclusions in $\Omc$:
\begin{enumerate}

\item if $a \in C^\Imc$, $C \sqsubseteq D \in \Omc$, and $D \neq
  \bot$, then add $D(a)$~to~\Imc;

\item if $(a,b) \in r^\Imc$ and $r \sqsubseteq s \in \Omc$, then add
  $s(a,b)$ to \Imc.

\end{enumerate}
In Rule~1, `add $D(a)$ to \Imc' means to add to \Imc a finite
tree-shaped database that represents the \ELI-concept $D$, identifying
its root with~$a$. For example, the concept $A \sqcap \exists r . (B
\sqcap \exists s . \top)$ corresponds to the database $\{A(a), r(a,b),
B(b), s(b,c) \}$. Our chase is oblivious, a formal definition can be
found in the appendix.
We sometimes apply the chase directly to a CQ~$q$, implicitly meaning
its application to the database $\Dmc_q$.
The following lemma summarizes the main properties of the chase.


%
%
%
%
%
%
%
\begin{lemma}
\label{lem:chaseprop}
  Let \Dmc be a database and $Q=(\Omc,\Sbf,q)$ an OMQ from
  $(\ELHIbot,\text{UCQ})$. Then
  \begin{enumerate}

  \item $\Dmc$ is inconsistent with \Omc iff there is an $a \in
    \dom(\mn{ch}_\Omc(\Dmc))$ and a $C \sqsubseteq \bot \in \Omc$ such
    that $a \in C^{\mn{ch}_\Omc(\Dmc)}$;

  \item $Q(\Dmc)=q(\mn{ch}_\Omc(\Dmc))$,
 if \Dmc is consistent with \Omc;

  \item $\mn{ch}_\Omc(\Dmc) \rightarrow \Imc$ via a homomorphism that
    is the identity on $\mn{dom}(\Dmc)$, for every model \Imc of \Dmc
    and \Omc;

  \item the tree width of \Dmc and of $\mn{ch}_\Omc(\Dmc)$ are
    identical.

  \end{enumerate}
\end{lemma}

\subsection{Parameterized Complexity}

We study the evaluation problem for OMQs (defined below) both in terms of a traditional complexity analysis and
in terms of its parameterized complexity; cf., \cite{GroheBook}.
A \emph{parameterized problem} over
an alphabet $\Sigma$ is a pair $(P,\kappa)$, with $P \subseteq \Sigma^*$ a
decision problem and $\kappa$ a \emph{parameterization} of $P$, that
is, a \PTime computable function \mbox{$\kappa: \Sigma^* \rightarrow
  \mathbb{N}$}. A prime example is p-{\sc clique}, where $P$ is the set of all
pairs $(G,k)$ with $G$ an undirected graph that contains a $k$-clique
and $\kappa(G,k)=k$.

A problem $(P,\kappa)$ is
\emph{fixed-parameter tractable (fpt)} if there is a computable
function $f: \mathbb{N} \rightarrow \mathbb{N}$ and an algorithm that decides
$P$ in time $|x|^{O(1)} \cdot f(\kappa(x))$, where $x$ denotes the
input. We use \FPTclass to denote the class of all parameterized problems
that are fpt. Notice that \FPTclass corresponds to a relaxation of the
usual notion of tractability: a problem in \PTime is also in
\FPTclass, but the latter class also contains some \NPclass-complete problems.

An \emph{fpt-reduction} from a problem
$(P_1,\kappa_1)$ over $\Sigma_1$ to a problem
$(P_2,\kappa_2)$ over $\Sigma_2$ is a function $\rho:\Sigma_1^*
\rightarrow \Sigma_2^*$ such that, for some computable functions $f,g:
\mathbb{N} \rightarrow \mathbb{N}$,
\begin{enumerate}

\item $x \in P_1$ iff $\rho(x) \in P_2$, for all $x \in \Sigma_1^*$;

\item $\rho(x)$ is computable in time $|x|^{O(1)} \cdot f(\kappa_1(x))$, for $x \in \Sigma_1^*$;
\item $\kappa_2(\rho(x)) \leq g(\kappa_1(x))$, for all $x \in \Sigma_1^*$.

\end{enumerate}

An important parameterized complexity class is $\Wclass[1]
\supseteq \FPTclass$. Hardness for $\Wclass[1]$ is
defined in terms of fpt-reductions. It is believed that $\FPTclass \neq \Wclass[1]$,
the status of this problem being comparable to that of \mbox{$\PTime
  \neq \NPclass$}. Hence, if a parameterized
  problem $(P,\kappa)$ is $\Wclass[1]$-hard then $(P,\kappa)$ is not fpt unless
  $\FPTclass = \Wclass[1]$. A well-known $\Wclass[1]$-hard problem is precisely
  p-{\sc clique}~\cite{DF95}.


\section{OMQ Evaluation and Semantic Tree-likeness}
\label{sect:eval}


\subsection{OMQ Evaluation}

The main concern of this work is the \emph{evaluation problem} for
classes of OMQs \Qbf, defined as follows:

\begin{center}
\fbox{\begin{tabular}{ll}
\small{PROBLEM} : & {\sc Evaluation}$(\Qbf)$ \\
{\small INPUT} : &  An OMQ $Q=(\Omc,\Sbf,q(\xbf))$ from
\Qbf, \\ & an $\Sbf$-database \Dmc, a tuple
$\abf \in \mn{dom}(\Dmc)^{|\xbf|}$ \\
{\small QUESTION} : & Is it the case that $\abf \in Q(\Dmc)$?  \end{tabular}}
\end{center}

We are particularly interested in classifying the complexity of
{\sc Evaluation}$(\Qbf)$ for \emph{all} subsets \Qbf of an OMQ
language $(\Lmc,\Qmc)$ of interest, where we view the latter as a set
of OMQs.


  We are also interested in the parameterized version of this problem,
  with the parameter being the size $|Q|$ of the OMQ~$Q$, as customary
  in the database literature \cite{PY99}, which we call $\sf p$-{\sc
    Evaluation}$(\Qbf)$. In particular, if $\sf p$-{\sc
    Evaluation}$(\Qbf)$ is in \FPTclass, then it can be solved in time
$|\Dmc|^{O(1)} \cdot f(|Q|),$
for a computable function $f : \mathbb{N} \to \mathbb{N}$.
%
In general, the evaluation problem for CQs is NP-hard, and its parameterized version $\Wclass[1]$-hard~\cite{PY99}. Therefore, the same holds for the OMQ evaluation problem.

\begin{proposition}
For any of the DLs $\Lmc$ introduced above,
\begin{enumerate}
\item {\sc Evaluation}$(\Lmc,\text{{\em CQ}})$ is {\em NP}-hard;
\item $\sf p$-{\sc Evaluation}$(\Lmc,\text{{\em CQ}})$ is $\Wclass[1]$-hard.
\end{enumerate}
The above hold even when the ontology is empty.
\end{proposition}


On the other hand, 
CQ evaluation is tractable if restricted to CQs of tree width bounded
by $k$, for any $k$. As established by Bienvenu et al., this positive
behavior extends to OMQ evaluation in
$(\mathcal{ELH},\text{CQ}_k)$~\cite{DBLP:conf/ijcai/BienvenuOSX13}, and
it is not hard to extend their result to $(\ELpoly,\text{UCQ}_k)$. We
refrain
from giving details.

\begin{proposition} 
 {\sc Evaluation}$(\ELpoly,\text{{\em UCQ}}_k)$ is in \PTime for each fixed $k \geq 1$.
\end{proposition}

Adding inverse roles, however, destroys this property. In fact,
evaluation is \ExpTime-complete already in $(\ELI,\text{CQ})$, with
the lower bound being a consequence of the fact that the subsumption
problem in \ELI is \ExpTime-hard
\cite{DBLP:conf/owled/BaaderLB08}. Even with inverse roles, however,
evaluating OMQs in which the actual queries are of bounded tree width
is still fixed-parameter tractable.

\begin{restatable}{proposition}{profptjelia}
\label{pro:fpt-jelia}
$\sf p$-{\sc Evaluation}$(\ELHIbot,\text{\em UCQ}_k)$ is in \FPTclass, for any $k \geq 1$, with single exponential running time in the parameter.
\end{restatable}


\subsection{Semantic Tree-likeness for OMQs}

Recall that CQs $q$ and $q'$ over schema \Sbf are {\em equivalent} if
$q(\Dmc) = q'(\Dmc)$, for every $\Sbf$-database $\Dmc$.  Grohe's
Theorem establishes that, under the assumption $\FPTclass \neq \Wclass[1]$,
the classes of CQs that can be evaluated in \PTime over
$\Sbf$-databases are precisely those of bounded tree width {\em modulo
  equivalence}. Also, fixed-parameter tractability does not add
anything to standard tractability in this scenario.

\begin{theorem}[Grohe's Theorem \cite{DBLP:journals/jacm/Grohe07}] \label{thm:grohe}
Let $\Qbf$ be a recursively enumerable class of CQs over a schema $\Sbf$.
The following are equivalent, assuming $\FPTclass \neq \Wclass[1]$:
\begin{itemize}
\item the evaluation problem for CQs in $\Qbf$ is in \PTime;
\item the evaluation problem for CQs in $\Qbf$ is in \FPTclass;
\item  there is a $k \geq 1$ such that every $q \in \Qbf$ is equivalent to a CQ $q'$ in {\em CQ}$_k$.
\end{itemize}
\end{theorem}

Interestingly, the notion that characterizes tractability in this case, namely, being of bounded tree width modulo equivalence,
is decidable. Recall that a {\em retract} of a CQ $q$ is a homomorphic image $q'$ of $q$ that is also equivalent to $q$, and  a {\em core} of $q$ is a maximum retract of it, i.e., a retract that admits no further retractions \cite{HellNesetrilBook}.
It can be proved that a CQ $q$ is equivalent to a CQ $q'$ in CQ$_k$, for $k \geq 1$, iff the core of $q$ is in CQ$_k$. This problem is NP-complete, for each $k \geq 1$ \cite{DKV02}.
%
There is also a natural generalization of this characterization and of Theorem \ref{thm:grohe}
to the class of UCQs \cite{DBLP:journals/tocl/Chen14}.

At this point, it is natural to ask whether it is possible to obtain a
characterization of the classes of OMQs that can be efficiently
evaluated, in the style of Theorem \ref{thm:grohe} and, in particular,
whether a suitably defined notion of ``being equivalent to a query of
small tree width'' for OMQs exhausts tractability or \FPTclass for
OMQ evaluation, as is the case in Grohe's Theorem.
%
 The following definition introduces such a notion.
 Notice that equivalence is applied no longer on the level of the (U)CQ,
 but to the whole OMQ.

\begin{definition}[\textbf{UCQ$_k$-equivalence}]
Let $\Lmc$ be one of the DLs introduced above.
An OMQ $Q=(\Omc,\Sbf,q)$ from $(\Lmc,\text{UCQ})$ is \emph{UCQ$_k$-equivalent} if there exists an OMQ $Q'=(\Omc',\Sbf,q')$ from $(\Lmc,\text{UCQ}_k)$ such that $Q \equiv Q'$. If even $\Omc=\Omc'$, then we say that $Q$ is \emph{UCQ$_k$-equivalent while preserving the ontology}.
  \end{definition}

  Likewise, we define \emph{CQ$_k$-equivalence} and
  \emph{CQ$_k$-equivalence while preserving the ontology}.  In
  informal contexts, we may refer to (U)CQ$_k$-equivalence as
  \emph{semantic tree-likeness}.
We denote by $(\Lmc,\Qmc)^\equiv_{\text{\Qmc}'_k}$, where $\Qmc,\Qmc' \in \{ \text{CQ}, \text{UCQ} \}$, the class of OMQs from $(\Lmc,\Qmc)$ that are $\Qmc'_k$-equivalent. For example, $(\ELHIbot,\text{CQ})^\equiv_{\text{UCQ}_k}$ is the restriction of $(\ELHIbot,\text{CQ})$ to OMQs that are equivalent to an OMQ from $(\ELHIbot,\text{UCQ}_k)$.

    \begin{figure}[t]
    \begin{boxedminipage}[t]{\columnwidth}
      \begin{center}
\begin{tabular}{ccc}
 \begin{tikzpicture}[auto, scale=0.7]
 \GraphInit
 \footnotesize
  \node(qp) at (-0.6, 1.5) {$q:$};

 \node(x1) [active] at (1.5, 1.5) {$x_1$};
 \node(x1l)  at   (2.2, 1.5) {$A_1$};

 \node(x2) [active] at (0, 0) {$x_2$};
  \node(x2l)  at (-0.6, 0) {$A_2$};

 \node(x3) [active] at (1.5, -1.5) {$x_3$};
 \node(x3l)  at (2.2, -1.6) {$A_3$};

 \node(x4) [active] at (3, 0) {$x_4$};
 \node(x4l)  at (3.6, 0) {$A_4$};

 \draw [-latex] (x2) -- (x1) node[midway, left] {$r$};
 \draw [-latex] (x2) -- (x3) node[midway, left] {$r$};
  \draw [-latex] (x4) -- (x3) node[midway, right] {$r$};
 \draw [-latex] (x4) -- (x1) node[midway, right] {$r$};

 \end{tikzpicture}
&

 &

 \begin{tikzpicture}[auto, scale=0.7]
 \GraphInit
 \footnotesize

 \node(qp) at (-0.5, 1.5) {$q':$};
 \node(x1) [active] at (1.5, 1.5) {$x_1$};
 \node(x1l)  at   (2.2, 1.5) {$A_1$};

 \node(x2) [active] at (0, 0) {$x_2$};
  \node(x2l)  at (-0.6, 0) {$A_2$};

 \node(x3) [active] at (1.5, -1.5) {$x_3$};
 \node(x3l)  at (2.2, -1.6) {$A_3$};

 \node(x4) [active] at (3, 0) {$x_4$};
 \node(x4l)  at (3.6, 0) {$A_4$};

 \draw [-latex] (x2) -- (x1) node[midway, left] {$r$};
 \draw [-latex] (x2) -- (x3) node[midway, left] {$r$};
  \draw [-latex] (x4) -- (x1) node[midway, right] {$r$};

 \end{tikzpicture}

\end{tabular}

      \end{center}
\vspace*{-4mm}
    \end{boxedminipage}
    \caption{CQs for Example~\ref{ex:first}}
    \label{fig:CQktwo}
  \end{figure}
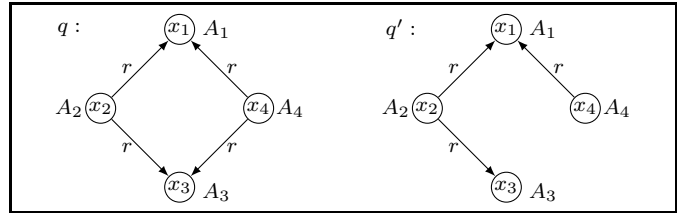

\begin{example}\label{ex:first}
  {\em We first illustrate that the ontology can have an impact on tree width.
  To this end, consider the OMQ $Q_1=(\Omc_1,\Sbf_{\mn{full}},q)$ from
  $(\EL,\text{CQ})$ given by
  $$
  \begin{array}{rcl}
    \Omc_1 &=& \{ A_2 \sqsubseteq A_4 \}\\[1mm]
  q()&=&r(x_2,x_1) \wedge r(x_4,x_1) \wedge r(x_2,x_3) \wedge
  r(x_4,x_3) \, \wedge \\[1mm]
  && A_1(x_1) \wedge A_2(x_2) \wedge A_3(x_3) \wedge A_4(x_4),
  \end{array}
  $$
  see also Figure~\ref{fig:CQktwo}.
  Then $q$ is a core of tree width 2, and thus not equivalent to a CQ of tree width 1. Yet $Q$ is from
  $(\EL,\text{CQ})^\equiv_{\text{CQ}_1}$ as it is equivalent to the
  OMQ $(\Omc_1,\Sbf_{\mn{full}},q|_{\{x_1,x_2,x_3\}})$ in which the
  CQ has tree width~1.

  We next show that the schema can have an impact as well. This is in
  a sense trivial as every OMQ based on the empty schema has tree
  width~1. The following example is more interesting. Let
  $Q_2=(\Omc_2,\Sbf_{\mn{full}},q)$ where
  $$
  \begin{array}{rcr@{\;}c@{\;}lr@{\;}c@{\;}l}
    \Omc_2 &=& \{ \qquad B_1 &\sqsubseteq& A_1, & B_2 &\sqsubseteq& A_1,
    \\[1mm]
    && \exists r . B_1 &\sqsubseteq& A_4, &  B_2 &\sqsubseteq& A_3 \ \ \}.
  \end{array}
  $$
  Then it is not hard to see that $Q_2$ is not in
  $(\ELHIbot,\text{CQ})^\equiv_{\text{UCQ}_1}$. If, however, the
  concept name $A_1$ is omitted from the schema, then $Q_2$ is
  equivalent to the OMQ $(\Omc_2,\Sbf_{\mn{full}} \setminus \{ A_1 \},q')$ where
  $$
  \begin{array}{rcl}
  q'() = r(x_2,x_1) \wedge r(x_4,x_1) \wedge r(x_4,x_3) \, \wedge \\[1mm]
A_1(x_1) \wedge A_2(x_2) \wedge A_3(x_3) \wedge A_4(x_4)
  \end{array}
  $$
  and thus in $(\EL,\text{CQ})^\equiv_{\text{CQ}_1}$. To see 
  this, take a homomorphism $h$ from $q'$ to
  $\Imc=\mn{ch}_{\Omc_2}(\Dmc)$ for any $\Sbf_{\mn{full}} \setminus \{ A_1
  \}$-database \Dmc. Then $h(x_1) \in B_1^\Imc$ or $h(x_1) \in
  B_2^\Imc$. In the former case, we obtain from $h$ a homomorphism
  from
  $q$ to \Imc by setting $h(x_4)=h(x_2)$, in the latter case we
  set $h(x_3)=h(x_1)$.
} \qed
\end{example}
In general, CQ$_k$-equivalence and
UCQ$_k$-equivalence do not coincide, i.e., sometimes it is
possible to rewrite into a disjunction of tree-like CQs, but not
into a single one.
\begin{restatable}{proposition}{propUCQnoCQ}
\label{prop:UCQnoCQ}
In $(\mathcal{ELI},\text{CQ})$, the notions of CQ$_1$-equivalence
while preserving the ontology and UCQ$_1$-equivalence while preserving
the ontology do not coincide.
\end{restatable}
On the other hand, CQ$_k$-equivalence and
UCQ$_k$-equivalence coincide in $(\mathcal{ELIH}_\bot,\text{UCQ})$,
for all $k \geq 1$, when we restrict our attention to the full
schema (see Corollary \ref{coro:full-schema-equiv} below).
%





%

\smallskip

{\bf {\em A Characterization of Semantic Tree-likeness.}}
We provide a characterization of when an OMQ $Q$ is UCQ$_k$-equivalent.
But first we need some auxiliary terminology.

A CQ $q$ is a \emph{contraction} of a CQ $q'$ if it can be obtained
from $q'$ by identifying variables. When an answer variable $x$ is
identified with a non-answer variable~$y$, the resulting variable is
$x$; the identification of two answer variables is not allowed.
Let $Q=(\Omc,\Sbf,q) \in (\ELHIbot,\text{UCQ})$ and $k
\geq 1$.  The \emph{UCQ$_k$-approximation of} $Q$ is the OMQ
$Q_a=(\Omc,\Sbf,q_a)$, where $q_a$ denotes the UCQ that consists of all
contractions of a CQ from $q$ of tree width at most $k$. By construction, $Q_a \subseteq Q$, and in this sense $Q_a$ is an approximation of $Q$ from below. The following result gives two central properties of $Q_a$, in particular that it is the best possible such approximation.
\begin{restatable}{theorem}{thmbestapprox}
\label{thm:bestapprox}
Let $Q$ be an OMQ from $(\ELHIbot,\text{UCQ})$, $k \geq 1$, and $Q_a$
the UCQ$_k$-approxi\-mation of $Q$. Then
\begin{enumerate}

\item $Q(\Dmc)=Q_a(\Dmc)$ for any $\Sbf$-database \Dmc of treewidth $\leq$ $k$;

\item $Q' \subseteq Q_a$ for every $Q' \in (\ELHIbot,\text{UCQ}_k)$ with $Q' \subseteq Q$.

\end{enumerate}
\end{restatable}
Let $Q=(\Omc,\Sbf,q)$.  The proof of Point~1 uses the fact that a
homomorphism from $q$ to $\mn{ch}_\Omc(\Dmc)$ gives rise to a
collapsing of $q$ whose tree width is not larger than that of \Dmc.
For Point~2, we `unravel' the input database into a database of tree
width at most $k$ and apply Point~1.


We obtain the
following key corollary; `$3 \Rightarrow 2$' and `$2 \Rightarrow 1$' are immediate, while `$1 \Rightarrow 3$' follows from Theorem~\ref{thm:bestapprox}.

%
\begin{corollary}
\label{prop:wasaclaim}
  Let $Q$ be an OMQ from $(\ELHIbot,\text{UCQ})$ and $k
  \geq 1$.  The following are equivalent:
  \begin{enumerate}

  \item $Q$ is UCQ$_k$-equivalent;

  \item $Q$ is UCQ$_k$-equivalent while preserving the ontology;

  \item $Q$ is equivalent to its UCQ$_k$-approximation.

  \end{enumerate}
\end{corollary}
In $(\ELHIbot,\text{UCQ})$, the notion of UCQ$_k$-equivalence thus
coincides with UCQ$_k$-equivalence while preserving the
ontology. Moreover, Corollary~\ref{prop:wasaclaim} implies
decidability of UCQ$_k$-equivalence since OMQ containment is decidable
in the OMQ languages considered in this paper
\cite{DBLP:conf/ijcai/Bienvenu0LW16}. This is further elaborated in
Section~\ref{sect:deciding}.

\smallskip

{\bf {\em Full Schema.}}\footnote{There were some issues in this
  section in the conference version
  \cite{DBLP:conf/lics/BarceloFLP19}
  of this paper. We present here a reworked version.}
  We now study the case of OMQs based on the
full schema and on CQs, which admits constructions that are closer to
the case without ontologies. Recall that a CQ is equivalent to a
CQ$_k$ iff its core has tree width at most~$k$. It is illustrated by
Example~\ref{ex:first}, however, that we cannot simply replace CQs in
OMQs with their core to attain minimal treewidth as there are OMQs
$Q=(\Omc,\Sbf_{\mn{full}},q) \in (\EL,\text{CQ})^\equiv_{\text{CQ}_1}$
such that the core of $q$ has treewidth excceding 1. We identify a
stronger operation than taking the core that serves this purpose.

Observe that cores can be defined as maximum retracts and that every
retract is a contraction while the converse is clearly false. To
obtain OMQs of minimal treewidth, we use contractions of the actual
query that (need not be retracts and) preserve equivalence on the
level of OMQs. An OMQ $Q'=(\Omc,\Sbf_{\mn{full}},q')$ is a
\emph{maximum contraction} of an OMQ
$Q=(\Omc,\Sbf_{\mn{full}},q) \in (\ELHIbot,\text{CQ})$ if
$Q \equiv Q'$, $q'$ is a contraction of $q$, and
$Q \not\equiv (\Omc,\Sbf_{\mn{full}},q'')$ for any proper contraction
$q''$ of $q'$. In contrast to cores, the CQs in maximum contractions
are not unique up to isomorphism.
\begin{example}
  Let $Q=(\Omc,\Sbf_{\mn{full}},q)$ with
  $$
  \begin{array}{rcl}
    \Omc &= \{ & A_1 \sqsubseteq \exists r^- . (A_2 \sqcap A_4 \sqcap
                 \exists r . A_3) \\
         && A_1 \sqsubseteq A_2 \sqcap \exists r . (A_1 \sqcap A_3 \sqcap
            \exists r^- . A_4) \quad \}
  \end{array}
  $$
  and $q$ is the Boolean CQ from Figure~\ref{fig:CQktwo}. Then $Q_i
  =(\Omc,\Sbf_{\mn{full}},q_i)$ is a maximum rewriting of $Q$, for
  each
  $i \in \{1,2\}$ and where $q_1$ is obtained from $q$ by identifying
  $x_2$ and $x_4$ and $q_2$ is obtained from $q$ by identifying
  $x_1$ and $x_3$.
\end{example}
We now establish the main property of maximum contractions:
an OMQ from $(\ELHIbot,\text{CQ})$ based on the full schema is
UCQ$_k$-equivalent iff some or all of its maximum contractions (which is
equivalent) fall into $(\ELHIbot,\text{CQ}_k)$. In this
sense, rewritings behave like a core for CQs without an ontology.
\begin{restatable}{theorem}{thmfulldatabase}
\label{thm:fulldatabase}
Let $Q=(\Omc,\Sbf_{\mn{full}},q)$ be a non-empty OMQ from
$(\ELHIbot,\text{CQ})$ and $k \geq 1$.
 The following are equivalent:
\begin{enumerate}
\item $Q$ is UCQ$_k$-equivalent;
\item some maximum contraction of $Q$ falls within $(\ELHIbot,\text{CQ}_k)$;
\item all maximum contractions of $Q$ fall within $(\ELHIbot,\text{CQ}_k)$.
\end{enumerate}
\end{restatable}
The interesting part of the proof is `$1 \Rightarrow 3$'. It works
by showing that if $Q'=(\Omc,\Sbf,q') \in (\ELHIbot,\text{UCQ}_k)$
is equivalent to $Q$ and $Q_c=(\Omc,\Sbf,q_c)$ is a maximum
contraction
of $Q$, then there is an injective homomorphism from $q_c$ to
$\mn{ch}_\Omc(\Dmc_{p})$ for some CQ $p$ in the UCQ $q'$, and
thus the tree width of $q_c$ is bounded by that of $q'$.
We obtain the following corollary.
%
%
%
\begin{corollary} \label{coro:full-schema-equiv}
In $(\ELHIbot,\text{CQ})$  based on the full schema, CQ$_k$-equivalence and UCQ$_k$-equivalence coincide, for $k \geq 1$.
\end{corollary}
An important property of cores of CQs that underlies the \PTime upper
bound in Theorem~\ref{thm:grohe} is that each core $q_c$ of a CQ $q$
is a subquery of $q$ in the sense that when we view both $q_c$ and $q$
as a hypergraph, then $q_c$ is an induced subhypergraph of $q$. In
contrast, the CQs from maximum contractions are not subqueries of the
CQ in the original OMQ. This is problematic for establishing our
\PTime upper bound later on and thus we go one step further.

Let $Q=(\Omc,\Sbf_{\mn{full}},q) \in (\ELHIbot,\text{CQ})$.  A
\emph{rewriting} of $Q$ is any OMQ
$\widehat Q=(\Omc,\Sbf_{\mn{full}},\widehat q)$ that can be obtained
as follows:
\begin{enumerate}

\item choose a maximum contraction $Q'=(\Omc,\Sbf_{\mn{full}},q')$ of $Q$;

\item choose a homomorphism $h$ from $q'$ to $\mn{ch}_\Omc(\Dmc_q)$;

\item take as $\widehat q$ the restriction of $q$ to the range of $h$;

\item let $V$ be the range of $h$ extended with all 
  $x \in \mn{var}(q)$ such that some fresh constant in the subdatabase
  of tree width 1 that the chase has generated below~$x$ is in the
  range of $h$;

\item for each $C \sqsubseteq D \in \Omc$ and each variable $x \in V$
  such that $x \in C^{\mn{ch}_\Omc(\Dmc_q)}$, let $q_C$ be $C$ viewed
  as a CQ using fresh variables and add $q_C$ to $\widehat q$,
  identifying $x$ with the root of~$q_C$.
 
\end{enumerate}
We say that the rewriting is \emph{based} on maximum contraction~$Q'$.
The following lemma states the main properties of rewritings.
\begin{restatable}{lemma}{lemrewritings}
\label{lem:rewritings}
  Let $Q=(\Omc,\Sbf_{\mn{full}},q)$ be a non-empty OMQ from $(\ELHIbot,\text{CQ})$ and let
  $\widehat Q=(\Omc,\Sbf_{\mn{full}},\widehat q)$ be a rewriting of
  $Q$ based on maximum contraction $Q'=(\Omc,\Sbf_{\mn{full}},q')$.
  Then
  \begin{enumerate}

  \item $Q \equiv \widehat Q$ and

  \item the treewidth of $\widehat q$ does not exceed that of $q'$.

  \end{enumerate}
\end{restatable}
It follows from Lemma~\ref{lem:rewritings} that
Theorem~\ref{thm:fulldatabase} still holds when maximum contractions
are replaced with rewritings.

\section{Fixed-Parameter Tractability}
\label{sect:fpt}

The aim of this section is to establish the following theorem.
\begin{theorem}
\label{thm:main2}
  For any recursively enumerable class of OMQs $\Qbf \subseteq
  (\ELHIbot,\text{UCQ})$, the following are
  equivalent, unless $\FPTclass=\Wclass[1]$:
  \begin{enumerate}
  \item $\sf p$-{\sc Evaluation}$(\Qbf)$ is in \FPTclass;
  \item $\Qbf \subseteq (\ELHIbot,\text{UCQ})^{\equiv}_{\text{UCQ}_k}$
    for some $k \geq 1$.
  \end{enumerate}
  If either statement is false, $\sf p$-{\sc Evaluation}$(\Qbf)$ is $\Wclass[1]$-hard.
\end{theorem}
%
%
We remark that Theorem~\ref{thm:main2} also covers OMQs where the
ontology is formulated in DL-Lite$^\Rmc_{\mn{horn}}$, introduced in
Section~\ref{sect:deciding}.
Below, we state the two directions of Theorem~\ref{thm:main2} as separate
theorems, starting with the much simpler `$2 \Rightarrow 1$' direction.
\begin{restatable}{theorem}{thmelifpt}
\label{thm:elifpt}
$\sf p$-{\sc Evaluation}$((\ELHIbot,\text{UCQ})^\equiv_{\text{UCQ}_k})$ is in \FPTclass, for any $k \geq 1$,
with single exponential running time in the parameter.
\end{restatable}
The above follows from Corollary~\ref{prop:wasaclaim}, which states that an OMQ from $(\ELHIbot,\text{UCQ})^\equiv_{\text{UCQ}_k}$ is equivalent to its UCQ$_k$-approximation $Q_a \in (\ELHIbot,\text{UCQ}_k)$, and Proposition~\ref{pro:fpt-jelia}.



Now for the rather non-trivial `$1 \Rightarrow 2$' direction,
which we consider a main achievement of this
paper. 
\begin{restatable}{theorem}{thmgrohemainlower}
\label{thm:grohemainlower}
Let $\Qbf \subseteq (\ELHIbot,\text{UCQ})$ be a recursively enumerable
class of OMQs such that, for any $k \geq 1$, $\Qbf \not\subseteq
(\ELHIbot,\text{UCQ})^\equiv_{\text{UCQ}_k}$. Then $\sf p$-{\sc Evaluation}$(\Qbf)$ is $\Wclass[1]$-hard.
\end{restatable}
%

%

As stated in Theorem \ref{thm:grohe}, Grohe established a
characterization of those classes of Boolean CQs that can be evaluated
in \PTime combined complexity \cite{DBLP:journals/jacm/Grohe07}, a
special case of Theorem~\ref{thm:main2} where ontologies are empty and
schemas are full. The `lower bound part' of Grohe's proof is by
fpt-reduction from p-CLIQUE, a $\Wclass[1]$-hard problem. We prove
Theorem~\ref{thm:grohemainlower} by following the same approach,
carefully reusing a central construction
from~\cite{DBLP:journals/jacm/Grohe07}. For $k,\ell \geq 1$, the
\emph{$k \times \ell$-grid} is the graph with vertex set $\{ (i,j)
\mid 1\leq i \leq k \text{ and } 1 \leq j \leq \ell \}$ and an edge
between $(i,j)$ and $(i',j')$ iff $|i-i'|+|j-j'|=1$. A \emph{minor} of
an undirected graph is defined in the usual way, see,
e.g.,~\cite{DBLP:journals/jacm/Grohe07}. When $k$ is understood
from the context, we use $K$ to denote ${k \choose 2}$.
The following is what we use
from Grohe's proof.
\begin{theorem}[Grohe]
\label{thm:grohetechnew}
Given an undirected graph $G=(V, E)$, a $k > 0$, and a
connected 
$\Sbf$-database~$\Dmc$ such that $G_\Dmc$ 
contains the $k \times K$-grid as a minor, 
one can construct in time $f(k)\cdot\mn{poly}(|G|,|\Dmc|)$
an $\Sbf$-database $\Dmc_G$ such that:
  \begin{enumerate}

  \item there is a surjective homomorphism $h_0$ from $\Dmc_G$ to
    $\Dmc$ such that for every edge $\{a, b\}$ in the Gaifman graph of
    $\Dmc_G$:
    $s(a,b) \in \Dmc_G$ iff $s(h_0(a),h_0(b)) \in \Dmc$ for all roles~$s$;

  \item $G$ contains a $k$-clique iff there is a homomorphism $h$ from
    $\Dmc$ to $\Dmc_G$ such that $h_0(h(\cdot))$ is the identity.

  \end{enumerate}
\end{theorem}
A careful analysis of \cite{DBLP:journals/jacm/Grohe07} reveals that
the proof given there establishes Theorem~\ref{thm:grohetechnew}
without the `such that' part of Condition~(1), which we need to deal
with role inclusions. That part, however, can
be attained by first suitably switching from the original schema to a
schema that is based on sets of relations from the original one, then
applying Grohe, and then switching back.

%
To avoid overly messy notation, we first prove
Theorem~\ref{thm:grohemainlower} for the case where $\Qbf \subseteq
(\ELHIbot,\text{CQ})$ consists only of Boolean OMQs. In the appendix, we
explain how to extend the proof to the non-Boolean case, and from CQs
to UCQs.

For the fpt-reduction from p-CLIQUE, assume that $G$ is an undirected
graph and $k \geq 1$ a clique size, given as an input to the
reduction. 
By Robertson and Seymour's
Excluded Grid Theorem, there is an $\ell$ such that every graph of
tree width exceeding $\ell$ contains the $k \times K$-grid as a
minor~\cite{RS86}. By our assumption on $\Qbf$, we find an OMQ
$Q=(\Omc,\Sbf,q)$ from \Qbf such that
$Q \notin (\ELHIbot,\text{CQ})^\equiv_{\text{UCQ}_\ell}$. Since the
choice of $Q$ is independent of $G$ and since it is decidable whether
an OMQ from $(\ELHIbot,\text{CQ})$ belongs to
$(\ELHIbot,\text{CQ})^\equiv_{\text{UCQ}_\ell}$ by
Theorem~\ref{thm:decupper} in Section \ref{sect:deciding}, we can simply enumerate the OMQs
from \Qbf until we find $Q$.

Let $Q_a$ be the UCQ$_\ell$-approximation of~$Q$. Note that any
\Sbf-database \Dmc with $\Dmc \models Q$ and $\Dmc \not\models Q_a$
must be of tree width exceeding $\ell$ since $Q_a$ is equivalent
to $Q$ on $\Sbf$-databases of tree width at most~$\ell$ by
Theorem~\ref{thm:bestapprox}. Thus \Dmc
contains the $k \times K$-grid as a minor, which enables the
application of Theorem~\ref{thm:grohetechnew}. We could find such \Dmc
by brute force enumeration 
and then hope to show that $\Dmc_G \models Q$ iff there is a
homomorphism $h$ from $\Dmc$ to $\Dmc_G$ such that $h_0(h(\cdot))$ is
the identity and thus, by Theorem~\ref{thm:grohetechnew}, iff $G$
contains a $k$-clique. This would in fact be easy if \Omc was empty
and \Sbf was full since then we could assume \Dmc to be isomorphic to
$q$, but neither of this is guaranteed. As we show in the following,
however, it is possible to construct \Dmc in a very careful way so
that its relational structure is sufficiently tightly linked to $q$
to enable the reduction.

\subsection{The Construction of the Database}

Injective homomorphisms are an important ingredient to identifying
\Dmc since they link a CQ much closer to a database than non-injective
homomorphisms. In fact, a main idea is to construct \Dmc such that for
some contraction $q_c$ of $q$: if $q_c$ maps to $\mn{ch}_\Omc(\Dmc_G)$
but only in terms of \emph{injective} homomorphisms, then the
same is true for $q_c$ and $\mn{ch}_\Omc(\Dmc)$.

For a database \Dmc and a Boolean CQ~$p$, 
we write $\Dmc \models^{io} p$ if $\Dmc \models p$ and all
homomorphisms $h$ from $p$ to $\Dmc$ are injective. 
 Here, `io' stands for `injectively only'.
We start with a simple observation.
\begin{restatable}{lemma}{lemiolem}
\label{lem:iolem}
If $\Dmc
\models p$, for \Dmc  a potentially infinite database and $p$ a CQ, then $\Dmc
\models^{io} p_c$ for some contraction $p_c$ of $p$.
\end{restatable}
Let $q_1,\dots,q_n$ be the maximal connected components of $q$. For
$1 \leq i \leq n$, let $Q_i=(\Omc,\Sbf,q_i)$. We can assume
w.l.o.g.\ that $Q_i \not\subseteq Q_j$ for all $i \neq j$ because
if this is not the case, then we can drop the component $q_j$ from $q$
and the resulting OMQ is equivalent to $Q$. Since
$Q \notin (\ELHIbot,\text{CQ})^\equiv_{\text{UCQ}_\ell}$, it is clear
that $Q_w \notin (\ELHIbot,\text{CQ})^\equiv_{\text{UCQ}_\ell}$ for
some $w$, with $1\leq i\leq n$. 
%
From now on, we use
$Q_a$ to denote the UCQ$_\ell$-approximation of $Q_w$ (rather than
of $Q$), which we also compute as part of the reduction.

To achieve the desiderata for \Dmc mentioned above, we next identify
an \Sbf-database \Dmc such that $\Dmc \models Q_w$ and $\Dmc
\not\models Q_a$ and, additionally, if $\mn{ch}_\Omc(\Dmc)
\models^{io} q_c$ for a contraction $q_c$ of $q_w$, then there is no
`less constrained' contraction that does the same, even in databases
that homomorphically map to \Dmc. Here a contraction $q_c$ of $q_w$ is
\emph{less constrained} than a contraction $q'_c$ of $q_w$, written
$q_c \prec q_c'$, when $q'_c$ is a proper contraction of $q_c$. We
write $q_c \preceq q'_c$ when $q_c \prec q'_c$ or $q_c = q'_c$.

\begin{restatable}{lemma}{lemiomin}
\label{lem:iomin}
There is an $\Sbf$-database \Dmc such that the following conditions
are satisfied:
  \begin{enumerate}

  \item $\Dmc \models Q_w$ and $\Dmc \not \models Q_a$;


  \item if $\mn{ch}_\Omc(\Dmc) \models^{io} q_c$, for $q_c$ a contraction
    of $q_w$, then there is no $\Sbf$-database
    $\Dmc'$ and  contraction $q'_c$ of $q_w$ such that $\Dmc'
    \rightarrow \Dmc$, $\mn{ch}_\Omc(\Dmc') \models^{io}
    q'_c$, and \mbox{$q'_c \prec q_c$}.

  \end{enumerate}
\end{restatable}
\begin{proof}
Since $Q \notin
  (\EL,\text{CQ})^\equiv_{\text{UCQ}_\ell}$ and $Q_a \subseteq Q_w$, we
  find an $\Sbf$-database $\Dmc_0$  such that $\Dmc_0 \models Q_w$, but $\Dmc_0 \not\models
  Q_a$.
%
  The database $\Dmc_0$  does not necessarily satisfy
  Condition~2, though. We thus replace it by a more suitable database
  \Dmc, which we identify in an iterative
  process. Start with $\Dmc = \Dmc_0$ and as
  long as there are an $\Sbf$-database $\Dmc'$ and contractions $q_c,q'_c$ of $q$ such that
  $\mn{ch}_\Omc(\Dmc) \models^{io} q_c$, $\Dmc' \rightarrow
  \Dmc$, $\mn{ch}_\Omc(\Dmc') \models^{io} q'_c$, and $q_c$ is
  a proper contraction of $q'_c$, replace $\Dmc$ by $\Dmc'$.

  It is clear that the resulting \Dmc satisfies
  Condition~(2). Condition~(1) is satisfied as well: we have
  $\mn{ch}_\Omc(\Dmc) \models^{io} q_c$ for some contraction
  $q_c$ of $q_w$, thus $\Dmc \models Q_w$; further, $\Dmc
  \rightarrow \Dmc_0$ and $\Dmc_0 \not \models Q_a$
  yield $\Dmc \not\models Q_a$.  We prove in the appendix that
  this iterative process terminates.
\end{proof}
The conditions in Lemma \ref{lem:iomin} are
decidable.  It can be shown that it suffices to consider
databases $\Dmc'$ of a certain `pseudo tree shape' (c.f.\
\cite{DBLP:conf/ijcai/Bienvenu0LW16}) which enables a
reduction to
satisfiability of {\em monadic second-order logic} (MSO) sentences on
trees.

\begin{restatable}{lemma}{lemdecHomPre}
\label{lem:decHomPre}
Given an $\Sbf$-database \Dmc and an OMQ $Q$ from
$(\ELHIbot,\text{CQ})$, it is decidable whether Conditions~1 and~2
from Lemma~\ref{lem:iomin} hold.
\end{restatable}
Let $\Dmc_0$ be the \Sbf-database from Lemma~\ref{lem:iomin}. Since
the properties of $\Dmc_0$ are independent of $G$, and due to
Lemma~\ref{lem:decHomPre}, we can find $\Dmc_0$ by
enumeration. However, $\Dmc_0$ is still not as required and needs to
be manipulated further to make it suitable for the reduction. We start
with some preliminaries about unravelings.


%

For each $a \in \Ind(\Dmc_0)$, let $\Dmc_0^{a}$ be the unraveling of $\Dmc_0$ into a database of tree width~1 starting at $a$,
defined in the appendix.  The proof of the following lemma
is omitted.
%
\begin{lemma}
\label{lem:unravcon}
$\Dmc_0 \models (\Omc,\Sbf,p)(a)$ iff $\Dmc_0^{a} \models
(\Omc,\Sbf,p)(a)$ for all unary CQs $p$ with $\Dmc_p$ of tree width
one.\footnote{Note that this is a stricter requirement than $p$ being
  of tree width~1 because answer variables are omitted from tree
  decompositions.}
%
%
%
%
%
\end{lemma}
%
%
Of course, $\Dmc^{a}_0$ can be infinite. By compactness, however, there is
a finite $\Dmc_a \subseteq \Dmc^{a}_0$ such that Lemma~\ref{lem:unravcon}
is satisfied for all (finitely many) $p$ that use only symbols from \Omc and $q$ and
satisfy $|p| \leq \max\{|\Omc|,|q|\}$. For brevity, we say that
$\Dmc_a$ \emph{satisfies Lemma~\ref{lem:unravcon} for all relevant
  CQs}.  We can find $\Dmc_a$ by constructing $\Dmc^{a}_0$ level by level
and deciding after each such extension whether we have found the
desired database, by checking the condition in
Lemma~\ref{lem:unravcon} for all relevant CQs.


Now for the further manipulation of $\Dmc_0$. We show that $\Dmc_0$
can be replaced with a database $\Dmc^+$ that is more closely linked
to $q_w$ than $\Dmc_0$ is.  For every $\Dmc \subseteq \Dmc_0$, let
$\Dmc^+$ denote the result of starting with~\Dmc and then disjointly
adding a copy of $\Dmc_a$, identifying the root of this copy with~$a$,
for each $a \in \mn{dom}(\Dmc)$.  For what follows, choose $\Dmc
\subseteq \Dmc_0$ minimal such that $\Dmc^+ \models Q_w$.  Note that
$\Dmc$ contains only \emph{binary facts} of the form $r(a,b)$ with $a
\neq b$, but no \emph{unary facts} of the form $A(a)$ or $r(a,a)$
since the latter can be made part of $\Dmc_a$. We can find $\Dmc$ by
considering all subsets of~$\Dmc_0$.
\begin{restatable}{lemma}{lemAprop}\label{lem:Aprop}~\\[-4mm]
  \begin{enumerate}

  \item $\Dmc^+$ satisfies
    Conditions~1 and~2 of Lemma~\ref{lem:iomin};

  \item $\Dmc$ has tree width exceeding $\ell$.


  \end{enumerate}
\end{restatable}

By Point~2 of Lemma~\ref{lem:Aprop} and choice of $\ell$, we have that $\Dmc$
contains the $k \times K$-grid as a minor. We can thus apply
Theorem~\ref{thm:grohetechnew} to $G$, $k$, and \Dmc, obtaining an
$\Sbf$-database $\Dmc_G$ and a homomorphism $h_0$ from $\Dmc_G$ to
\Dmc such that Points~1 and~2 of that theorem are satisfied. Recall
that $q_1,\dots,q_n$ are the maximal connected components of $q$,
giving rise to OMQs $Q_1,\dots,Q_n$, and that
$Q_i \not\subseteq Q_j$ for all $i \neq j$. As a consequence, for
each $i \neq w$ we can choose an $\Sbf$-database $\Dmc_i$ with
$\Dmc_i \models Q_i$ and $\Dmc_i \not\models Q_w$.  Let
\begin{enumerate}

\item $\Dmc^+_G$ be obtained by starting with $\Dmc_G$ and then
  disjointly adding, for each $a \in \mn{dom}(\Dmc_G)$, a copy of
  $\Dmc_{h_0(a)}$ identifying the root of this copy with~$a$;

\item $\Dmc^{*}_G$ be obtained by further disjointly adding
  $\Dmc_1,\dots,\Dmc_{w-1},\Dmc_{w+1},\dots,\Dmc_n$.

\end{enumerate}
The fpt reduction of p-CLIQUE consists then in computing $Q$ and
$\Dmc^{*}_G$ from $G$ and $k \geq 1$.

\subsection{Correctness of the Reduction}

We show in
the subsequent lemma that $\Dmc^{*}_G \models Q$ if and only if $G$
has a $k$-clique. For a CQ $p$, we use $nt(p)$ to denote the result of
removing all `dangling trees' from $p$, where trees might include
reflexive loops and multi-edges and `nt' stands for `no
trees'. Formally, $nt(p)$ is the maximal subset of $p$ (viewed as a
set of atoms) such that there is no articulation point $x \in
\mn{var}(p)$ that separates $nt(p)$ into components $p_1,p_2$ with
$p_2$ of tree width~1. It should be clear that $nt(p)$ is uniquely
defined when $p$ is connected and contains a non-tree part, that is,
the tree width of $p$ exceeds 1.
\begin{restatable}{lemma}{lemCliqueRed}
\label{lem:CliqueRed}
  $G$ has a $k$-clique iff $\Dmc^{*}_G \models Q$.
\end{restatable}
\begin{proof}\ The `only if' direction is easy. If $G$ has
  a $k$-clique, then $\Dmc \rightarrow \Dmc_G$ by Point~2 of
  Theorem~\ref{thm:grohetechnew}. It is straightforward to extend a
  witnessing homomorphism to one from $\Dmc^+$ to $\Dmc^+_G$, and thus
  $\Dmc^+ \rightarrow \Dmc^{+}_G$. Consequently, $\Dmc^+ \models Q_w$
  implies $\Dmc^+_G \models Q_w$. By construction of $\Dmc^{*}_G$ it holds that
  $\Dmc^{*}_G \models Q$.

 For the `if' direction, assume that
  $\Dmc^{*}_G \models Q$. By choice of the components in
  $\Dmc^{*}_G \setminus \Dmc^{+}$, this means that
  $\Dmc^{+}_G \models Q_w$.  Then
  $\mn{ch}_\Omc(\Dmc^+_G) \models q_w$ and by Lemma~\ref{lem:iolem},
  we find a contraction $q_c$ of $q_w$ such that
  $\mn{ch}_\Omc(\Dmc^+_G) \models^{io} q_c$. We have
  $\Dmc_G \rightarrow \Dmc$ via the homomorphism $h_0$ from
  Theorem~\ref{thm:grohetechnew} and it is straightforward to extend
  $h_0$ so that it yields
  $\mn{ch}_\Omc(\Dmc^+_G) \rightarrow \mn{ch}_\Omc(\Dmc^+)$. It
  follows that $\mn{ch}_\Omc(\Dmc^+) \models q_c$. Thus we find a
  contraction $q'_c$ of $q_c$ such that
  $\mn{ch}_\Omc(\Dmc^+) \models^{io} q'_c$. We must have $q_c=q'_c$
  since $\Dmc^+$ satisfies Condition~2 of Lemma~\ref{lem:iomin}, via
  Lemma~\ref{lem:Aprop}.  Let $h$ be a homomorphism from $q_c$ to
  $\mn{ch}_\Omc(\Dmc^+_G)$.  Then the composition $h_0(h(\cdot))$ is a
  homomorphism from $q_c$ to $\mn{ch}_\Omc(\Dmc^+)$. Since
  $\mn{ch}_\Omc(\Dmc^+) \models^{io} q_c$, this homomorphism must be
  injective.  Let~$g$ be its restriction to the variables in
  $nt(q_c)$.\footnote{This is uniquely defined since $q_c$ is clearly
    connected and, moreover, has tree width exceeding $\ell$
    because $\Dmc^+ \not\models Q_a$ and thus $q_c$ is not a
    CQ in the UCQ $q_a$ in $Q_a$.}


  The range of $g$ must fall into $\mn{dom}(\Dmc)$ since $g$ is
  injective: if the range of $g$ involved an element from a tree
  width~1 part of $\mn{ch}_\Omc(\Dmc^+)$, added by the transition from
  \Dmc to $\Dmc^+$ or by the chase, then because of the injectivity of
  $g$ this gives rise to an articulation point in $nt(q_c)$ that
  separates $nt(q_c)$ into two components $q_1,q_2$ with $G_{q_2}$ a
  tree, but such an articulation point does not exist. Moreover, the
  elements of $\mn{ch}_\Omc(\Dmc^+)$ that have not been added by the
  $\cdot^+$-construction or by the chase are precisely those in
  $\mn{dom}(\Dmc)$.


  Moreover, $g$ must satisfy a certain ontoness condition regarding
  the subset \Dmc of $\mn{ch}_\Omc(\Dmc^+)$. When we speak of an
  \emph{edge in \Dmc}, we mean an edge $e=\{a,b\} \subseteq
  \mn{dom}(\Dmc)$ in the Gaifman graph of \Dmc. We say that $g$
  \emph{maps an atom $r(x,y) \in nt(q_c)$ to $e$} if
  $\{g(x),g(y)\}=\{a,b\}$. It can be verified that
  \begin{itemize}

  \item[($\dagger$)] for every edge $e$ in \Dmc, there is an atom in $nt(q_c)$
    that $g$ maps to $e$.

  \end{itemize}
  Assume to the contrary that $g$ maps no atom in $nt(q_c)$
  to an edge $\{a,b\} \in \Dmc$. We show in the appendix
  that, then the database $\Dmc_1$ obtained from \Dmc by removing all
  binary facts that involve $a$ and $b$ is such that $\Dmc_1^+
  \models Q$, contradicting the choice of \Dmc.



  We are now ready to finish the proof. At this point, we know that
  $g$ is a restriction of $h_0(h(\cdot))$, that it is injective, and
  that its range is a subset of $\mn{dom}(\Dmc)$. In fact, the range
  of $g$ must be exactly $\mn{dom}(\Dmc)$, by ($\dagger$) and since \Dmc
  contains only binary facts. As a consequence, the inverse $h_0^-$ of
  $h_0$ is an injective total function from $\mn{dom}(\Dmc)$ to
  $\mn{dom}(\mn{ch}_\Omc(\Dmc_G^+))$.  We next argue that its range
  actually falls within $\mn{dom}(\Dmc_G)$, that is, it does not hit
  any tree width 1 parts of $\mn{ch}_\Omc(\Dmc_G^+)$, added by the
  $\cdot^+$-construction or by the chase. Any constant from
  $\mn{dom}(\Dmc)$ occurs in a fact of the form $r(a,b)$.  By the
  ontoness condition~($\dagger$), $g$ maps some atom $r(x,y) \in nt(q_c)$ to
  the edge $\{a,b\}$ in \Dmc. But then $\mn{ch}_\Omc(\Dmc_G^+)$ must
  contain the fact $r(h(x),h(y))$ and, moreover,
  $\{h(x),h(y)\}=\{h_0^-(a),h_0^-(b)\}$. But $r(h(x),h(y))$ cannot be
  in any of the tree width~1 parts of $\mn{ch}_\Omc(\Dmc_G^+)$: since
  $h$ is injective, this would give rise to an articulation point in
  $nt(q_c)$ that separates $nt(q_c)$ into two components $q_1,q_2$
  with $G_{q_2}$ a tree. With $\{h(x),h(y)\} = \{h_0^-(a),h_0^-(b)\}$,
  we obtain $h_0^-(a),h_0^-(b) \in \mn{dom}(\Dmc_G)$ as desired.

  Thus, $h_0^-$ is a function from $\mn{dom}(\Dmc)$ to
  $\mn{dom}(\Dmc_G)$.  We show that it is a homomorphism from \Dmc to
  $\Dmc_G$, and thus Point~2 of Theorem~\ref{thm:grohetechnew} yields
  that $G$ contains a $k$-clique, finishing the proof.
  Let $r(a,b) \in \Dmc$.  By the ontoness condition ($\dagger$), we have that $g$ maps
  some atom $s(x,y) \in nt(q_c)$ to the edge $\{a,b\}$. We have
  already argued that $\{h(x),h(y)\} = \{h_0^-(a),h_0^-(b)\} \subseteq
  \mn{dom}(\Dmc_G)$.  Since $s(h(x),h(y)) \in \mn{ch}_\Omc(\Dmc_G^+)$
  and $\{h(x),h(y)\} \subseteq \mn{dom}(\Dmc_G)$, there must be some
  fact $s'(h(x),h(y)) \in \Dmc_G$. By the `such that' part of Point~1
  of Theorem~\ref{thm:grohetechnew} and since $\{h(x),h(y)\} = 
  \{h_0^-(a),h_0^-(b)\}$, we have $r(a,b) \in \Dmc_G$.
\end{proof}

We explain in the appendix how to extend the above proof to the case
where OMQs need not be Boolean, which essentially amounts to choosing
also concrete answers along with databases, and then removing and
reading the constants from the answers at the right places in the
proof.
%
%
%
We also explain how to extend the proof from CQs to UCQs. A difficulty
lies in identifying a \emph{connected} component of some CQ in the UCQ
$q$ that can play the role of $q_w$ in the original proof, despite the
presence of the other disjuncts in $q$. We overcome this be viewing
$q$ as a disjunction of conjunctions of connected CQs and rewriting
$q$ into an equivalent conjunction of disjunctions of connected CQs.

\section{PTime Combined Complexity}
\label{sect:ptimeupper}

The aim of this section is to establish the following theorem.
\begin{theorem}
\label{thm:main1}
  For any recursively enumerable class of OMQs $\Qbf \subseteq
  (\ELpoly,\text{UCQ})$ based on the full schema, the following are
  equivalent, unless $\FPTclass=\Wclass[1]$:
  \begin{enumerate}
  \item {\sc Evaluation}$(\Qbf)$ is in \PTime combined complexity;
  \item $\sf p$-{\sc Evaluation}$(\Qbf)$ is in $\text{\rm FPT}$;
  \item $\Qbf \subseteq (\ELpoly,\text{UCQ})^{\equiv}_{\text{UCQ}_k}$
    for some $k \geq 1$.
  \end{enumerate}
  If either statement is false, $\sf p$-{\sc Evaluation}$(\Qbf)$ is $\Wclass[1]$-hard.
\end{theorem}
%
%

The `$1 \Rightarrow 2$' direction is trivial.
For showing `$2 \Rightarrow 3$', observe that $\Qbf \subseteq
(\ELHIbot,\text{UCQ})$. Thus, by Theorem~\ref{thm:grohemainlower} and the
hypothesis $\FPTclass \neq \Wclass[1]$, $\Qbf \subseteq
(\ELHIbot,\text{UCQ})^\equiv_{\text{UCQ}_k}$. Therefore, for every $Q
\in \Qbf$, there exists $Q' \in (\ELHIbot,\text{UCQ}_k)$ such that $Q
\equiv Q'$. Since, by Corollary~\ref{prop:wasaclaim},
UCQ$_k$-equivalence coincides with UCQ$_k$-equivalence while
preserving the ontology, we can assume that $Q' \in
(\ELpoly,\text{UCQ}_k)$. This implies that $Q \in
(\ELpoly,\text{UCQ})^{\equiv}_{\text{UCQ}_k}$.
It thus remains to address the `$3
\Rightarrow 1$' direction, that is, to prove the following.
\begin{restatable}{theorem}{thmpolycombcompl}
\label{thm:polycombcompl}
{\sc Evaluation}$((\ELpoly,\text{UCQ})^\equiv_{\text{UCQ}_k})$ based on
the full schema is in \PTime combined complexity, for any $k \geq 1$.
%
%
%
\end{restatable}
%
Evaluating an OMQ $(\Omc,\Sbf,q)$ from $(\ELpoly,\text{UCQ})$
is the same as evaluating every OMQ $(\Omc,\Sbf,p)$, $p$ a CQ in~$q$, and taking the union of the answer
sets. To establish Theorem~\ref{thm:polycombcompl}, it thus suffices
to prove that {\sc Evaluation}$((\ELpoly,\text{CQ})^\equiv_{\text{UCQ}_k})$
based on the full schema is in \PTime.

We do this by using a suitable form of existential pebble game. Such
games are also employed in the case of CQ evaluation over relational
databases, that is, in the special case of
Theorem~\ref{thm:polycombcompl} when the ontology is empty
\cite{DKV02,DBLP:journals/jacm/Grohe07}.  In that case, the game is
played on the CQ $q$ and the input database~\Dmc, details are given
later.  When the ontology $\Omc$ is non-empty, a natural idea is to
play the pebble game on $q$ and $\mn{ch}_\Omc(\Dmc)$ instead, which
can be shown to give the correct result. However, $\mn{ch}_\Omc(\Dmc)$
need not be finite. There is a way to compute in polynomial time a
finite representation of
$\mn{ch}_\Omc(\Dmc)$~\cite{LutzTomanWolter-IJCAI09}, but using that
representation in place of $\mn{ch}_\Omc(\Dmc)$ requires to
rewrite~$q$ in a way that might increase the tree width 
and as a consequence there is no guarantee that the
resulting game delivers the correct
result. We thus start by giving a novel characterization of answers to
OMQs from $(\ELpoly,\text{CQ})$ that is tailored towards being
verified by existential pebble games.

\subsection{Characterization of OMQ Answers}

We require several definitions
and preliminaries.
%

Let $q$ be a CQ. 
A database \Dmc 
is a \emph{ditree} if
the directed graph $(\dom(\Dmc),\{(a,b) \mid r(a,b) \in \Dmc\})$ is a
tree.  Note that multi-edges are admitted while reflexive loops are
not.  We say that $q$ \emph{is a homomorphic preimage of a
  ditree} if there is a homomorphism from $q$ to a ditree database
\Dmc. Consider the CQ $q'$ obtained from $q$ by exhaustively
identifying variables $x_1$ and $x_2$ whenever there are atoms
$r(x_1,y)$ and $s(x_2,y)$. It can be verified that $q$ is a
homomorphic preimage of a ditree if and only if $\Dmc_{q'}$ is a
ditree. This also means that it is decidable in \PTime whether a given
$q$ is a homomorphic preimage of a ditree.  If this is the case, then
$\Dmc_{q'}$ is \emph{initial} among all ditrees \Dmc that $q$ is a
homomorphic preimage of, that is, $\Dmc_{q'}$ admits a homomorphism
to any such \Dmc. We use $\mn{dtree}_{q}(x_0)$ to denote $\Dmc_{q'}$
viewed as a CQ, constants corresponding to variables, in which the
root constant is the only answer variable $x_0$ and all other
variables are quantified. If $q$ is not a homomorphic preimage of a
ditree, then $\mn{dtree}_{q}$ is undefined.

A pair of variables $x,y$ from $q$ is \emph{guarded}
if they are linked by an edge in the Gaifman graph of $\Dmc_q$.  Let
$G^q_2$ be the set of all guarded pairs of variables from $q$.  For
every $(x,y) \in G^q_2$ with $y$ quantified and for every $i \geq 0$,
define $\mn{reach}^i(x,y)$ to be the smallest set such that
\begin{enumerate}

\item $x \in \mn{reach}^0(x,y)$ and $y \in \mn{reach}^1(x,y)$;

\item if $z \in \mn{reach}^i(x,y)$, $i > 0$, and $r(z,u) \in q$, then
  $u \in \mn{reach}^{i+1}(x,y)$;

\item if $y \in \mn{reach}^{i+1}(x,y)$ and $r(z,y) \in q$, then
  $z \in \mn{reach}^{i}(x,y)$.

\end{enumerate}
Moreover, $\mn{reach}(x,y) = \bigcup_i
\mn{reach}^i(x,y)$.  A guarded pair $(x,y)$ is
\emph{$\exists$-eligible} if $q|_{\mn{reach}(x,y)}$ is a homomorphic
preimage of a ditree.  We use $\mn{dtree}_{(x,y)}$ as a shorthand
for $\mn{dtree}_{q|_{\mn{reach}(x,y)}}$.

Informally, $(x,y)$ being $\exists$-eligible means that in a
homomorphism $h$ from $q$ to $\mn{ch}_\Omc(\Dmc)$, for some database
\Dmc and some \ELpoly-ontology \Omc, atoms $r(x,y) \in q$ can `cross
the boundary' between $\Dmc$ and the part of $\mn{ch}_\Omc(\Dmc)$
generated by the chase in the sense that $h(x) \in \mn{dom}(\Dmc)$ and
$h(y)$ is mapped to a constant that was introduced by the chase. Note
that the chase generates only structures that are ditrees.

Let $Q=(\Omc,\Sbf_{\mn{full}},q)$ be an OMQ from $(\ELpoly,\text{CQ})$
and let \Dmc be an \Sbf-database that is consistent with \Omc.  We now
define the central notion underlying the announced characterization,
called $\Dmc$-labeling of $q$, which (partially) represents a
homomorphism from the CQ $q$ to $\mn{ch}_\Omc(\Dmc)$.

An \emph{$\exists$-MCC} is a subquery $p \subseteq q$ that constitutes
a maximally connected component of $q$ and contains only quantified
variables. For an \Sbf-database \Dmc, we use $\mn{ch}^-_\Omc(\Dmc)$ to
denote the restriction of $\mn{ch}_\Omc(\Dmc)$ to the constants in
$\mn{dom}(\Dmc)$.
A \emph{\Dmc-labeling} of $q$ is a function $\ell:\mn{var}(q) \rightarrow \mn{dom}(\Dmc) \cup \{\exists \} \cup (G^q_2 \times \mn{dom}(\Dmc)) \}$ such that the following conditions are satisfied:
\begin{enumerate}

\item $\ell(x) \in \mn{dom}(\Dmc)$ for every answer variable $x$;

\item the restriction of $\ell$ to the variables in $V:=\{ x \mid
  \ell(x) \in \mn{dom}(\Dmc) \}$ is a homomorphism from $q|_V$ to
  $\mn{ch}^-_\Omc(\Dmc)$;

\item if $r(x,y) \in q$ and $\ell(y) \in \mn{dom}(\Dmc)$, then
  $\ell(x) \in \mn{dom}(\Dmc)$;

\item if $(x,y) \in G^q_2$, $\ell(x) \in \mn{dom}(\Dmc)$, $\ell(y) \notin \mn{dom}(\Dmc)$, then
  \begin{enumerate}

  \item $(x,y)$ is $\exists$-eligible,

  \item $\Dmc \models (\Omc,\Sbf,\mn{dtree}_{(x,y)})(\ell(x))$, and

  \item $\ell(y)=((x',y'),\ell(x))$ where $x' \in \mn{reach}^0(x,y)$
    and $y' \in \mn{reach}^1(x,y)$;

  \end{enumerate}

\item if $r(x,y) \in q$ and $\ell(x) = ((x',y'),a)$, then
  $\ell(y)=\ell(x)$;

\item if $r(x,y) \in q$, $\ell(y) = ((x',y'),a)$, and $y \notin \mn{reach}^0(x',y')$, then
  $\ell(x)=\ell(y)$;

\item if $r(x,y) \in q$, $\ell(y) = ((x',y'),a)$, and
  $y \in \mn{reach}^0(x',y')$, then $\ell(x)=a$;

\item if $q'$ is an $\exists$-MCC of $q$ such that $\ell(x) \notin
  \mn{dom}(\Dmc)$ for every variable $x$ in $q'$, then $q'$ is a
  homomorphic preimage of a ditree and 
  $\Dmc \models (\Omc,\Sbf,\exists x_0 \, \mn{dtree}_{q'})$.

\end{enumerate}
A \Dmc-labeling of $q$ represents a homomorphism $h$ from $q$ to
$\mn{ch}_\Omc(\Dmc)$ with the following conditions. If $\ell(z)=a\in
\mn{dom}(\Dmc)$, then $h(z)=a$. If $\ell(z)={\exists}$, then $z$ is in
an $\exists$-MCC of $q$ and mapped to a constant generated by the
chase. Finally, if $\ell(z)=((x,y),a)$, then (the same is true or)
$h(x) =a \in \mn{dom}(\Dmc)$, $h(y)$ is a constant generated by the
chase, and $h(z)$ is in the tree-shaped sub-database
of $\mn{ch}_\Omc(\Dmc)$ rooted at~$h(y)$.

\begin{restatable}{lemma}{lemcharptimetwo}
\label{lem:charptime2}
For every
$\abf \in\mn{dom}(\Dmc)^{|\xbf|}$, $\Dmc \models Q(\abf)$
iff there is a \Dmc-labeling $\ell$ of $q(\xbf)$ such that
\mbox{$\ell(\xbf)=\abf$}.
\end{restatable}

Note that Conditions~1 to~8 can
all be verified in polynomial time, essentially because the evaluation
of OMQs from $(\ELpoly,\text{CQ})$ is in \PTime when the actual CQ is
tree-shaped; this is implicit in \cite{LutzTomanWolter-IJCAI09}, see
also \cite{DBLP:conf/ijcai/BienvenuOSX13}.

\begin{lemma}\label{lem:d-labeling-ptime}
  For an OMQ $Q=(\Omc,\Sbf,q)$ from $(\ELpoly,\text{CQ})$, an
  \Sbf-database \Dmc, and a mapping $\ell:\mn{var}(q) \rightarrow
  \mn{dom}(\Dmc) \cup \{\exists \} \cup (G^q_2 \times \mn{dom}(\Dmc))
  \}$, the problem of deciding whether $\ell$ is a $\Dmc$-labeling of
  $q$ is in \PTime.
\end{lemma}

%
%

\subsection{Existential Pebble Games}

We now describe the polynomial time algorithm for evaluating OMQs from
$(\ELpoly,\text{CQ})^\equiv_{\text{UCQ}_k}$ based on the full schema,
first recalling the existential $k+1$-pebble games from~\cite{DKV02},
in a form that does not make pebbles explicit.  The game is played
between two players, Spoiler and Duplicator, on a CQ $q(\xbf)$, a
database~\Dmc, and a candidate answer \abf. The positions are pairs
$(V,h)$ that consist of a set $V$ of quantified variables from $q$ of
size at most $k+1$ and a mapping
$h:V \cup \xbf \rightarrow \mn{dom}(\Dmc)$ such that $h(\xbf)=\abf$.
The initial position is $(\emptyset,\emptyset)$. In each round of the
game, Spoiler chooses a new set $V$ of size at most $k+1$. Then
Duplicator chooses a new mapping $h:V \rightarrow \mn{dom}(\Dmc)$ such
that if $(V',h')$ was the previous position, then $h(x)=h'(x)$ for all
$x \in V \cap V'$.  Spoiler wins when $h$ is not a homomorphism from
$q|_V$ to \Dmc. Duplicator wins if she has a winning strategy, that
is, if she can play forever without Spoiler ever winning. It is known
that when $q$ is of tree width bounded by $k$, then Duplicator has a
winning strategy if and only if there is a homomorphism from $q$ to
$\Dmc$. This remains true if $q$ is equivalent to a CQ of tree width
bounded by $k$. The existence of a winning strategy for Duplicator can
be decided in polynomial time by a straightforward elimination
procedure: start with the set of all positions, exhaustively eliminate
those from which Duplicator loses in one round, and then check whether
$(\emptyset,\emptyset)$ has survived.


Now let $Q=(\Omc,\Sbf_{\mn{full}},q)$ be an OMQ from
$(\ELpoly,\text{CQ})$, \Dmc an \Sbf-database, and \abf a candidate
answer. To decide whether $\Dmc \models Q(\abf)$, we can assume that
\Dmc is consistent with \Omc since this property is decidable in
polynomial time (implicit in~\cite{LutzTomanWolter-IJCAI09}) and the
result is clear on inconsistent databases.  By
Lemma~\ref{lem:charptime2}, it suffices to find a \Dmc-labeling
of~$q$. This is achieved by a version of the existential
$k+1$-pebble game in which positions take the form $(V,\ell)$ where
$V$ is as before and $\ell$ is a mapping from $V \cup \xbf $ to
$\mn{dom}(\Dmc) \cup \{ \exists \} \cup (G^q_2(q) \times
\mn{dom}(\Dmc))$ such that $\ell(\xbf)=\abf$. The moves of Spoiler and
Duplicator are as before.  Spoiler wins if $\ell$ is not a
\Dmc-labeling of $q|_V$ and the winning condition for Duplicator
remains unchanged. The existence of a winning strategy for Duplicator
can be decided in polynomial time by the same elimination
procedure because, by Lemma~\ref{lem:d-labeling-ptime}, it can be
decided in polynomial time whether a given mapping $\ell$ is a
\Dmc-labeling. The following can be proved in the same way
as without ontologies, see~\cite{DKV02,DBLP:journals/jacm/Grohe07}.
\begin{lemma}
  If $q$ is of tree width at most $k$, then Duplicator has a winning
  strategy if and only if there is a \Dmc-labeling of~$q$.
\end{lemma}
The remaining obstacle on the way to prove
Theorem~\ref{thm:polycombcompl} is that $q$ needs not be of tree width
$k$. It would be convenient to play on a rewriting of $q$ instead,
which by Theorem~\ref{thm:fulldatabase} and Lemma~\ref{lem:rewritings}
is of tree width bounded by $k$.  However, we have no way of computing
a full rewriting in \PTime.  The solution is to first extend $q$ with
additional atoms as in Step~5 of the construction of
rewritings in Section~\ref{sect:eval} and to then play on the
resulting CQ $q'$. It can be shown that this gives the correct result
because when $(\Omc,\Sbf_{\mn{full}},q')$ is a rewriting of $Q$,
then $q'$ must syntactically be a subquery of $q'$.

\section{Deciding Semantic Tree-Likeness}
\label{sect:deciding}

We study the complexity of deciding whether a given OMQ is
$\text{UCQ}_k$-equivalent. Apart from $\ELHIbot$ and its fragments
introduced in Section~\ref{sect:prelims}, we also consider the
additional fragments DL-Lite$^\Rmc$ and DL-Lite$^\Rmc_{\mn{horn}}$,
which are prominent in ontology-based data integration~\cite{ACKZ09,CGLLR07}.

A \emph{basic concept} is a concept name or of one of the forms
$\top$, $\bot$, $\exists r . \top$, and $\exists r^- . \top$. A
\emph{DL-Lite$^\Rmc_{\mn{horn}}$-ontology} is a finite set of statements of
the form
$$
   B_1 \sqcap \cdots \sqcap B_n \sqsubseteq B
   \quad
   r \sqsubseteq s
   \quad
   r \sqsubseteq s^-
   \quad
   r_1 \sqcap \cdots \sqcap r_n \sqsubseteq \bot
$$
where $B_1,\dots,B_n,B$ range over basic concepts and
$r,s,r_1,\dots,r_n$ range over role names. A \emph{DL-Lite$^\Rmc$-ontology} \Omc is a \emph{DL-Lite$^\Rmc_{\mn{horn}}$-ontology} such that whenever $B_1
\sqcap \cdots \sqcap B_n \sqsubseteq B \in \Omc$, then $n=1$ or $B=
\bot$.

\subsection{Non-full Schema}

We first concentrate on the case where the schema is non-full. The
next result provides lower bounds.

\begin{restatable}{theorem}{thmmainhard}
\label{thm:mainhard}
  For any $k \geq 1$, $\text{UCQ}_k$-equivalence is
  \begin{enumerate}
  \item \ExpTime-hard in  $(\EL,\text{CQ})$;
  \item \TwoExpTime-hard in  $(\ELI,\text{CQ})$;
  \item $\Pi^p_2$-hard in $(\text{DL-Lite}^\Rmc,\text{CQ})$.
  \end{enumerate}
  The same lower bounds apply to $\text{CQ}_k$-equivalence, both while
  preserving the ontology and in the general case.
\end{restatable}


Point~1 is proved by reduction from the problem whether a given OMQ
$(\Omc,\Sbf,A(x))$ from $(\EL_\bot,\text{CQ})$ is empty (as defined in
Section~\ref{sect:prelims}), which is known to be
\ExpTime-hard~\cite{jair-data-schema}.
Point~2 is shown by reducing the word problem of an exponentially
space bounded alternating Turing machine $M$, with work alphabet of
size at least $k+1$, to the complement of (U)CQ$_k$-equivalence.
Our reduction carefully makes use of
a construction from~\cite{DBLP:conf/ijcai/Bienvenu0LW16}, where it is
shown that containment in $(\ELI,\text{CQ})$ is~\TwoExpTime-hard.
For Point~3 we provide a reduction from $\forall\exists$-QBF, building
on and extending an \NPclass-hardness proof for the combined
complexity of (a restricted version of) query evaluation in
$(\text{DL-Lite}^\Rmc,\text{CQ})$ on databases of the form $\{ A(a)
\}$ given in \cite{DBLP:conf/dlog/KikotKZ11}.

The next result establishes matching upper bounds.

\begin{restatable}{theorem}{thmdecupper}
\label{thm:decupper}
  For any $k \geq 1$, UCQ$_k$-equivalence is
  \begin{enumerate}
  \item in \ExpTime in $(\ELpoly,\text{UCQ})$;
  \item in \TwoExpTime in $(\ELHIbot,\text{UCQ})$;
  \item in $\Pi^p_2$ in $(\text{DL-Lite}^\Rmc_{\mn{horn}},\text{UCQ})$.
  \end{enumerate}
\end{restatable}
The proof rests on Corollary~\ref{prop:wasaclaim}, that is, we compute
the UCQ$_k$-approximation $Q_a$ of the input OMQ $Q$ and check whether
$Q \subseteq Q_a$ (the converse holds unconditionally). It was shown
in \cite{DBLP:conf/ijcai/Bienvenu0LW16} that OMQ containment is
\ExpTime-complete in $(\ELHbot,\text{CQ})$ and \TwoExpTime-complete in
$(\ELHIbot,\text{CQ})$ and in \cite{DBLP:conf/kr/BienvenuLW12}. that
OMQ-containment is $\Pi^p_2$-complete in
$(\text{DL-Lite}^\Rmc_{\mn{horn}},\text{CQ})$. These results extend to
the OMQ languages in Points~1 and~2 of Theorem~\ref{thm:decupper}.  We
show that with a bit of care we can obtain the same overall
complexities despite the fact that the UCQ in $Q_a$ consisting of
exponentially many (polynomial size) CQs. 

\subsection{Full Schema}

We now focus on the special case of the full schema, where
the complexity of deciding semantic tree-likeness turns out to be
identical to that of query evaluation.
%

\begin{restatable}{theorem}{thmnpcompl}
\label{thm:npcompl}
For any $k \geq 1$, and OMQs based on the full schema,
(U)CQ$_k$-equivalence is complete for
\begin{enumerate}

\item \NPclass between $(\EL,\text{CQ})$ and $(\ELpoly,\text{UCQ})$;

\item \ExpTime between $(\ELI,\text{CQ})$ and $(\ELHIbot,\text{UCQ})$;

\item \NPclass between $(\text{DL-Lite}^\Rmc,\text{CQ})$ and
  $(\text{DL-Lite}^\Rmc_{\mn{horn}})$.

\end{enumerate}
\end{restatable}
The \NPclass lower bounds are inherited from the case where the
ontology is empty~\cite{DKV02}, while the \ExpTime lower bound is
proved by a reduction from the subsumption problem in \ELI~\cite{DBLP:conf/owled/BaaderLB08}.
The upper bounds rest on Theorem~\ref{thm:fulldatabase}, that is,
given an input OMQ $Q=(\Omc,\Sbf,q)$ we first extend $q$ to a CQ $q'$
based on \Omc as in the 
second step of the construction of retracts, then guess a subquery $q''$ of $q'$, and
finally check whether $Q \subseteq (\Omc,\Sbf,q'')$. This yields the
upper bounds in Theorem~\ref{thm:npcompl} as containment of two OMQs
that share the same ontology and are based on the full schema
trivially reduces to OMQ evaluation, which is of the stated
complexities
\cite{DBLP:conf/dlog/KikotKZ11,LutzTomanWolter-IJCAI09,DBLP:conf/kr/OrtizRS10}.

\section{Dealing with Functional Roles}\label{sec:dl-lite-f}

Until now we considered description logics that do not admit functional roles, an important feature for ontology modeling. Here we take a first look, focussing on DL-Lite$^\Fmc$ as a basic yet prominent such DL~\cite{CGLLR07}.  While a main observation is that functional roles result in serious technical complications for tree widths exceeding one, we are also able to obtain some interesting initial results.
Throughout the section, we focus on Boolean CQs (BCQs).
%
%
%
%
%

Recall that a \emph{basic concept} is a concept name or of one of the forms $\top$, $\bot$, $\exists r . \top$, and $\exists r^- . \top$. A \emph{DL-Lite$^\mathcal{F}$-ontology} is a finite set of statements of
the form
$$
   B_1 \sqsubseteq B_2
   \quad
   B_1 \sqcap \cdots \sqcap B_n \sqsubseteq \bot
   \quad
   r_1 \sqcap \cdots \sqcap r_n \sqsubseteq \bot
   \quad
   (\mathsf{funct}~r)
$$
where $B_1,\dots,B_n$ range over basic concepts and
$r_1,\dots,r_n,r$ range over role names.
An interpretation {\em satisfies} $(\mathsf{funct}~r)$ if
whenever $(d,e_1) \in r^{\mathcal{I}}$ and $(d,e_2) \in r^{\mathcal{I}}$, then $e_1 = e_2$.

In the $\text{DL-Lite}^{\Fmc}$ case, our problems are tightly related to the semantic tree-likeness of BCQs over relational databases in the presence of key dependencies 
studied by Figueira~\cite{Figueira16}.
We argue in the appendix that the results of~\cite{Figueira16} can easily be generalized to unions of BCQs (UBCQ). This entails an interesting statement about DL-Lite$^{\Fmc}_{=}$, the fragment of DL-Lite$^{\Fmc}$ that admits only functionality assertions, but no inclusions of any kind.
%
%
\begin{theorem}[Figueira]\label{pro:dl-lite-f-ucq-k-equiv}
  For OMQs from $(\text{DL-Lite}^{\Fmc}_{=},\text{UBCQ})$ based on the full schema, UBCQ$_k$-equivalence while preserving the ontology is in~\TwoExpTime, for any $k \geq 1$.
  %
%
%
%
%
  Moreover, an equivalent OMQ from $(\text{DL-Lite}^{\Fmc}_{=},\text{UBCQ}_k)$ can be constructed in double exponential time (if it exists).
\end{theorem}

%
The above result relies on a sophisticated argument based on tree walking automata. It is open whether the \TwoExpTime~upper bound is optimal. The best known lower bound is \NPclass, from the case without functionality assertions~\cite{DKV02}.



%

We now use Theorem~\ref{pro:dl-lite-f-ucq-k-equiv} to obtain results
for $\text{DL-Lite}^{\Fmc}$.
%
%
Let $(\text{DL-Lite}^{\Fmc},\text{UBCQ})^{\equiv,\mi{po}}_{\text{UBCQ}_k}$ denote the class of OMQs from $(\text{DL-Lite}^{\Fmc},\text{UBCQ})$ that are UBCQ$_k$-equivalent while preserving the ontology.
%
%
\begin{theorem}\label{the:main-dl-lite-f}
For any $k \geq 1$,
\begin{enumerate}
\item $\sf p$-{\sc Evaluation}$((\text{DL-Lite}^{\Fmc},\text{UBCQ})^{\equiv,\mi{po}}_{\text{UBCQ}_k})$ based on the full schema is in \FPTclass.

\item For OMQs from $(\text{DL-Lite}^\Fmc,\text{UBCQ})$ based on the full schema, UBCQ$_k$-equivalence while preserving the ontology is in {\rm 3}\textsc{ExpTime}.
\end{enumerate}
\end{theorem}
The proof of Theorem~\ref{the:main-dl-lite-f} uses first-order rewritability.
For a DL-Lite$^\Fmc$-ontology $\Omc$, let $\Omc^{=}$ be the set of functionality assertions in $\Omc$. We show that every OMQ $Q$ from $(\text{DL-Lite}^{\Fmc},\text{UBCQ})$ can be rewritten into a UBCQ $\mathsf{rew}(Q)$ such that, for every database $\Dmc$ that satisfies $\Omc^=$, $Q(\Dmc) = \mathsf{rew}(Q)(\Dmc)$.
Further, our rewriting procedure preserves the tree width, that is, for every $Q \in (\text{DL-Lite}^\Fmc,\text{UBCQ}_k)$, $\mathsf{rew}(Q) \in \text{UBCQ}_k$.
We obtain the following.

\begin{restatable}{lemma}{lemrewritingucqequiv}
\label{lem:rewriting-ucq-equiv}
Fix $k \geq 1$. For an OMQ $Q = (\Omc,\Sbf_{\mn{full}},q) \in (\text{DL-Lite}^\Fmc,\text{UBCQ})$ the following are equivalent:
\begin{enumerate}
\item $Q$ is UBCQ$_k$-equivalent while preserving the ontology;

\item $(\Omc^=,\Sbf_{\mn{full}},\mathsf{rew}(Q))$ is UBCQ$_k$-equivalent while preserving the ontology.
\end{enumerate}
\end{restatable}

Since the UBCQ $\mathsf{rew}(Q)$ can be constructed in time single exponential in the size of $Q$~\cite{CGLLR07}, Theorem~\ref{the:main-dl-lite-f} is a consequence of Theorem~\ref{pro:dl-lite-f-ucq-k-equiv} and Lemma~\ref{lem:rewriting-ucq-equiv}.


As already observed in~\cite{PabloAndreasPods,Figueira16}, the case $k=1$, is more well-behaved than the general case. The main reason is the fact that the chase procedure for DL-Lite$_{=}^{\Fmc}$, which identifies terms according to the functionality assertions in the ontology, preserves tree width~1 in the sense that when a database \Dmc is of tree width 1, then so is $\mn{ch}_{\Omc}(\Dmc)$.  This fails for tree widths larger than 1. It has, in fact, been observed by Figueira (Lemma~4.3 in~\cite{Figueira16}) that when starting with a database of tree width $k > 1$, then the chase can arbitrarily increase the tree width even if \Omc consists only of a single functionality assertion.

By exploiting the above property, we can strengthen Theorem~\ref{the:main-dl-lite-f} for the case $k=1$, using an approach that does not rely on Theorem~\ref{pro:dl-lite-f-ucq-k-equiv}. In fact, we show that every OMQ $Q=(\Omc,\Sbf,q)$ from $(\text{DL-Lite}^{\Fmc},\text{UBCQ})$ can be rewritten in polynomial time into an OMQ $Q'=(\Omc^\sqsubseteq,\Sbf,q')$, where $\Omc^\sqsubseteq$ denotes the result of removing from \Omc all functionality assertions, such that $Q$ is UBCQ$_1$-equivalent iff $Q'$ is; the proof of the latter again involves first-order rewritability. This allows us to apply results from previous sections, in particular Corollary~\ref{prop:wasaclaim}, Theorem~\ref{thm:main1}, and Theorem~\ref{thm:npcompl}.
%
%
\begin{restatable}{theorem}{thedllitefacyclic}
\label{the:main-dl-lite-f-acyclic}~\\[-4mm]

\begin{enumerate}

\item In $(\text{DL-Lite}^{\Fmc},\text{UBCQ})$, UBCQ$_1$-equivalence coincides with
  UBCQ$_1$-equivalence while preserving the ontology.

\item {\sc Evaluation}$((\text{DL-Lite}^{\Fmc},\text{UBCQ})^{\equiv}_{\text{UBCQ}_1})$ based on the full schema is in \PTime.

\item For OMQs from $(\text{DL-Lite}^\Fmc,\text{UBCQ})$ based on the full schema, UBCQ$_1$-equivalence is \NPclass-complete.

\end{enumerate}
\end{restatable}
%
%
%
%
%
The upper bounds in Theorems~\ref{the:main-dl-lite-f} and~\ref{the:main-dl-lite-f-acyclic} extend to DL-Lite$^\Fmc_{\mn{horn}}$ in a
straightforward way, that is, to
the extension of DL-Lite$^\Fmc$ with statements of the form $ B_1 \sqcap \cdots \sqcap B_n \sqsubseteq B$. Moreover, Theorem~\ref{the:main-dl-lite-f-acyclic} extends to the non-Boolean case.

\section{Conclusion}

An intriguing open problem that emerges from this paper is whether
Theorem~\ref{thm:main1} can be generalized to the case where the
schema is not required to be full, that is, whether OMQ evaluation is
in \PTime in $(\EL,\text{CQ})^\equiv_{\text{UCQ}_k}$ and modest
extensions thereof such as
$(\ELpoly,\text{UCQ}) ^\equiv_{\text{UCQ}_k}$. Note that the companion
Theorem~\ref{thm:main2} does not assume the full schema.  Related to
this seems to be the problem whether CQ$_k$-equivalence coincides with
UCQ$_k$-equivalence in $(\mathcal{EL},\text{CQ})$; recall that this is
not the case in $(\mathcal{ELI},\text{CQ})$ by
Proposition~\ref{prop:UCQnoCQ}.

Being a bit more adventurous, one could be interested in classifying
\PTime combined complexity within $(\ELI,\text{CQ})$ and related
languages, and in classifying \PTime combined complexity and \FPTclass for
DLs with functional roles, existential rule languages such as
(frontier-)guarded rules, and DLs that include negation and
disjunction such as \ALC and \shiq.

\smallskip
\noindent
{\bf Acknowledgments.} We thank the anonymous reviewers for their
helpful comments. This
research was supported by ERC consolidator grant 647289 CODA,
by EPSRC grant EP/S003800/1 EQUID, by the Millennium Institute for
Foundational Research on Data, and Fondecyt grant 1170109.

\cleardoublepage

\appendix

\section*{Technical lemmas and constructions}

We give some fundamental lemmas and technical constructions
used in proofs throughout the paper. The following facts about
homomorphisms, databases, and OMQs are well-known, see for example
\cite{DBLP:journals/tods/BienvenuCLW14}.
\begin{lemma}
\label{lem:hom}
  Let $\Dmc_1$ and $\Dmc_2$ be databases, $h$ a homomorphism from $\Dmc_1$
  to $\Dmc_2$, and $Q=(\Omc,\Sbf,q)$ an OMQ from $(\ELHIbot,\text{UCQ})$. Then
  \begin{enumerate}
  \item if $\Dmc_2$ is consistent with \Omc, then so is $\Dmc_1$;

  \item if $\Dmc_1 \models Q(\abf)$, then $\Dmc_2 \models Q(h(\abf))$ for all
    tuples $\abf$ over $\adom(\Dmc_1)$.

  \end{enumerate}
\end{lemma}

\smallskip

{\bf The chase.}
We next introduce a version of the well-known chase procedure, see for
example \cite{beeri-chase}. %
The constants in
$\mn{ch}_\Omc(\Dmc)$ that are not in \Dmc are called \emph{fresh}.

Let \Omc be an \ELHIbot-ontology. Consider
the following two chase rules that extend an interpretation \Imc based
on the inclusions in \Omc:
\begin{enumerate}

\item if $a \in C^\Imc$, $C \sqsubseteq D \in \Omc$, and $D \neq
  \bot$, then add $D(a)$~to~\Imc;

\item if $(a,b) \in r^\Imc$ and $r \sqsubseteq s \in \Omc$, then add
  $s(a,b)$ to \Imc.

\end{enumerate}
In Rule~1, `add $D(a)$ to \Imc' means to add to \Imc a finite database
that represents the \ELI-concept $D$, identifying its root with
$a$. For example, the concept $A \sqcap \exists r . (B \sqcap \exists
s . \top)$ corresponds to the database $\{A(a), r(a,b), B(b), s(b,c)
\}$.

Let \Dmc be a database. The \emph{chase of \Dmc with \Omc},
denoted $\mn{ch}_\Omc(\Dmc)$, is the potentially infinite
interpretation that is obtained as the limit of any sequence
$\Imc_0,\Imc_1,\dots$, where
\begin{itemize}
\item $\Imc_0=\Dmc$,
\item $\Imc_{i+1}$ is obtained
from $\Imc_i$ by applying both chase rules in all possible ways, and
\item
rule application is {\em fair}, which intuitively means that if an assertion $C \sqsubseteq D$ with $D \neq \bot$, or $r \sqsubseteq s$ in $\Omc$ is not satisfied at some point during the construction of the chase, then eventually it will be satisfied.
\end{itemize}
Since our chase is \emph{oblivious}, that
is, Rule~1 adds a $D(a)$ to $\Imc_i$ even if $a \in D^{\Imc_i}$, the
result of all these sequences is identical up to isomorphism.

\smallskip

{\bf Containment under full schema.}
The following lemma characterizes containment between OMQs based on
UCQs and the full schema. It is an extension of a classic
lemma for containment of UCQs \cite{SY80}.

\begin{lemma}
  \label{lem:contfull}
  Let $Q_i = (\Omc_i,\Sbf_{\mn{full}},q_i) \in (\ELbot,\text{UCQ})$,
$i  \in \{1,2\}$. Then $Q_1 \subseteq Q_2$ iff for every CQ $q'_1$ in
  $q_1$ such that $\Dmc_{q'_1}$ is consistent with $\Omc_1$, there is
  a CQ $q'_2$ in $q_2$ $q'_2 \rightarrow \mn{ch}_{\Omc_2}(q'_1)$.
\end{lemma}

\smallskip

{\bf Unravelings.}
We next describe the unraveling of a database into a (potentially
infinite) database of bounded tree width. Let \Dmc be a database and
$k \geq 1$. The \emph{$k$-unraveling of \Dmc} is obtained as the limit
of a (potentially) infinite sequence of databases
$\Dmc_0,\Dmc_1,\dots$. Along with this sequence, we define additional
sequences $(V_0,E_0,\mu_0),(V_1,E_1,\mu_1),\dots$ and
$\pi_0,\pi_1,\dots$ such that $(V_i,E_i,\mu_i)$ is a tree
decomposition of $\Dmc_i$ of width at most $k$ and $\pi_i$ is a
function from $\mn{dom}(\Dmc_i) \rightarrow \mn{dom}(\Dmc)$ such that
for each $v \in V_i$ we have that $\pi_i|_{\mu_i(v)}$ is an isomorphism between
$\Dmc_i|_{\mu_i(v)}$ and an induced subdatabase of~\Dmc. The initial database $\Dmc_0$ is empty and $V_0$ consists of a single vertex $v$ with $\mu_0(v)=\emptyset$. For each $i \geq 0$, we obtain $\Dmc_{i+1}$, $(V_{i+1},E_{i+1},\mu_{i+1})$, and $\pi_{i+1}$ by extending $\Dmc_{i}$, $(V_i,E_i,\mu_i)$, and $\pi_{i}$ as follows: for every $v \in V_i$ that is a leaf in the tree $(V_i,E_i)$ and every non-empty subdatabase $\Dmc'$ of \Dmc with $\mn{dom}(\Dmc') \leq k + 1$, add a successor $v'$ of $v$; as $\mu_{i+1}(v')$, use the isomorphic copy of $\Dmc'$ obtained by replacing $a \in \mn{dom}(\Dmc')$ with $b \in \mu_i(v)$ if $\pi_i(b) = a$ and with a fresh constant otherwise. These renamings also define $\pi_{i+1}$, in the obvious way. Let $\Dmc^u$, $(V,E,\mu)$, and $\pi$ be the limits of the constructed sequences. It is clear that $(V,E,\mu)$ is an (infinite) tree decomposition of \Dmc of width at most $k$.

For a tuple $\abf$ over $\Ind(\Dmc)$, we define the \emph{k-unraveling of \Dmc up to \abf}, denoted $\Dmc^u_\abf$, to be the database obtained by starting with the unraveling $\Dmc|^u_{\overline{\abf}}$ of the restriction $\Dmc|_{\overline{\abf}}$ of \Dmc to those facts that do not involve a constant from \abf and then adding the following facts:

\begin{enumerate}

\item all facts from \Dmc that involve only constants
  from \abf;

\item for every $r(a,b) \in \Dmc$ with $\{a,b\} \cap \abf = \{ a \}$
  and every $c \in \Ind(\Dmc|^u_{\overline{\abf}})$ with $\pi(c)=b$,
  the fact $r(a,c)$;

\item for every $r(b,a) \in \Dmc$ with $\{a,b\} \cap \abf = \{ a \}$
  and every $c \in \Ind(\Dmc|^u_{\overline{\abf}})$ with $\pi(c)=b$,
  the fact $r(c,a)$.

\end{enumerate}
The following lemma summarizes the main properties of $k$-unravelings.
\begin{lemma}
\label{lem:kunravprop}
  Let \Dmc be a database, $k \geq 1$, $\abf \in \mn{dom}(\Dmc)^n$, and
  $\Dmc^u_\abf$ be the $k$-unraveling of \Dmc up to \abf. Then
  \begin{enumerate}

  \item $\Dmc^u_\abf \rightarrow \Dmc$ via a homomorphism that
    is the identity on~\abf;

  \item $\mn{ch}_\Omc(\Dmc^u_\abf) \rightarrow \mn{ch}_\Omc(\Dmc)$ via
    a homomorphism that is the identity on~\abf, for all
    \ELHIbot-ontologies \Omc;

  \item $\Dmc \models Q(\abf)$ iff $\Dmc^u_\abf \models Q(\abf)$,
    for all OMQs $Q$ from $(\ELHIbot,\text{UCQ}_k)$;

  \item \Dmc is consistent with \Omc iff $\Dmc^u_\abf$ is 
    consistent with \Omc, for all $\ELHIbot$-ontologies \Omc.

  \end{enumerate}
\end{lemma}
\begin{proof}
(1) The function $\pi: \Ind(\Dmc|^u_{\overline{\abf}}) \rightarrow
  \Ind(\Dmc|_{\overline{\abf}})$ introduced during the construction
  of $\Dmc|^u_{\overline{\abf}}$ is a homomorphism from
  $\Dmc|^u_{\overline{\abf}}$ to $\Dmc|_{\overline{\abf}}$. We can extend
  $\pi$ to a homomorphism $\pi'$ from $\Dmc^u_{\abf}$ to \Dmc by setting
  $\pi'(a)=a$ for all $a \in \abf$.

\smallskip 
  (2) We obtain the desired homomorphism $\pi''$ from
  $\mn{ch}_\Omc(\Dmc^u_\abf)$ to $\mn{ch}_\Omc(\Dmc)$
  with $\pi''(\abf)=\abf$ by further extending the homomorphism
  $\pi'$ from $\Dmc^u_{\abf}$ to \Dmc from Point~(1). This is in
  fact easy by replicating in \Dmc any application of a chase rule
  in~$\Dmc^u_\abf$.

\smallskip 

  (3) The `if' direction follows from Point 1 and
  Lemma~\ref{lem:hom}.  For the `only if' direction, assume that
  $\Dmc \models Q(\abf)$ and let $Q=(\Omc,\Sbf,q)$. Then there
  is a homomorphism $h$ from some CQ
  $p=\exists \ybf \, \varphi(\xbf,\ybf)$ in $q$ to
  $\mn{ch}_\Omc(\Dmc)$ with $h(\xbf)=\abf$. Let the CQ $p'$ be
  obtained from $p$ as follows:
  \begin{enumerate}

  \item if $y_1,\dots,y_n$ are quantified variables in $p$ such that
    $p$ contains an atom that contains both $y_i$ and $y_{i+1}$ for
    $1 \leq i < n$, $h(y_1),h(y_n) \in \mn{dom}(\Dmc)$, and
    $h(y_2),\dots,h(y_{n-1}) \notin \mn{dom}(\Dmc)$, then identify
    $y_1$ and $y_n$, that is, replace all occurrences of $y_n$ with $y_1$.

  \item drop all variables $y \in \mn{var}(p)$ such that $h(y) \notin
    \mn{dom}(\Dmc)$.

  \end{enumerate}
  Clearly, $h$ is also a homomorphism from $p' = \exists \ybf' \,
  \varphi'(\xbf,\ybf')$ to $\mn{ch}_\Omc(\Dmc)$ with range
  $\mn{dom}(\Dmc)$.  Moreover, it can be verified that
  $G_{p'|_{\ybf'}}$ is a minor of $G_{p|_\ybf}$ and thus $p'$ has
  treewidth at most~$k$. Extend $p'$ to a CQ $p''$ as follows: for
  every variable $x$ in $p'$ and CI $C\sqsubseteq D \in \Omc$ such
  that $h(x) \in C^{\mn{ch}_\Omc(\Dmc)}$, take a CQ $q_C$ that
  represents $C$ and is variable disjoint from $p'$, and add it to
  $p'$ identifying the root with $x$. Clearly, also $p''= \exists \ybf' \,
  \varphi'(\xbf,\ybf'')$ is of
  treewidth at most $k$ and $h$ can be extended to a homomorphism
  from $p''$ to $\mn{ch}_\Omc(\Dmc)$.

  Let $(V,E,\mu)$ be a tree decomposition of $G_{p''|_{\ybf''}}$ of
  width at most $k$.  We use $h$ and $(V,E,\mu)$ to define a
  homomorphism $h'$ from $p''$ to $\mn{ch}_\Omc(\Dmc^u_\abf)$ with
  $h''(\xbf)=\abf$.

  To start, choose some root $v \in V$ and let $(V',E',\mu')$ be the
  tree decomposition of $D|^u_{\abf}$ from the construction of
  $\Dmc^u_\abf$. By construction of $\Dmc^u_\abf$ and of
  $(V',E',\mu')$, there must be a $v' \in V'$ such that
  \begin{itemize}
  \item[($*$)]   $y \in \mu(v)$ and $h(y) \notin \abf$ implies that there is a
  $c_y \in \mu'(v')$ with $\pi(c_y)=h(y)$.
\end{itemize}
Set $h'(y) = c_y$ for all such $y$ and $h'(y)=h(y)$ if $h(y) \in
\abf$. We now proceed to the neighbors $u$ of $v$ in $(V,E,\mu)$ and
define $h'$ for the variables $y \in \mu(u)$ in the same way, then
proceed to neighbors of neighbors $u$, and so on. In each step except
the first one, there will be variables $y \in \mu(u)$ with $h'(y)$
already defined. But by construction of $\Dmc^u_\abf$ and of
$(V',E',\mu')$, we find a $u' \in V'$ such that ($*$) is satisfied
(with $u$/$u'$ in place of $v$/$v'$) with the additional assumption
that if $h'(y)$ is already defined, then $c_y=h'(y)$.

This finishes the construction of $h'$. It is readily verified that
$h'$ is indeed a homomorphism and that $h'(\xbf)=\abf$. Using the
extension of $p'$ to $p''$, it is now possible in a straightforward
way to extend $h'$ to a homomorphism from $q$ to
$\mn{ch}_\Omc(\Dmc^u_\abf)$.  This yields $\Dmc^u_\abf \models
Q(\abf)$, as required. 
\smallskip 

(4) The `if' direction follows from Point 1 and Lemma~\ref{lem:hom}.
For the `only if' direction, assume that $\Dmc$ is inconsistent with
\Omc. Then \Omc contains a CI $C \sqsubseteq \bot$ and
$a \in C^{\mn{ch}_\Omc(D)}$ for some~$a$, that is, there is a
homomorphism $h$ from $C$ viewed as a tree-shaped CQ $q_C(x)$ to
${\mn{ch}_\Omc(D)}$. As in the proof of Point~3, we can construct a
homomorphism $h'$ from $q_C(x)$ to $\mn{ch}_\Omc(D^u_\abf)$.
It follows that $D^u_\abf$ is inconsistent with \Omc.
\end{proof}
We also need a slightly different version of unraveling a database
\Dmc into a database of tree width~1, which we introduce next.
The difference to the unravelings above is that we choose a
constant $a \in \adom(\Dmc)$ at which to start the unraveling.

For a constant $a \in \adom(\Dmc)$, the \emph{1-unraveling of $\Dmc$
  at $a$} is defined similarly to the 1-unraveling of $\Dmc$ except
that we start the construction with the (potentially empty)
restriction of $\Dmc$ to facts that involve only the constant~$a$,
with the tree decomposition $(V_0,E_0,\mu_0)$ where $V_0$ contains a
single vertex~$v$ with $\mu_0(v)=\{a\}$ rather than $\mu_0(v) =
\emptyset$, and with the function $\pi_0$ that is the identity on~$a$.

\section*{Proofs for Section~\ref{sect:eval}}


\profptjelia*
\begin{proof}\ (sketch) Let the following be given: an OMQ $Q=(\Omc,\Sbf,q)$
  from $(\mathcal{ELHI}_\bot,\text{UCQ}_k)$, an \Sbf-database $\Dmc_0$,
  and a candidate answer $\abf \in \Ind(\Dmc_0)$.

  Our algorithm consists of three different steps. We start with some
  preliminaries. We assume w.l.o.g.\ that the ontology is in normal
  form, that is, all concept inclusions in it are of one of the forms
  $$
  \top \sqsubseteq A, \quad A \sqsubseteq \bot, \quad A_1 \sqcap A_2 \sqsubseteq A, \quad
  \exists r . A \sqsubseteq B, \quad A \sqsubseteq \exists r . B
  $$
  where $A, B, A_1, A_2$ range over concept names and $r$ ranges over roles.
  It is well-known that every can be converted into normal form in
  linear time with the resulting ontology being equivalent up to the
  introduction of fresh concept names
  \cite{DBLP:books/daglib/0041477}.

  An \emph{extended database} is a database that might
  additionally contain \emph{concept facts} of the form $C(a)$ with
  $C$ an \ELI-concept. An interpretation \Imc \emph{satisfies} a
  concept fact $C(a)$ if $a \in C^\Imc$. It is a \emph{model} of an
  extended database \Bmc if \Imc is a model of all facts in \Bmc,
  including the concept facts.
  We write $\Bmc,\Omc \models C(a)$ if
  every model of \Bmc and \Omc satisfies $C(a)$.
  With $\mn{sub}(\Omc)$, we denote the set of all subconcepts of
  concepts that occur in \Omc. For $\ELI_\bot$-concepts $C,D$, we
  write $\Omc \models C \sqsubseteq D$ if $C^\Imc \subseteq D^\Imc$
  for all models \Imc of \Omc. Whether $\Omc \models C \sqsubseteq D$
  can be decided in single exponential time \cite{DBLP:books/daglib/0041477}.
  
  \smallskip
  \emph{Step~1}. We exhaustively apply the following chase rules, starting
  with the database $\Dmc = \Dmc_0$ and producing a sequence of
  extended databases:
  \begin{enumerate}

  \item if $C_1(a),\dots,C_n(a) \in \Dmc$, $\Omc \models C_1 \sqcap 
    \cdots \sqcap C_n \sqsubseteq C$ with $C \in \mn{sub}(\Omc)$, and 
    $C(a) \notin \Dmc,$ then add $C(a)$ to \Dmc; 

  \item if $r(a,b),C(b) \in \Dmc$,
    $\exists r . C \sqsubseteq D \in \Omc$, and $D(a) \notin \Dmc,$
    then add $D(a)$ to \Dmc;

  \item if $r(a,b) \in \Dmc$, $r \sqsubseteq s \in \Omc$, and $s(a,b)
    \notin \Dmc$, then add $s(a,b)$ to \Dmc.


  \end{enumerate}
  It is clear that no chase rule is applicable after at most
  polynomially many steps. From now on, we use \Dmc to refer to the extended 
  database resulting from the chase. It is not difficult to show the
  following.
  \\[2mm]
  {\bf Claim.} For all $a \in \Ind(\Dmc_0)$ and $C \in
  \mn{sub}(\Omc)$, $\Dmc_0,\Omc \models C(a)$ iff $C(a) \in \Dmc$.
  \\[2mm]
  Therefore, if \Dmc contains a concept fact of the form $\bot(a)$,
  then $\Dmc_0$ is inconsistent with \Omc and we can answer `true'.
  Otherwise, $\Dmc_0$ is consistent with \Omc.

  \smallskip \emph{Step~2}. A \emph{type} is a subset
  $t \subseteq \mn{sub}(\Omc)$.  For each interpretation \Imc and
  $a \in \Ind(\Dmc)$, the type realized at $a$ in \Imc is
  $t^\Imc_a= \{ C \in \mn{sub}(\Omc) \mid a \in C^\Imc \}$. For
  types, $t_1,t_2$, we write $\Omc \models t_1 \rightsquigarrow t_2$
  if every model of \Omc that realizes $t_1$ also realizes~$t_2$. For
  every role $r$, we write
  $\Omc \models^{\mn{max}} t_1 \sqsubseteq \exists r . t_2$ if
  $\Omc \models t_1 \sqsubseteq \exists r . t_2$ and $t_2$ is maximal
  with this property regarding set inclusion.  It is not hard to see
  that both $\Omc \models t_1 \rightsquigarrow t_2$ and
  $\Omc \models^{\mn{max}} t_1 \sqsubseteq \exists r . t_2$ can also
  be decided in single exponential time. For example, the former is
  equivalent to $\Bmc,\Omc' \models A(a)$ where \Bmc
  is the extended database $\{ C(a) \mid C \in t_1\}$, $\Omc'$ is
  \Omc extended with $\bigsqcap t_2 \sqsubseteq \bot$, and $A$ is a
  fresh concept name.
  
  In Step~2, we first extend \Dmc as follows: for every
  $a \in \mn{dom}(\Dmc)$, every type $t$ with
  $\Omc \models t^\Dmc_a \rightsquigarrow t$, and every $C \in t$, add
  $C(a_t)$. Note that this extension relies on Step~1 to make explicit
  the types of the constants in $\Ind(\Dmc_0)$. We then apply the
  following chase rule exactly $|q|$ times:
  \begin{itemize}

  \item[($*$)] for every $a \in \Ind(\Dmc)$ and every role $r$ and type $t$ 
    such that $\Omc \models^{\mn{max}} t^\Dmc_a \sqsubseteq \exists r
    . t$ and there is no $b \in \Ind(\Dmc)$ with $r(a,b) \in \Dmc$
    and $t^\Dmc_b \supseteq t$, add $r(a,b)$ and $C(b)$ for every
    $C \in t$, with $b$ a fresh constant.
    
  \end{itemize}
  This second chase generates a tree of depth up to $|q|$ below each
  constant from $\Ind(\Dmc_0)$ and we again denote the result with
  \Dmc. If we chased forever rather than only $|q|$ steps, we would
  obtain a (potentially infinite) universal model of $\Dmc_0$
  and~\Omc. The truncated version \Dmc constructed here is not
  universal, but every neighborhood of diameter at most $|q|$ in the
  universal model is also present in \Dmc, due to the extension of
  \Dmc carried out before starting the second chase. Thus,
  $\Dmc_0 \models Q(\abf)$ iff there is a homomorphism from $q(\xbf)$
  to \Dmc with $h(\xbf)=\abf$ iff $\Dmc \models q(\abf)$ in the
  standard sense of relational databases.
  
  \smallskip \emph{Step~3}. Apply the polynomial time algorithm for
  evaluating CQs of tree width bounded by $k$ (see
  Theorem~\ref{thm:grohe}) to $q$ and \Dmc. 

  \medskip The correctness of the algorithm follows from what was said
  above. Regarding the overall running time, it is not hard to see
  that $|\Ind(\Dmc)| \leq |D_0| + |D_0| \cdot |\Omc|^{|q|}$ and thus
  $|D_0| \leq p(|D_0| \cdot |\Omc|^{|q|})$, $p$ a polynomial. Since
  the check in Step~3 needs only polynomial time, this gives the desired
  $2^{p(|Q|)} \cdot p(|\Dmc_0|)$ bound on the running time where $p$
  is again a polynomial.
\end{proof}

\propUCQnoCQ*
\begin{proof}\
    Let $Q=(\Omc,\Sbf,q)$ be an OMQ from ($\ELI$, CQ) with:
    $$
    \begin{array}{rcrclccrcl}
      \Omc &=& \{ B_1 &\sqsubseteq& A_1 &\quad& B_1 \sqcap \exists r^-.B_4 &\sqsubseteq&  A_3 \\[1mm]
&&       B_2 &\sqsubseteq& A_2 &\quad& B_2 \sqcap \exists r.C_1 &\sqsubseteq&  A_4 \\[1mm]
&&       B_3 &\sqsubseteq& A_3 &\quad& \exists r.B_3 \sqcap C_2 &\sqsubseteq&  A_4 \\[1mm]
&&       B_4 &\sqsubseteq& A_4 &\quad& B_4 \sqcap \exists r.C_3 &\sqsubseteq& A_2 \\[1mm]
 &&      C_1 &\sqsubseteq& A_1 &\quad& \exists r.C_1 \sqcap C_4 &\sqsubseteq&  A_2 \\[1mm]
&&       C_2 &\sqsubseteq& A_2 &\quad& C_2 \sqcap \exists r.B_1 &\sqsubseteq&  A_4 \\[1mm]
&&       C_3 &\sqsubseteq& A_3 &\quad& \exists r.C_3 \sqcap B_2 &\sqsubseteq&  A_4 \\[1mm]
&&       C_4 &\sqsubseteq& A_4 &\quad& C_4 \sqcap \exists r.B_3 &\sqsubseteq&  A_2 \}, \\[1mm]
       \Sbf &=& \multicolumn{7}{l}{\{ B_1,B_2,B_3,B_4,C_1, C_2, C_3, C_4, r\}. }
    \end{array}
    $$
    and $q$ the CQ from Example \ref{ex:first}.


     Then $Q$ is UCQ$_1$-equivalent, but not CQ$_1$-equivalent while
     preserving the ontology. We have that $Q \equiv (\Omc, \Sbf, \exists x_1 \exists x_2 \exists x_3 \exists x_4\, (\phi_1 \vee \phi_2))$, where: $$\phi_1=r(x_2, x_1) \wedge r(x_2, x_3) \wedge A_1(x_1) \wedge A_2(x_2) \wedge A_4(x_2) \wedge A_3(x_3)$$ and $$\phi_2= r(x_2, x_1)\wedge r(x_4, x_1) \wedge A_1(x_1) \wedge A_3(x_1) \wedge A_2(x_2) \wedge A_4(x_4).$$

     To see that $Q$ is not CQ$_1$-equivalent while preserving the
     ontology, assume the contrary. Then, there exists some OMQ
     $Q'=(\Omc, \Sbf, q')$ with $q'$ from CQ$_1$ such that $Q \equiv
     Q'$. Let
     $$
     \begin{array}{rcl}
       \Dmc_1 &=& \{r(b_2, b_1), r(b_2, b_3), r(b_4, b_3),  r(b_4,
       b_1), \\[1mm]
       && \ B_1(b_1), B_2(b_2), B_3(b_3), B_4(b_4)\} \\[1mm]
       \Dmc_2 &=& \{r(c_2, c_1), r(c_2, c_3), r(c_4, c_3), r(c_4,
c_1), \\[1mm]
&&\ C_1(c_1), C_2(c_2), C_3(c_3), C_4(c_4)\}.
     \end{array}
     $$
     Then $\Dmc_1 \models Q$ and $\Dmc_2 \models Q$. Thus $\Dmc_1
     \models Q'$ and $\Dmc_2 \models Q' $ and consequently
     $\mn{ch}_{\Omc}(\Dmc_1) \times \mn{ch}_{\Omc}(\Dmc_2) \models q'$
     where `$\times$' denotes the direct product.  It can be verified
     that $\mn{ch}_{\Omc}(\Dmc_1) \times \mn{ch}_{\Omc}(\Dmc_2)$ is
     isomorphic to $\Dmc_q$.  As $q'$ is from CQ$_1$, it follows that
     $\Dmc_q^u \models q'$, where $\Dmc_q^u$ is the $1$-unraveling of
     $\Dmc_q$, that is,
         $$\begin{array}{lll}\Dmc_q^u& =&\{C_1(x_1),   C_2(x_2),
     C_3(x_3), C_4(x_4), C_1(x_5), \ldots \\[1mm]
         &&\hspace{2mm} r(x_2 ,x_1), r(x_2,x_3), r(x_4, x_3), r(x_4, x_5), \ldots \}\\
         \end{array}$$
Let $Q^u=(\Omc, \Sbf, q^u)$, $q^u$ a Boolean CQ such that $D_{q^u}=D_q^u$.  Then, $Q^u \subseteq Q'\equiv Q$ and $Q \subseteq Q^u$. Thus, $Q \equiv Q^u$.
    Now consider the database
    $$\begin{array}{lll} \Dmc_{12}&= &\{ r(x_2,x_1), r(x_2,x_3), r(x_4, x_3), r(x_4, x_5), \ldots, \\[1mm]
    && \hspace{1mm} B_1(x_1), B_2(x_2), B_3(x_3), B_4(x_4), \\[1mm]
    &&\hspace{1mm} C_1(x_5), C_2(x_6), C_3(x_7), C_4(x_8), \\[1mm]
    && \hspace{1mm} B_1(x_9), \ldots, \\[1mm]
    && \hspace{1mm}\ldots \qquad \qquad\}
    \end{array}
    $$
    We have that $\Dmc_{12} \models Q^u$, but $\Dmc_{12} \not \models Q$ -- contradiction.
\end{proof}

\thmbestapprox*
\begin{proof}
Let $Q=(\Omc,\Sbf,q)$ and $Q_a=(\Omc,\Sbf,q_a)$. For
  Point~1, further let \Dmc be an \Sbf-database of tree width at
  most~$k$.  We can assume that \Dmc is consistent with \Omc.  We have
  $Q_a(\Dmc) \subseteq Q(\Dmc)$ by construction of $Q_a$,
  independently of the tree width of \Dmc. For the converse, let $\abf
  \in Q(\Dmc)$. By Lemma~\ref{lem:chaseprop}, we find a CQ $p(\xbf)$
  in $q$ and a homomorphism $h$ from $p$ to
  $\mn{ch}_\Omc(\Dmc^u_\abf)$ with $h(\xbf)=\abf$. Let $\widehat
  p(\xbf)$ be the contraction of $p$ obtained by identifying any two
  variables $z_1$ and $z_2$ with $h(z_1)=h(z_2)$. Since \Dmc has tree
  width at most $k$, by Lemma~\ref{lem:chaseprop} so has
  $\mn{ch}_\Omc(\Dmc^u_\abf)$. Consequently, also $\widehat p$ has
  tree width at most $k$ and thus is a CQ in $q_a$. It follows that
  $\abf \in Q_a(\Dmc)$.

  For Point~2, let $Q'=(\Omc',\Sbf,q')$.  Take an $\Sbf$-database~\Dmc
  such that $\Dmc \models Q'(\abf)$. We can assume that \Dmc is
  consistent with \Omc.  Consider $\Dmc^u_\abf$, the
    $k$-unraveling of $\Dmc$ up to \abf.  Lemma~\ref{lem:kunravprop}
yields $\Dmc^u_\abf \models Q'(\abf)$ and
  consequently $\Dmc^u_\abf \models Q(\abf)$. From Point~1, we thus
  get $\Dmc^u_\abf \models Q_a(\abf)$. By Lemma~\ref{lem:kunravprop},
  there is also a homomorphism $g$ from $\mn{ch}_{\Omc}(\Dmc^u_\abf)$
  to $\mn{ch}_{\Omc}(\Dmc)$ with $h(\abf)=\abf$. Lemma~\ref{lem:hom}
  yields $\Dmc \models Q_a(\abf)$ as required.
\end{proof}

\thmfulldatabase*
\begin{proof}
  `$2 \Rightarrow 1$' is immediate and so is `$3 \Rightarrow 2$' since
  a maximum constraction clearly always exists.  We show
  `$1 \Rightarrow 3$'. Thus assume that
  $Q=(\Omc,\Sbf_{\mn{full}},q) \in (\ELHIbot,\text{CQ})$ is
  UCQ$_k$-equivalent.  By Corollary~\ref{prop:wasaclaim}, this is
  witnessed by an equivalent OMQ $Q'=(\Omc,\Sbf_{\mn{full}},q')$ from
  $(\ELHIbot,\text{UCQ}_k)$. Since $Q$ is non-empty, so is $Q'$ and we
  can assume w.l.o.g.\ that $\Dmc_{p}$ is consistent with \Omc for any
  CQ $p$ in $q'$. Let $Q_c=(\Omc,\Sbf_{\mn{full}},q_c)$ be any maximum
  contraction of $Q$. Assume w.l.o.g.\ that $q$, $q'$, and $q_c$ all
  have the same answer variables \xbf. Since $Q_c$ and $Q'$ are
  equivalent and by Lemma~\ref{lem:contfull} and since $\Dmc_{q'}$ is
  consistent with \Omc, there must be a CQ $p$ in $q'$ and
  homomorphisms
  \begin{itemize}

  \item $h_1$ from $q_c$ to $\mn{ch}_\Omc(p)$ with $h_1(\xbf)=\xbf$ and

  \item $h_{2}$ from $p$ to $\mn{ch}_{\Omc}(q_c)$ with
    $h_2(\xbf)=\xbf$.

  \end{itemize}
  We can extend $h_2$ to a homomorphism from $\mn{ch}_\Omc(p)$ to
  $\mn{ch}_\Omc(q_c)$, by inductively following the applications of
  the chase rules. The composition $h=h_2 \circ h_1$ is a homomorphism
  from $q_c$ to $\mn{ch}_\Omc(q_c)$ and must be injective since
  otherwise $Q_c$ would not be a maximal contraction. Thus, $h_1$ must
  also be injective. But $p$ has tree width at most $k$, and
  consequently so has $\mn{ch}_\Omc(p)$. It thus follows that $q_c$
  also has tree width at most~$k$.
\end{proof}

\lemrewritings*
\begin{proof}
  For Point~1, it is obvious that $Q \subseteq \widehat Q$. For the
  converse, let $V$ be the set of variables identified in Step~4 of
  the construction of $\widehat Q$ and let $\mn{ch}_\Omc(\Dmc_q)^-$
  denote the restriction of $\mn{ch}_\Omc(\Dmc_q)$ to the variables in
  $V$ and the constants from trees that the chase has generated below
  such variables.  Due to Step~5 of the construction of $\widehat Q$,
  $\mn{ch}_\Omc(\Dmc_q)^-$ is isomorphic to a subset of
  $\mn{ch}_\Omc(\Dmc_{\widehat q})$. We in fact even have
  $\mn{ch}_\Omc(\Dmc_q)^- \subseteq \mn{ch}_\Omc(\Dmc_{\widehat q})$
  if the fresh contants introduced by the chase are chosen uniformly,
  so let us assume that this is the case.  As a consequence, the
  homomorphism $h$ from Step~2 is a homomorphism from $q'$ to
  $\mn{ch}_\Omc(\Dmc_{\widehat q})$. Since $Q$ and thus $\widehat Q$
  is non-empty, $\Dmc_{\widehat q}$ is consistent with \Omc and thus
  it follows that $\widehat Q \subseteq Q' \equiv Q$.

  For Point~2, we again use the fact that $h$ is a homomorphism from
  $q'$ to $\mn{ch}_\Omc(\Dmc_{\widehat q})$, as argued in the proof of
  Point~1. Since $h$ is also a homomorphism from $q'$ to
  $\mn{ch}_\Omc(\Dmc_q)$ and $\Dmc_q$ is consistent with \Omc, $h$ must
  be injective as otherwise $Q'$ would not be a maximum rewriting of
  $Q$. By construction of $\widehat q$ and since the chase adds only
  tree-shaped subdatabases, the treewidth of
  $\mn{ch}_\Omc(\Dmc_{\widehat q})$ is identical to the treewidth of
  the restriction $\widehat q_0$ of $q$ to the range of $h$, as taken
  in Step~3 of the construction of $\widehat Q$. Since $h$ is
  injective, $\widehat q_0$ is an induced subgraph of $q'$ and thus
  the treewidth of $\widehat q_0$ does not exceed that of $q'$, as
  required.
\end{proof}

\section*{Proofs for Section~\ref{sect:fpt}}


\noindent
{\bf Theorem~\ref{thm:grohetechnew}} (Grohe){\bf .}
Given an undirected graph $G=(V, E)$, a $k > 0$, and a
connected 
$\Sbf$-database~$\Dmc$ such that $G_\Dmc$ 
contains the $k \times K$-grid as a minor, 
one can construct in time $f(k)\cdot\mn{poly}(|G|,|\Dmc|)$
an $\Sbf$-database $\Dmc_G$ such that:
  \begin{enumerate}

  \item there is a surjective homomorphism $h_0$ from $\Dmc_G$ to
    $\Dmc$ such that for every edge $\{a, b\}$ in the Gaifman graph of
    $\Dmc_G$:
    $s(a,b) \in \Dmc_G$ iff $s(h_0(a),h_0(b)) \in \Dmc$ for all roles~$s$;

  \item $G$ contains a $k$-clique iff there is a homomorphism $h$ from
    $\Dmc$ to $\Dmc_G$ such that $h_0(h(\cdot))$ is the identity.

  \end{enumerate}

\noindent
\begin{proof}\
  A careful analysis of \cite{DBLP:journals/jacm/Grohe07} reveals that
  the proof given there establishes Theorem~\ref{thm:grohetechnew}
  without the `such that' condition in Point~1. That condition,
  however, can easily be attained as follows. Assume that $G$, $k$,
  and \Dmc are given. First rewrite \Dmc into a new schema $\Sbf'$
  that consists of the concept names in \Sbf and a fresh role name
  $r_R$ for every set
  $R \subset \{ r,r^- \mid r \text{ role name in } \Sbf \}$, by
  replacing every maximal set $\{ r_0(a,b),\dots,r_n(a,b) \}$,
  $r_1,\dots,r_n$ (potentially inverse) roles, with the fact
  $r_{r_0,\dots,r_n}(a,b)$. Then apply Theorem~\ref{thm:grohetechnew}
  without the `such that' condition in Point~1, obtaining $D'_G$ and
  $h_0$. Now revert back $D'_G$ to schema \Sbf in the obvious way.  It
  can be verified that the resulting database $D_G$ and $h_0$ satisfy
  Conditions~1 and~2, also with the `such that' condition in Point~1.
\end{proof}

\lemiolem*
\begin{proof}\  We can find $p_c$ by the following iterative process:
  start with $p_c=p$ and then exhaustively take a homomorphism $h$
  from $p_c$ to $\mn{ch}_\Omc(\Dmc)$ and if $h$ is not injective,
  replace $p_c$ with the contraction $p'_c$ of $p_c$ induced by $h$,
  that is, identify quantified variables $y_1$ and $y_2$ if
  $h(y_1)=h(y_2)$. Clearly, this process must terminate since the
  number of variables in $p_c$ decreases with each step.
\end{proof}

\lemiomin*
\begin{proof}\
%
%
 It remains to argue that the iterative process described in the main part of the paper terminates.  For an $\Sbf$-database \Bmc let $\Qmc_{\Bmc}$ denote the set of contractions
  $q_c$ of $q_w$ such that $\mn{ch}_\Omc(\Bmc) \models^{io}
  q_c$. The partial order `$\preceq$' lifts
  to sets $\Qmc_1,\Qmc_2$ of contractions of $q$ in the obvious way:
  $\Qmc_1 \preceq \Qmc_2$ if for every $q_c \in \Qmc_1$, there is a
  $q'_c \in \Qmc_2$ with $q_c \preceq q'_c$. As usual, we write
  $\Qmc_1 \prec \Qmc_2$ if $\Qmc_1 \preceq \Qmc_2$ and $\Qmc_1
  \not\preceq \Qmc_2$.

  Observe that whenever \Dmc is replaced by $\Dmc'$ in the iterative process, then $\Qmc_{\Dmc'} \prec
  \Qmc_{\Dmc}$:
  \begin{enumerate}

  \item $\Qmc_{\Dmc'} \preceq \Qmc_{\Dmc}$:
  Let $q_c \in \Qmc_{\Dmc'}$. Using $\Dmc' \rightarrow
    \Dmc$, one can prove that $\mn{ch}_\Omc(\Dmc')
    \rightarrow \mn{ch}_\Omc(\Dmc)$, following the applications
    of the chase rules. From $\mn{ch}_\Omc(\Dmc') \models q_c$,
    we thus obtain $\mn{ch}_\Omc(\Dmc) \models q_c$. By
    Lemma~\ref{lem:iolem}, we find a $q_c'$ with $q_c \preceq q'_c$
    and $\mn{ch}_\Omc(\Dmc) \models^{io} q'_c$, thus $q'_c \in
    \Qmc_{\Dmc}$.

  \item $\Qmc_{\Dmc} \not\preceq \Qmc_{\Dmc'}$:
  Assume that $\Dmc$ was replaced by $\Dmc'$
    because of the contractions $q_c$ and $q'_c$ with $\mn{ch}_\Omc(\Dmc)
    \models^{io} q_c$, $\mn{ch}_\Omc(\Dmc')
    \models^{io} q'_c$ and $q'_c \preceq q_c$. Then $q_c \in \Qmc_{\Dmc}$. We show
    that there is no $q''_c \in \Qmc_{\Dmc'}$ with $q_c \preceq
    q''_c$. Assume to the contrary that there is such a $q''_c$. Then $\mn{ch}_\Omc(\Dmc') \models^{io} q''_c$ and thus there is a homomorphism $h$ from $q''_c$ to $\mn{ch}_\Omc(\Dmc')$.  We also have $q_c \rightarrow q''_c$ and $q'_c \rightarrow q_c$ via a non-injective homomorphism since $q_c$ is a
    proper contraction of $q'_c$. By composition, we obtain a
    non-injective homomophism $h'$ from $q'_c$ to
    $\mn{ch}_\Omc(\Dmc')$, in contradiction to
    $q'_c \in \Qmc_{\Dmc'}$.

  \end{enumerate}
%
 Note that `$\preceq$' is a partial
order on the contractions of~$q_w$.
  It remains to note then that when `$\prec$' is interpreted on
  sets of contractions is trivially well-founded. This is because $\prec$ is strict
  and there are only finitely many contractions of $q_w$.
%
%
\end{proof}

\lemdecHomPre*
\begin{proof}\ Condition~1 is decidable since OMQ evaluation
  is. For Condition~2, we first show that it suffices to look at
  databases $\Dmc'$ of a certain restricted shape called pseudo
  trees~\cite{DBLP:conf/ijcai/Bienvenu0LW16} and then argue that this
  enables a reduction to the satisfiability of MSO sentences on
  trees. An alternative for the second step is a reduction to the
  emptiness problem of alternating tree automata, which we conjecture
  to deliver a \TwoExpTime upper bound.

  \smallskip
  \emph{Pseudo trees}. A database \Dmc is a \emph{pseudo tree} if it
  is the union of databases $\Dmc_0,\dots,\Dmc_m$ that satisfy the
  following conditions:
\begin{enumerate}


\item $\Dmc_{1},\dots,\Dmc_m$ are of tree width one;

\item $\Ind(\Dmc_{i})
  \cap \Ind(\Dmc_{j}) = \emptyset$, for $1 \leq i < j \leq k$;

\item $\Dmc_i \cap \Dmc_0$ is a singleton for $1 \leq i \leq k$.

\end{enumerate}
The \emph{width} of $\Dmc$ is $|\Ind(\Dmc_0)|$.
%
\\[2mm]
{\bf Claim.} The condition in Point~2 of Lemma~\ref{lem:iomin} is
equivalent to the same condition when `$\Sbf$-database $\Dmc'$' is
replaced with `pseudo tree \Sbf-database $\Dmc'$ of width at most
$|q_w|$. 
\\[2mm]
\emph{Proof of claim.}  Assume that if
$\mn{ch}_\Omc(\Dmc) \models^{io} q_c$, $q_c$ a contraction of $q_w$, and
that there is an $\Sbf$-database $\Dmc'$ and a contraction $q'_c$ of
$q_w$ such that $\Dmc' \rightarrow \Dmc$,
$\mn{ch}_\Omc(\Dmc') \models^{io} q'_c$, and $q'_c \prec q_c$.  Let
$h$ be an (injective) homomorphism from $q'_c$ to
$\mn{ch}_\Omc(\Dmc')$. Construct a pseudo tree database $\Dmc''$ as
the union of the restriction $\Dmc_0$ of \Dmc to the range of $h$ and
the databases $\Dmc_1,\dots,\Dmc_m$ which are the 1-unravelings of
some constant $a$ from $\Dmc_0$.  It can be verified that $\Dmc'' \rightarrow \Dmc'$, thus $\Dmc' \rightarrow \Dmc$. Moreover,
$\mn{ch}_\Omc(\Dmc'') \models^{io} q'_c$. We have thus replaced
$\Dmc'$ with the pseudo tree database $\Dmc''$.

\smallskip \emph{Reduction to MSO}. It thus suffices to decide the
existence of a pseudo-tree database that satisfies the conditions for
$\Dmc'$ in Point~2 of Lemma~\ref{lem:iomin}.  Let $\Sigma$ be a finite
alphabet. A \emph{$\Sigma$-labeled tree} takes the form $(T,\ell)$
where $T \subseteq S^*$, $S$ a set of any cardinality, is closed under
prefixes and $\ell:T \rightarrow \Sigma$ is a node labeling
function. The satisfiability problem for monadic second-order logic
(MSO) on $\Sigma$-labeled trees is decidable. It
is not hard to encode pseudo-tree \Sbf-databases into $\Sigma$-labeled
trees for an appropriately chosen $\Sigma$,
see~\cite{DBLP:conf/ijcai/Bienvenu0LW16}. Note that the entire
$\Dmc_0$-component can be encoded as a single node label since it
contains only boundedly many constants. Further, it is possible to
express the conditions for $\Dmc'$ in Point~2 of Lemma~\ref{lem:iomin}
as MSO sentences. This is technically very closely related to the tree
automata constructions in~\cite{DBLP:conf/ijcai/Bienvenu0LW16}.  We
refrain from going into technical detail.
\end{proof}

\lemAprop*
\begin{proof}  For Condition~1 of
  Lemma~\ref{lem:iomin}, it is clear that $\Dmc^+ \models
  Q_w$. Further, recall that $\Dmc_0 \not\models Q_a$. It is
  straightfoward to show that $\Dmc^+ \rightarrow \Dmc_0$ and thus also $\Dmc^+ \not\models Q_a$.

  \smallskip

  For Condition~2 of Lemma~\ref{lem:iomin}, assume to the contrary that there are an $\Sbf$-database $\Dmc'$ and contractions $q_c$, $q'_c$ of $q_w$ such that $\mn{ch}_\Omc(\Dmc^+) \models^{io} q_c$, $\Dmc'\rightarrow \Dmc^+$,
  $\mn{ch}_\Omc(\Dmc') \models^{io} q'_c$, and $q'_c \prec
  q_c$. We use a case distinction:
  \begin{itemize}

  \item $\mn{ch}_\Omc(\Dmc_0) \models^{io} q_c$.

    From $\Dmc' \rightarrow \Dmc^+$ and
    $\Dmc^+ \rightarrow \Dmc_0$, we obtain
    $\Dmc' \rightarrow \Dmc_0$. This together with
    $\mn{ch}_\Omc(\Dmc_0) \models^{io} q_c$ and
    $\mn{ch}_\Omc(\Dmc') \models^{io} q'_c$ yields a
    contradiction to $\Dmc_0$ satisfying Condition~2 of
    Lemma~\ref{lem:iomin}.

  \item $\mn{ch}_\Omc(\Dmc_0) \not\models^{io} q_c$.
 From $\mn{ch}_\Omc(\Dmc^+) \models^{io} q_c$ and $\Dmc^+ \rightarrow \Dmc_0$, we obtain $\mn{ch}_\Omc(\Dmc_0) \models q_c$.  By Lemma~\ref{lem:iolem}, there is a contraction $q'_c$ of $q_c$ with $\mn{ch}_\Omc(\Dmc_0) \models^{io} q'_c$. But then we must have $q'_c = q_c$ as otherwise we obtain a contradiction to $\Dmc_0$ satisfying Condition~2 of Lemma~\ref{lem:iomin} (instantiated with $\Dmc=\Dmc'=\Dmc_0)$. Contradiction.

  \end{itemize}

\smallskip

Now for Point~2 of Lemma~\ref{lem:Aprop}. From the fact that
$\Dmc^+ \not\models Q_a$ and that $Q_a$ is equivalent to $Q$ on
$\Sbf$-databases of tree width at most~$\ell$, we obtain that the tree
width of $\Dmc^+$ exceeds~$\ell$.  But by construction of $\Dmc^+$,
the tree width of $\Dmc^+$ cannot be higher than that of \Dmc.
%
%
\end{proof}

\lemCliqueRed*

\begin{proof}\ By what was said in the main body of the paper, it
  remains to prove that %
  \begin{itemize}

  \item[($\dagger$)] for every edge $e$ in \Dmc, there is an atom in $nt(q_c)$
    that $g$ maps to $e$.

  \end{itemize}
  In fact, assume to the contrary that $g$ maps no atom in $nt(q_c)$
  to some edge $\{a,b\} \in \Dmc$. We argue that, then the database
  $\Dmc_1$ obtained from \Dmc by removing all binary facts that
  involve $a$ and $b$ is such that $\Dmc_1^+ \models Q$, contradicting
  the choice of \Dmc. It suffices to show that there is a homomorphism
  $h_1$ from $q_c$ to $\mn{ch}_\Omc(\Dmc_1^+)$. We obtain $h_1$ by
  manipulating the homomorphism $h_0(h(\cdot))$ from $q_c$ to
  $\mn{ch}_\Omc(\Dmc^+)$ in a suitable way. If $h_0(h(\cdot))$ does
  not map any atom in $q_c$ to $\{a,b\}$, then there is nothing to do.
  Otherwise, iterate over all atoms $r(x,y) \in q_c$ that are mapped
  to $\{a,b\}$. 
  Then $r(x,y) \notin nt(q_c)$, that is, this atom must be in a part
  of $q_c$ that is of tree width~1. Assume w.l.o.g.\ that
  $h_0(h(x))=a$ and that $y$ is a successor of $x$ in that part, that
  is, $y$ is further away from the root than $x$; the cases that
  $h_0(h(x))=b$ and/or $x$ is a successor of $y$ can be treated
  analogously and the case that $x=y$ cannot occur since \Dmc contains
  no facts of the form $r(c,c)$. Let $q_y$ be the tree width~1
  subquery of $q_c$ rooted at $y$.  Since $h_0(h(\cdot))$ is
  injective, none of the atoms in $q_y$ is mapped to
  $\{a,b\}$. In $\Dmc_1^+$, a copy of the tree width~1 database
  $\Dmc_a$ has been added, the root identified with $a$. We can
  use $h_0(h(\cdot))$ to (re)define $h_1$ for all variables in
  $q_y$. Informally, whenever $h_0(h(z))$ is (a copy of)
  $c \in \Ind(\Dmc_0)$ in $\Dmc^+$ for any variable $z$ from $q_y$,
  then we can choose for $h_1(z)$ a corresponding copy of $c$ in the
  mentioned copy of $\Dmc_a$ in $\Dmc_1^+$.  This establishes ($\dagger$).
\end{proof}

\subsection{Non-Boolean OMQs}

We next consider the case where $\mathcal{Q} \subseteq
(\ELHIbot,\text{CQ})$ might contain non-Boolean OMQs. We start with a lemma which strengthens the result from Point (1) of Theorem \ref{thm:bestapprox} and whose proof follows from the proof of the same theorem:

\begin{lemma}
\label{lem:kequivansvars}
Let $Q$ be an $n$-ary OMQ from $(\ELHIbot,\text{UCQ})$, $k \geq 1$, and $Q_a$ the UCQ$_k$-approximation of $Q$. Then for every $\Sbf$-database \Dmc and $n$-tuple $\abf$ over $\adom(\Dmc)$ such that $D|_{\overline{\abf}}$ is of tree width $k$, it is the case that: $\Dmc \models Q(\abf)$  iff $\Dmc \models Q_a(\abf)$.
\end{lemma}

Let $Q=(\Omc,\Sbf,q) \notin (\ELHIbot,\text{UCQ})^\equiv_{\text{UCQ}_{\ell}}$ of arity $n$.  Instead of considering the maximal connected components of $q$ as in the Boolean case, we consider the connected components $q_1$, \ldots, $q_p$ of the Boolean CQ $q|_{\ybf}$, where $\ybf$ are the quantified variables of $q$. For each $i$, let $q^+_i$ be the restriction of $q$ to the variables $\xbf \cup \mn{var}(q_i)$. Note that it is not guaranteed that all variables from \xbf occur in some atom of $q^+_i$. We nevertheless assume that the answer variables of $q^+_i$ are exactly $\xbf$, which can be achieved e.g.\ by admitting dummy atoms of the form $\mn{adom}(x_i)$ where \mn{adom} is assumed to
be true for all constants in the input database. Also, let $Q_i=(\Omc,\Sbf,q^+_i)$.

As $Q=(\Omc,\Sbf,q) \notin (\ELHIbot,\text{UCQ})^\equiv_{\text{UCQ}_{\ell}}$, it must be the
case that there is some $w$ such that $Q_w \notin
(\ELHIbot,\text{UCQ})^\equiv_{\text{UCQ}_{\ell}}$.  Let
$q^-(\xbf)=\bigwedge_{1\leq i \neq w \leq p} q^+_i$ and
$Q^-=(\Omc,\Sbf,q^-)$. It can be assumed that $Q \not \equiv Q^-$ (otherwise $Q$ could be replaced by $Q^-$). Let $D^-$ be a database such that $D^- \models Q^-(\abf)$, but $D^- \not \models Q(\abf)$, for some tuple $\abf$ over $\adom(D^-)$. Then $D^- \not \models Q_w(\abf)$. We assume w.l.o.g. that all constants in $\abf$ are distinct. If this is not the case, we can duplicate constants from $\adom(D^-)$ to obtain a tuple with distinct elements which is an answer to $Q^-$.

Let $Q_a=(\Omc,\Sbf, q_a)$ be the $\ell$-approximation of
$Q_w$. Furthermore, let $\Bmc$ be a database such that $\Bmc \models Q_w(\bbf)$, but $\Bmc \not \models Q_a(\bbf)$ for some tuple $\bbf$ over $ \adom(\Bmc)$. We can assume again w.l.o.g. that all constants in
$\bbf$ are distinct and also that the domains of $D^-$ and $\Bmc$ are
disjoint. This also means that we can rename $\bbf$ as $\abf$ in $\Bmc$ and obtain that:  $\Bmc \models Q_w(\abf)$, but $\Bmc \not \models Q_a(\abf)$. 


For a contraction $q_c(\xbf)$ of $q$, a database \Dmc, and an $\abf \in \mn{dom}(\Dmc)^{|\xbf|}$, we write $\Dmc \models^{io} q_c(\abf)$ if all homomorphims $h$ from $q_c$ to $\Dmc$ with $h(\xbf)=\abf$ are injective.
%
%
We also use again a modified version of Lemma \ref{lem:iomin},
established by the same proof.

\begin{lemma}
\label{lem:iominSecond}
Let $Q=(\Omc,\Sbf,q)$ be an OMQ from $(\ELHIbot,\text{CQ})$ and \Dmc
an \Sbf-database with $\Dmc \models Q(\abf)$.
Then there is an $\Sbf$-database $\Dmc'$ such that the following
conditions are satisfied:
  \begin{enumerate}
  \item $\Dmc' \rightarrow \Dmc$, $\Dmc' \models Q(\abf)$ and
  \item if $\mn{ch}_\Omc(\Dmc') \models^{io} q_c(\abf)$, $q_c$ a contraction
    of $q$, then it is not the case that there is an $\Sbf$-database
    $\Dmc''$ and a contraction $q'_c$ of $q$ such that $\Dmc''
    \rightarrow \Dmc'$, $\mn{ch}_\Omc(\Dmc'') \models^{io}
    q'_c(\abf)$, and \mbox{$q'_c \prec q_c$}.
  \end{enumerate}
\end{lemma}

We consider now the database $\Dmc_0$ obtained from Lemma~\ref{lem:iominSecond} when $Q$ is assumed to be $Q_w$, $\Dmc$ to be $\Bmc$ and $\abf$ to be itself. We consider subsets $\Dmc$ of $\Dmc_0$ and construct $\Dmc^+$ as in the Boolean case. We take such a $\Dmc$ which is subset-minimal with the property that $\Dmc^+ \models Q_w(\abf)$.  As $\Dmc^+ \rightarrow \Dmc_0$, $\Dmc_0 \rightarrow \Bmc$ and $\Bmc \not \models Q_a(\abf)$, it follows that $\Dmc^+ \not \models Q_a(\abf)$.  Further on, from Lemma \ref{lem:kequivansvars} it follows that $(\Dmc^+)_{\overline{\abf}}$ has tree width exceeding $\ell$. From the construction of $\Dmc^+$ it follows that the tree width of $\Dmc|_{\overline{\abf}}$ exceeds $\ell$, so it must contain the $k \times K$ grid as a minor. We apply Theorem \ref{thm:grohetechnew} to $\Dmc|_{\overline{\abf}}$ and obtain a new database $(\Dmc|_{\overline{\abf}})_G$ and a homomorphism $h_0$ from $(\Dmc|_{\overline{\abf}})_G$ to $\Dmc|_{\overline{\abf}}$ such that Points 1 and 2 of that theorem are satisfied. We construct a new
database $\Dmc_G$ by reattaching the $\abf$-part of $\Dmc$ to $(\Dmc|_{\overline{\abf}})_G$ as follows:
\begin{itemize}
\item for every atom $A(a) \in \Dmc$ with $a \in \abf$, we add $A(a)$ to $(\Dmc|_{\overline{\abf}})_G$;
\item for every atom $R(a_1, a_2) \in \Dmc$ with $a_1, a_2 \in \abf$, we add $R(a_1, a_2)$ to $(\Dmc|_{\overline{\abf}})_G$;
\item for every atom $R(a, b)$ or $R(b, a)$ in $\Dmc$ with $a \in \abf$ and $b \in \adom(\Bmc')$, and every $b' \in \adom((\Dmc|_{\overline{\abf}})_G)$ such that $h_0(b')=b$ we add $R(a, b')$ or $R(b', a)$ to $(\Dmc|_{\overline{\abf}})_G$;
\end{itemize}

It can be checked that $\Dmc_G$ maps into $\Dmc$ with an extension of the homomorphism $h_0$ with the identity homomorphism on $\abf$. We construct $\Dmc_G^+$ as in the original proof and than let $\Dmc_G^*$ be the union of $\Dmc^-$ and $\Dmc_G^+$. It can be shown then that $G$ has a $k$-clique iff $\Dmc_G^* \models Q(\abf)$.

\subsection{From CQs to UCQs}

We next explain how to extend the proof from CQs to UCQs,
that is, to the case where $\Qmc \subseteq (\ELHIbot,\text{UCQ})$. 
Let $Q=(\Omc,\Sbf,q) \notin (\ELHIbot,\text{UCQ})^\equiv_{\text{UCQ}_\ell}$. Note that $q$ is a disjunction of conjunctions of connected Boolean CQs (\emph{conCQs}, for short), and that we can find an equivalent conjunction of
disjunctions of conCQs~$q'$. The conjuncts of $q'=q_1 \wedge \cdots
\wedge q_n$ are disjunctions (unions) of connected Boolean CQs
(\emph{UconCQs}, for short). Let $Q_i=(\Omc,\Sbf,q_i)$, for $1 \leq i
\leq n$. We can assume w.l.o.g.\ that $Q_i \not\subseteq Q_j$ for all
$i \neq j$. Also, some $Q_v$ must be non-equivalent to its
UCQ$_\ell$-approximation.  In the remainder of the proof, the UconCQs $q_1,\dots,q_{v-1},q_{v+1},\dots,q_n$ play exactly the role of the conCQs $q_1,\dots,q_{w-1},q_{w+1},\dots,q_n$ in the original proof. 
Let us look more
closely at the UconCQ $q_v$. Let $q_v=p_1 \vee \cdots \vee p_m$ and let $P_i = (\Omc,\Sbf,p_i)$ for
\mbox{$1 \leq i \leq m$}. We can assume w.l.o.g.\ that $P_i
\not\subseteq P_j$ for all $i \neq j$. We also replace any $p_i$ with
its UCQ$_\ell$-approximation $p^a_i$ whenever the OMQ
$$
P^a_i := (\Omc,\Sbf,p_1\vee\dots\vee p_{i-1} \vee p^a_i \vee p_{i+1} \vee \dots \vee p_m)
$$
is equivalent to $Q_v$. Since $Q_v$ is not equivalent to its
UCQ$_\ell$-approximation, these assumptions imply that there is a $w$
such that $P_w \not\subseteq P^a_w$. Let $P_a$ be the UCQ$_\ell$-approximation of~$P_w$. 
Since $P_w \not\subseteq P^a_w$,
there must be an \Sbf-database \Dmc such that $\Dmc \models P_w$,
$\Dmc \not\models P^a$, and $\Dmc \not\models P_i$ for all $i \neq w$.
In the remainder of the proof, the conCQ $p_w$ plays the role of the
conCQ $q_w$ in the original proof. We next observe that the proof of
Lemma~\ref{lem:iomin} actually establishes a slightly stronger result,
namely that for every OMQ $Q^*$ from $(\ELHIbot,\text{CQ})$ over
schema \Sbf and \emph{every} $\Sbf$-database $\Dmc$ with $\Dmc \models
Q^*$, there is an $\Sbf$-database $\Dmc'$ such that $\Dmc'
\rightarrow \Dmc$, $\Dmc' \models Q^*$, and Condition~2 of
Lemma~\ref{lem:iomin} is satisfied (with $Q_w$ replaced by $Q^*$).
%




%

For the remaining proof, we start with the database $\Dmc_0$ that is
obtained as $\Dmc'$ when invoking the strengthened lemma with $Q=P_w$
and the database \Dmc that we had identified before. Then, $\Dmc_0
\models P_w$, and from the fact that $\Dmc_0 \rightarrow \Dmc$, it
follows that: $\Dmc_0 \not \models P^a$, and $\Dmc_0 \not \models
P_i$, for every $i$, with $i \neq w$. The rest of the proof goes
through with essentially no change.

\section*{Proofs for Section~\ref{sect:ptimeupper}}

Towards a proof of Lemma~\ref{lem:charptime2}, we first characterize
answers to OMQs based on a weaker form of \Dmc-labelings. Arguably,
these are more intuitive, but also less `local' than the
\Dmc-labelings defined in Section~\ref{sect:ptimeupper}.

A \emph{weak \Dmc-labeling} of $q$ is a function $\ell:\mn{var}(q)
\rightarrow \mn{dom}(\Dmc) \cup \{ \exists \}$ such that the following
conditions are satisfied:
\begin{enumerate}

\item $\ell(x) \in \mn{dom}(\Dmc)$ for every answer variable $x$;

\item the restriction of $\ell$ to the variables in $V:=\{ x \mid
  \ell(x) \in \mn{dom}(\Dmc) \}$ is a homomorphism from $q|_V$ to
  $\mn{ch}^-_\Omc(\Dmc)$;

\item if $r(x,y) \in q$ and $\ell(y) \in \mn{dom}(\Dmc)$, then
  $\ell(x) \in \mn{dom}(\Dmc)$;

\item if $(x,y) \in G_2$, $\ell(x) \in \mn{dom}(\Dmc)$, and
  $\ell(y)=\exists$, then
  \begin{enumerate}

  \item $(x,y)$ is $\exists$-eligible;

  \item $\Dmc \models (\Omc,\Sbf,\mn{dtree}_{(x,y)})(\ell(x))$; and

  \item $\ell(x)=\ell(x')$ for all $x' \in \mn{reach}^0(x,y)$.

  \end{enumerate}

\item if $q'$ is an $\exists$-MCC of $q$ such that
  $\ell(x) \notin \mn{dom}(\Dmc)$ for every variable $x$ in $q'$, then
  $q'$ is a homomorphic preimage of a ditree and 
  $\Dmc \models (\Omc,\Sbf,\exists x_0 \, \mn{dtree}_{q'})$.

\end{enumerate}
\begin{lemma}
\label{lem:charptime}
Let \Dmc be an \Sbf-database that is consistent with \Omc and
$\abf \in \mn{dom}(\Dmc)^{|\xbf|}$. Then $\Dmc \models
Q(\abf)$ iff there is a weak \Dmc-labeling $\ell$ of $q(\xbf)$ such
that $\ell(\xbf)=\abf$.
\end{lemma}
\begin{proof}\ (sketch) `if'. Assume that $\ell$ is a weak
  \Dmc-labeling of $q(\xbf)$ such that $\ell(\xbf)=\abf$. To show that
  $\Dmc \models Q(\abf)$, it suffices to construct a homomorphism $h$
  from $q$ to $\mn{ch}_\Omc(\Dmc)$ with $h(\xbf)=\abf$. We start by
  putting $h(x)=\ell(x)$ whenever $\ell(x) \in \mn{dom}(\Dmc)$. Next,
  consider every $(x,y) \in G_2$ such that $\ell(x) \in
  \mn{dom}(\Dmc)$ and $\ell(y)=\exists$. We extend $h$ to all
  variables in $\mn{reach}(x,y)$ by using the homomorphism from
  $q|_{\mn{reach}(x.y)}$ to $\mn{dtree}_{(x,y)}$ (existence guaranteed
  by definition of $\exists$-eligible) and the homomorphism from
  $\mn{dtree}_{(x,y)}$ to $\mn{ch}_\Omc(\Dmc)$ that maps the root of
  $\mn{dtree}_{(x,y)}$ to $\ell(x)$ (existence guaranteed by
  Condition 4b). It remains to treat all $\exists$-MCCs $q'$ of $q$.
  Here, we combine the homomorphism from $q'$ to $\mn{dtree}_{q'}$ and
  from $\mn{dtree}_{q'}$ to $\mn{ch}_\Omc(\Dmc)$ (existence guaranteed by
  Condition 5). It can be verified that $h$ is indeed a homomorphism.

\smallskip
`only if'. Assume that $\Dmc \models Q(\abf)$. Then there is a
homomorphism $h$ from $q$ to $\mn{ch}_\Omc(\Dmc)$ with $h(\xbf)=\abf$.
For all variables $x$ in $q$, define
$$
\ell(x)=
\left \{
  \begin{array}{ll}
h(x) & \text{if } h(x) \in \mn{dom}(\Dmc); \\[1mm]
\exists & \text{otherwise}.
  \end{array}
\right.
$$
It can be verified that $\ell$ is a weak \Dmc-labeling of $q$ with
$\ell(\xbf)=\abf$. For Condition~4, one uses that when $(x,y) \in
G_2$, $\ell(x) \in \mn{dom}(\Dmc)$, and $\ell(y)=\exists$, then there
is a homomorphism from $q|_{\mn{reach}(x,y)}$ to the database that the
chase has generated below $\ell(x)$, which takes the form of a ditree
with multi-edges; this shows that $(x,y)$ is
$\exists$-eligible. Condition 4b then follows from the choice of
$\mn{dtree}_{(x,y)}$.
\end{proof}
The problem with weak \Dmc-labelings is that Condition~4c is not yet
sufficiently `local', that is, the variables $x$ and $x'$ mentioned
there can be arbitrarily far apart in a tree decomposition of
$q$. This is rectified in (non-weak) \Dmc-labelings as defined in
Section~\ref{sect:ptimeupper}. We are now ready to prove correctness
of the characterization of OMQ answers in terms of the latter.

\lemcharptimetwo*
\begin{proof}\ (sketch) `if'. Assume that there is a \Dmc-labeling
  $\ell$ of $q(\xbf)$ such that $\ell(\xbf)=\abf$.  Let $\ell'$ be
  obtained from $\ell$ by setting $\ell'(x)=\exists$ iff $\ell(x)
  \notin \mn{dom}(\Dmc)$. It suffices to show that $\ell'$ is a weak
  \Dmc-labeling of $q$. The only condition that is not immediate is
  Condition~4c of weak \Dmc-labelings. So assume that $(x,y) \in G_2$,
  $\ell(x) \in \mn{dom}(\Dmc)$, and $\ell(y)=\exists$. Let $x' \in
  \mn{reach}^0(x,y)$. By Condition 4c of \Dmc-labelings,
  $\ell(y)=((x'',y''),\ell(x))$ where $x'' \in \mn{reach}^0(x,y)$ and $y''
  \in \mn{reach}^1(x,y)$. By Conditions~5 to~7 of \Dmc-labelings
  and since it can be easily proved that
  $\mn{reach}^j(x,y)=\mn{reach}^j(x'',y'')$
  for all $j$, we obtain $\ell(x')=\ell(x)$ as required.

  \smallskip `only if'. Assume that $\Dmc \models Q(\abf)$. Then there
  is a homomorphism $h$ from $q$ to $\mn{ch}_\Omc(\Dmc)$ with
  $h(\xbf)=\abf$. For each variable $z$ such that $h(z) \notin
  \mn{dom}(\Dmc)$ and there is an $(x,y) \in G_2$ for which $h(x) \in
  \mn{dom}(\Dmc)$, $h(y) \notin \mn{dom}(\Dmc)$, Conditions~4a and~4b
  of \Dmc-labelings are satisfied, and $z \in \mn{reach}(y)$, choose
  such an $(x,y)$, and denote it with $(u_z,v_z)$ For all variables
  $x$ in $q$, define
$$
\ell(x)=
\left \{
  \begin{array}{l@{\quad}l}
h(x) & \text{if } h(x) \in \mn{dom}(\Dmc); \\[1mm]
((u_z,v_z),h(u_z)) & \text{if } h(x) \notin \mn{dom}(\Dmc) \\[1mm]
&
  \text{ \ and } u_z, v_z\text{ are defined}  \\[1mm]
\exists & \text{otherwise}.
  \end{array}
\right.
$$
It can be verified that $\ell$ is a \Dmc-labeling of $q$ with
$\ell(\xbf)=\abf$.
\end{proof}

\thmpolycombcompl*
\begin{proof}\
  Let $Q=(\Omc,\Sbf_{\mn{full}},q)$ be from
  $(\ELpoly,\text{CQ})^\equiv_{\text{UCQ}_k}$,
  $Q_r=(\Omc,\Sbf_{\mn{full}},q_r)$ a rewriting of $Q$, \Dmc the input
  database that is consistent with \Omc, and \abf a candidate
  answer. By Theorem~\ref{thm:fulldatabase} and
  Lemma~\ref{lem:rewritings}, $q_r$ is of tree width bounded by
  $k$. Ideally, we would like to play the modified game on $q_r$ and
  answer `yes' if Duplicator has a winning strategy for any of these
  CQs.

  However, we do not have a full rewriting in our hands as we have no
  way of computing one in \PTime.  To solve this problem, we first
  extend $q$ as follows: for each variable $x$ in $q$ and each concept
  inclusion $C \sqsubseteq D \in \Omc$ with
  $\mn{ch}_\Omc(\Dmc_q) \models C(x)$, $x$ viewed as a constant, take
  a fresh copy $q_C$ of $C$ viewed as a CQ and add $q_C$ to $q$,
  identifying $x$ with the root of $q_C$. Note that $Q$ is equivalent
  to $Q^+=(\Omc,\Sbf_{\mn{full}},q^+)$ and that by construction, $q_r$
  syntactically is a subquery of the resulting CQ $q^+$. We play the
  modified game on $q^+$ rather than on $q_r$.

  We have to argue that Duplicator has a winning strategy on~$q^+$ iff
  $\Dmc \models Q(\abf)$. The `if' direction is clear since
  $\Dmc \models Q(\abf)$ implies $\Dmc \models \Qmc^+(\abf)$, thus
  Lemma~\ref{lem:charptime2} yields a \Dmc-labeling $\ell$ of $q^+$
  with $\ell(\xbf)=\abf$, and $\ell$ clearly gives rise to a winning
  strategy for Duplicator on $q^+$.  Conversely, a winning strategy
  for Duplicator on $q^+$ also gives such a strategy on any subquery
  of $q^+$, such as $q_r$. Thus, there is a \Dmc-labeling $\ell$ of
  $q_r$, which means that $\Dmc \models Q_r(\abf)$ and thus
  $\Dmc \models Q(\abf)$.
\end{proof}

\section*{Proofs for Section~\ref{sect:deciding}}
\label{app:deciding}


\thmmainhard*
\begin{proof}\ Point~1 is proved by reduction from the following
  problem: given an OMQ of the form $Q=(\Omc,\Sbf, A(x))$ with \Omc formulated in $\EL_\bot$, is $Q$ empty?  For $Q$ to be empty there must be no \Sbf-database \Dmc that is consistent with \Omc and which satisfies $\Dmc \models Q(a)$ with $a \in \Ind(\Dmc)$.  This problem is known to be \ExpTime-hard \cite{jair-data-schema}.

  We start with the case $k=1$, that is, we consider
  UCQ$_1$-equivalence and CQ$_1$-equivalence, the latter both while
  preserving the ontology and in the general case. We use the same
  reduction for all three cases and afterwards explain how to
  generalize to $k>1$.

  Let $Q=(\Omc,\Sbf,A(x))$ be as stated above. Reserve fresh concept
  names $B, B_1,B_2$ and a fresh role name $r$. Let $\Omc^*$ be
  obtained from \Omc by replacing every concept inclusion of the form
  $C \sqsubseteq \bot$ with $C \sqsubseteq B$.  Now define
  $$
  \begin{array}{rcl}
    \Omc' &=& \Omc^* \cup \{ \exists s . B \sqsubseteq B \mid s \text{
      role name in } \Omc\} \; \cup \\[1mm]
              && \hspace*{8.5mm} \{ B \sqsubseteq A \sqcap \exists r . (B_1
              \sqcap B_2 \sqcap \exists r . \top) \} \\[1mm]
    \Sbf' &=& \Sbf \cup \{ r, B_1,B_2 \} \\[1mm]
    q(x) &=& \exists y_1 \exists y_2 \exists z \, A(x) \wedge r(x,y_1) \wedge
             r(x,y_2) \wedge B_1(y_1) \; \wedge \\[1mm]
             && \hspace*{14mm} B_2(y_2) \wedge r(y_1,z)
             \wedge r(y_2,z)
  \end{array}
  $$
  and set $Q'=(\Omc',\Sbf',q)$.
  \\[2mm]
  {\bf Claim.} $Q$ is empty iff $Q'$ is (U)CQ$_1$-equivalent (while
  preserving the ontology or not).
  \\[2mm]
  For (the contrapositive of) `if', assume that $Q$ is non-empty and
  take an $\Sbf$-database $\Dmc_0$ that is consistent with \Omc and
  satisfies $\Dmc_0 \models A(a)$. Glue to $a$ in $\Dmc_0$ a copy of
  $q$ without the atom $A(x)$ and call the resulting
  $\Sbf'$-database~$\Dmc$. Clearly, $\Dmc \models Q'(a)$.  Let
  $\Dmc^u_a$ be the 1-unraveling of $\Dmc$ up to~$a$. Then $\Dmc^u_a
  \not\models Q'(a)$ since the copy of $q$ that we have glued to
  $\Dmc_0$ has been `broken' by unraveling, $B_1,B_2$ do not occur in
  $\Omc^*$, and the concept inclusion $ B \sqsubseteq A \sqcap \exists
  r . (B_1 \sqcap B_2 \sqcap \exists r . \top)$ in $\Omc'$ cannot fire
  since consistency of $\Dmc_0$ with \Omc and
  Lemma~\ref{lem:kunravprop} imply that \Omc derives $B$ at $a$ in
  $\Dmc^u_a$. But by Lemma~\ref{lem:kunravprop}, no OMQ from
  $(\ELHIbot,\text{UCQ}_1)$ can distinguish $a$ in $\Dmc$ from $a$ in
  $\Dmc^u_a$.  Consequently, $Q'$ is not UCQ$_1$-equivalent.

  \smallskip For `only if', assume that $Q$ is empty. We show that
  $Q'$ is equivalent to $(\Omc',\Sbf',B(x))$ and thus
  CQ$_1$-equivalent while preserving the ontology.  The containment
  $(\Omc',\Sbf',B(x)) \subseteq Q'$ is immediate by construction of
  $Q'$. For the converse, let $\Dmc$ be an $\Sbf'$-database with $\Dmc
  \models Q'(a)$. This clearly implies that \Omc derives $A$ at $a$ in
  \Dmc. Since $Q$ is empty, this is only possible when $B$ is also
  derived at $a$. Thus $\Dmc \models (\Omc',\Sbf',B(x))(a)$.

  \smallskip Sketch for UCQ$_k$-case, $k>1$. The main properties of
  CQ~$q$ in the above proof are that it is of tree width at least
  $k+1$, homomorphically maps into a directed tree (that can be
  generated by an \EL-concept), and does not admit a homomorphism to
  its $k$-unraveling. To achieve the same for $k>1$, we can use as~$q$
  the undirected $k+2$-clique whose vertices we assume w.l.o.g.\ to be
  $\{1,\dots,k+2\}$, orient each edge $\{i,j\}$ in the direction from
  $i$ to $j$ if $i<j$, subdivide each edge $(i,j)$ into $j-i$ edges by
  introducing intermediate points, represent directed edges as
  $r$-atoms, add the atom $A(1)$, and finally label each vertex with a
  different concept name $B_i$, $i \geq 1$. For the case $k=2$, the
  resulting CQ is displayed in Figure~\ref{fig:CQkthree}.
  \begin{figure}[t]
    \begin{boxedminipage}[t]{\columnwidth}
      \begin{center}
 \begin{tikzpicture}[auto, scale=0.7]
 \GraphInit
 \scriptsize
 \node(x1) [active] at (0, 0) {$1$};
 \node(x1l)  at (-0.8, 0) {$A, B_1$};

 \node(x2) [active] at (4, 0) {$2$};
  \node(x2l)  at (4.6, 0) {$B_2$};

 \node(x3) [active] at (0, 4) {$3$};
 \node(x3l)  at (-0.6, 4) {$ B_7$};

 \node(x4) [active] at (4, 4) {$4$};
 \node(x4l)  at (4.6, 4) {$B_8$};

 \node(x13) [empty] at (0, 2) {};
  \node(x13l)  at (-0.5, 2) {$B_4$};

 \node(x24) [empty] at (4, 2) {};
  \node(x24l)  at (4.5, 2) {$B_6$};

 \node(x141)[empty]  at (1.33, 1.33) {};
  \node(x141l)  at (0.95, 1.55) {$B_3$};

 \node(x142)[empty]  at (2.66, 2.66) {};
 \node(x142l)  at (2.25, 2.85) {$B_5$};

 \draw [-latex] (x1) -- (x13) node[midway, left] {$r$};
 \draw [-latex] (x13) -- (x3) node[midway, left] {$r$};
  \draw [-latex] (x3) -- (x4) node[midway, above] {$r$};
 \draw [-latex] (x1) -- (x141) node[midway, above] {$r$};
 \draw [-latex] (x141) -- (x142) node[midway, above] {$r$};
 \draw [-latex] (x142) -- (x4) node[midway, above] {$r$};
 \draw [-latex] (x1) -- (x2) node[midway, above] {$r$};
 \draw [-latex] (x2) -- (x3);

 \node(l) at (1,3.2) {$r$};

 \draw [-latex] (x2) -- (x24) node[midway, right] {$r$};
 \draw [-latex] (x24) -- (x4) node[midway, right] {$r$};

 \end{tikzpicture}
      \end{center}
    \end{boxedminipage}
    \caption{CQ for $k=2$}
    \label{fig:CQkthree}
  \end{figure}
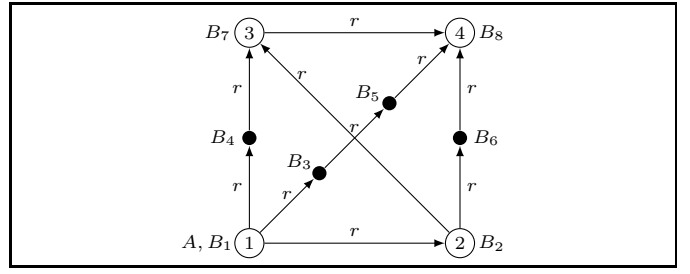
  The reduction can then be adapted by replacing $q$ as described,
  replacing $B_1,B_2$ in $\Sbf'$ with the set of all fresh concept
  names $B_i$ in $q$, and replacing the concept
  $A \sqcap \exists r . (B_1 \sqcap B_2 \sqcap \exists r . \top)$ on
  the right-hand side of the concept inclusion in the second line of
  the definition of $\Omc'$ with some \EL-concept $C$ such that
  $a \in C^\Imc$ implies $\Imc \models q(a)$ for all interpretations
  \Imc.  Such a concept can be obtained by identifying all vertices
  that are reachable from vertex~1 in exactly $i$ steps, for each
  $i$. In the example case $k=2$, it takes the form
  $$
    A \sqcap B_1 \sqcap \exists r . (B_2 \sqcap B_3 \sqcap B_4 \sqcap
    \exists r . (B_5 \sqcap B_6 \sqcap B_7 \sqcap \exists r . B_8)).
  $$
  In the proof of the `if' direction of the claim,
  1-unravelings are replaced with $k$-unravelings. Note that $q$ does
  not admit a homomorphism to its $k$-unraveling: since $q$ has tree
  width at least $k+1$, such a homomorphism would have to be
  non-injective and thus the $k$-unraveling of $q$ would have to
  comprise atoms $B_i(x),B_j(x)$ for some $i,j$ with $i \neq j$.

  \medskip

  Now for Point~2 of Theorem~\ref{thm:mainhard}. The following was
  established in the proof of Theorem~48
  in~\cite{DBLP:conf/ijcai/Bienvenu0LW16}.
  \begin{theorem}
\label{lem:2expbasic}
Let $M=(Q,\Sigma,\Gamma,q_0,\Delta)$ be an exponentially space bounded
alternating Turing machine and let $k=|\Gamma|-1$. Given an input $w$
to $M$, one can construct in polynomial time a Boolean OMQ
$Q_w=(\Omc_w,\Sbf_w,q_w)$ from $(\ELI,\text{CQ})$ such that for a
selected concept name $A^* \notin \Sbf_w$, the following conditions
are satisfied:
\begin{enumerate}

\item $M$ accepts $w$ iff there is an $\Sbf_w$-database
  \Dmc and an $a \in \Ind(\Dmc)$ such that $\Dmc \models
  (\Omc_w,\Sbf_w,A^*(x))(a)$ and $\Dmc \not\models Q_w$;

\item $q_w$ is connected, uses only symbols from $\Sbf_w$, is of tree
  width $k$, and not equivalent to any CQ of smaller tree width;


\item the restriction of $\mn{ch}_{\Omc_w}(\Dmc_{q_w})$ to symbols in
  $\Sbf_w$ is $\Dmc_{q_w}$.

\end{enumerate}
\end{theorem}
We remark that tree width is not explicitly mentioned
in~\cite{DBLP:conf/ijcai/Bienvenu0LW16}. However, the CQ $q_w$
constructed there contains a subquery $p$ whose Gaifman graph contains
as a minor the complete balanced bipartite graph $K_{k+1,k+1}$ with
exactly one adjacent edge dropped from each vertex. It is not hard to
verify that such a graph has degeneracy at least $k$, thus tree width
at least $k$. Moreover, the Gaifman graph of $p$ is a core, thus $p$
is not equivalent to a CQ of tree width smaller than $k$.


Let $k \geq 1$.  We can choose $M=(Q,\Sigma,\Gamma,q_0,\Delta)$ to
have a \TwoExpTime-hard word problem and satisfying $|\Gamma| > k+1$.
We reduce the word problem for $M$, as follows. Let $w$ be an input to
$M$, and let $Q_w=(\Omc_w,\Sbf_w,q_w)$ be as in
Theorem~\ref{lem:2expbasic}.  Define $q(x)=A^*(x) \wedge q_w$ where
$x$ is a fresh variable and let $Q=(\Omc_w,\Sbf_w,q)$.
\\[2mm]
{\bf Claim} $M$ accepts $w$ iff $Q$ is not (U)CQ$_k$-equivalent (while
preserving the ontology or not).
\\[2mm]
For `only if', assume that $M$ accepts $w$. By Point~1 of
Theorem~\ref{lem:2expbasic}, there is an $\Sbf_w$-database $\Dmc_0$
and an $a \in \Ind(\Dmc_0)$ such that $\Dmc_0 \models
(\Omc_w,\Sbf_w,A^*(x))(a)$ and $\Dmc_0 \not\models Q_w$.  Let \Dmc be
the disjoint union of $\Dmc_0$ and $\Dmc_{q_w}$. Then $\Dmc \models
Q(a)$. Now assume to the contrary of what is to be shown that $Q$ is
UCQ$_k$-equivalent. Then $Q$ is equivalent to its
UCQ$_k$-approximation $Q_a=(\Omc_w,\Sbf_w,q_a)$.  Clearly, we can
remove from $q_a$ any CQ in which the variable $x$ from $q$ is
identified with any other variable, without compromising equivalence
or altering tree width. Thus each CQ in $q_a$ takes the form $A^*(x)
\wedge q'_w$ where $q'_w$ is a contraction of $q_w$. Since $q_w$ is
connected, so is $q'_w$. Let $Q'_a$ be the Boolean OMQ
$(\Omc_w,\Sbf_w,q'_a)$ where $q'_a$ consists of the $q'_w$-parts of the
CQs in $q_a$. Clearly, $Q'_a$ is just the UCQ$_k$-approximation of
$Q_w$. From $\Dmc \models Q(a)$ and $\Dmc_0 \not\models Q_w$, we thus
obtain $\Dmc \models Q_a$ and $\Dmc_0 \not\models Q'_a$.
Consequently, $\Dmc_{q_w} \models Q'_a$. Since $q_w$ and thus $q'_w$
uses only symbols from $\Sbf_w$ and by Point~3 of
Theorem~\ref{lem:2expbasic}, this implies that there is a
homomorphism from $q'_w$ to $q_w$. This means that $q'_w$ is
equivalent to $q_w$, but $q'_w$ is of tree width $k$, in contradiction
to Point~2 of Theorem~\ref{lem:2expbasic}.

\smallskip

For `if', assume that $M$ does not accept $w$. By Point~1 of
Theorem~\ref{lem:2expbasic}, $\Dmc \models (\Omc_w,\Sbf_w,A^*(x))(a)$
implies $\Dmc \models Q_w$ for any $\Sbf_w$-database \Dmc and $a \in
\Ind(\Dmc)$. Consequently, $Q$ is equivalent to
$(\Omc_w,\Sbf_w,A^*(x))$ and thus CQ$_1$-equivalent while preserving
the ontology.

\medskip

We finally address Point~3 of Theorem~\ref{thm:mainhard}. It is proved
in \cite{DBLP:conf/kr/BienvenuLW12} that containment in
$(\text{DL-Lite}^\Rmc,\text{CQ})$ is $\Pi^p_2$-complete.  However, it
is rather unclear how to utilize the hardness proof given in that
paper for our purposes. We instead give a direct proof by reduction
from $\forall\exists$-QBF, building on and extending an
\NPclass-hardness proof for the combined complexity of (a restricted
version of) query evaluation in $(\text{DL-Lite}^\Rmc,\text{CQ})$
given in \cite{DBLP:conf/dlog/KikotKZ11}. To make the presentation
more digestible, we first show that containment between unary OMQs
$(\Omc,\Sbf,C(x))$ and $(\Omc,\Sbf,q(x))$, where \Omc is a
$\text{DL-Lite}^\Rmc$-ontology and $C$ a conjunction of concept names,
is $\Pi^p_2$-hard.

Thus let $\vp=\forall \xbf \exists \ybf \,(C_1 \vee \cdots \vee
C_\ell)$ be a sentence of $\forall\exists$-QBF where each
$C_1,\dots,C_\ell$ is a clause. Further let $\xbf=x_1 \cdots x_n$ and
$\ybf = y_1 \cdots y_m$.  Let the
$\text{DL-Lite}^\Rmc$-ontology \Omc contain the following:
$$
\begin{array}{rcll}
  T_i &\sqsubseteq& X_i & \text{for } 1 \leq i \leq n\\[.5mm]
  F_i &\sqsubseteq& X_i & \text{for } 1 \leq i \leq n\\[.5mm]
  T_i &\sqsubseteq& C_j & \text{if } x_i \in C_j\\[.5mm]
  F_i &\sqsubseteq& C_i & \text{if } \neg x_i \in C_j\\[.5mm]
  L_i & \sqsubseteq& \exists r . (L_{i+1} \sqcap Y_{i+1}) & \text{for }
  0 \leq i \leq m\\[.5mm]
  L_i & \sqsubseteq& \exists r . (L_{i+1} \sqcap \overline{Y}_{i+1}) & \text{for }
  0 \leq i \leq m\\[.5mm]
  Y_i & \sqsubseteq& \exists r^- . C^+_j & \text{if } y_i \in C_j\\[.5mm]
  \overline{Y}_i & \sqsubseteq& \exists r^- . C^+_j & \text{if } \neg
  y_i \in C_j\\[.5mm]
  C^+_i &\sqsubseteq& C_i & \text{for } 1 \leq i \leq \ell \\[.5mm]
  C^+_i &\sqsubseteq& \exists r^- . C^+_i& \text{for } 1 \leq i \leq \ell
\end{array}
$$
where, as usual in $\text{DL-Lite}^\Rmc$, $A \sqsubseteq \exists r
. (B_1 \sqcap \cdots \sqcap B_n)$ abbreviates $A \sqsubseteq \exists
r_0 . \top, r_0 \sqsubseteq r, \exists r_0^- . \top \sqsubseteq B_1,\dots,
\exists r_0^- . \top \sqsubseteq B_n$, where $r_0$ is a fresh role name.
Set
$$
  \Sbf =\{L_0,T_1,\dots,T_n,F_1,\dots,F_n\}
$$
and define the CQ $q(x_0)$ to be as follows, all variables
except~$x_0$ existentially quantified:
$$
\begin{array}{l}
  r(x_0,x_1) \wedge \cdots \wedge r(x_{m-1},x_m) \; \wedge\\[1mm]
  C_1(z_0^1) \wedge r(z_0^1,z_1^1) \wedge \cdots \wedge
  r(z_{m-1}^1,x_m) \; \wedge\\[1mm]
  \hspace*{3mm}\qquad\qquad\qquad\qquad \cdots \\[1mm]
  C_\ell(z_0^{\ell}) \wedge r(z_0^{\ell},z_1^{\ell}) \wedge \cdots \wedge
  r(z_{m-1}^{\ell},x_m).
\end{array}
$$
Let $Q_1=(\Omc,\Sbf,X_1(x)\wedge\cdots\wedge X_n(x))$ and
$Q_2=(\Omc,\Sbf,q)$.  Then we have the following, which implies
$\Pi^p_2$-hardness of the containment question mentioned above.
\\[2mm]
{\bf Claim 1.} $Q_1 \subseteq Q_2$ iff $\vp$ is true.
\\[2mm]
\emph{Proof of claim}. For the `if' direction, assume that
$Q_1 \subseteq Q_2$. Let $\pi$ be a truth assignment for the variables
\xbf. We have to show that $\pi$ can be extended to the variables in
\ybf such that $C_1 \vee \cdots \vee C_\ell$ is satisfied.  Let
$\Dmc_\pi$ be the \Sbf-database that contains the fact $T_i(a)$ if
$\pi(x_i)=1$ and $F_i(a)$ if $\pi(x_i)=0$. Clearly,
$\Dmc_\pi \models Q_1(a)$, and thus $\Dmc_\pi \models Q_2(a)$, that
is, there is a homomorphism $h$ from $q$ to $\mn{ch}_\Omc(\Dmc_\pi)$
with $h(x_0)=a$. It is easy to see that $\mn{ch}_\Omc(\Dmc_\pi)$ takes
the form of a binary tree of depth $m$ with root $a$ in which every
path $p$ corresponds to a truth assignment $\pi_p$ to the variables in
\ybf, and vice versa: the node on level $i$ is labeled with (exactly
one of) $Y_i$ or $\overline{Y}_i$ and $\pi_p(y_i)=1$ iff the former is
the case. By construction of $q$, the variable $x_m$ of $q$ must be
mapped to the final node of a path $p$, and we extend $\pi$ with
$\pi_p$. It can be verfied that the use of the concept names $C_j$ in
\Omc and $q$ imply that $\pi$ satisfies all clauses from the QBF.

\smallskip For the `only if' direction, assume that $\vp$ is true. Let
\Dmc be an $\Sbf$-database with $\Dmc \models Q_1(a)$. We produce a
homomorphism $h$ from $q$ to $\mn{ch}_\Omc(\Dmc)$ such that
$h(x_0)=a$. Since $\Dmc \models Q_1(a)$ and
$X_1,\dots,X_\ell \notin \Sbf$, we must have $T_i(a)$ or $F_i(a)$ (or
both) in \Dmc for $1 \leq i \leq n$. Let $\pi_\Dmc$ be a truth
assignment such that $\pi_\Dmc(x_i)=1$ implies $T_i(a) \in \Dmc$ and
$\pi_\Dmc(x_i)=0$ implies $F_i(a) \in \Dmc$. Since $\vp$ is true, we
can extend $\pi_\Dmc$ to a truth assignment to \ybf such that all
clauses are satisfied. Thus truth assignment identifies a path in
$\mn{ch}_\Omc(\Dmc)$ in the subtree database rooted at $a$. Then
$h$ can map the variables $x_0,\dots,x_m$ from $q$ to that path.
Since all clauses are satisfied, the remaining paths in $q$ can also
be mapped.

We next modify the $\Pi^p_2$-hardness proof just given so that it
applies to UCQ$_k$-equivalence in
$(\text{DL-Lite}^\Rmc_{\mn{horn}},\text{CQ})$.  Fix some $k \geq 1$.
We would like to reuse essentially the same CQ as before, but now we
have to make sure that it is of tree width exceeding~$k$. To achieve
this, we mix in the CQ from the proof of Point~1 of
Theorem~\ref{thm:mainhard} that we had obtained by starting with the
$k+2$-clique based on vertices $\{1,\dots,k+2\}$, orienting each edge
$\{i,j\}$ in the direction from $i$ to $j$ if $i<j$, subdividing each
edge $(i,j)$ into $j-i$ edges by introducing intermediate points,
representing directed edges as $r$-atoms, and finally labeling each
vertex with a different concept name $B_i$, $i \geq 1$. All variables
are unquantified, that is, the resulting CQ $p$ is just a set of
atoms.
%
By identifying all vertices that are reachable from vertex~1 in
exactly $i$ steps, for each $i$, we obtain a path-shaped contraction
of $p$ that can be represented as an \EL-concept
$$
  C=C_1 \sqcap \exists r . (C_2 \sqcap \exists r . (C_3 \sqcap \cdots
  \sqcap \exists r . C_{k+2} ) \cdots )
$$
such that $a \in C^\Imc$ implies $\Imc \models p(a)$ for all
interpretations~\Imc. Let \Omc be the ontology from the
previous reduction extended by the following concept
inclusions:
$$
\begin{array}{rcll}
  I_i & \sqsubseteq & C_i \sqcap \exists r . I_{i+1} & \text{for } 1
                                                       \leq i \leq k+1
  \\[.5mm]
  I_{k+2} & \sqsubseteq & L_0 \\[.5mm]
  B & \sqsubseteq & C_i &\text{for } 1 \leq i \leq \ell \\[.5mm]
  B &\sqsubseteq& \exists r . B \\[.5mm]
  B &\sqsubseteq& \exists r^- . B
\end{array}
$$
Let $\Sbf$ consist of $\{ I_1,T_1,\dots,T_n,F_1,\dots,F_n,r,B\}$ and all
concept names $B_i$ introduced in the construction of the CQ
$p$. Assume that the variable in $p$ that corresponds to vertex 1 from
the original clique is $x_0$ and the vertex that corresponds to vertex
$k+2$ is $x_{k+1}$. Define the CQ $q(x_0)$ to be as follows, all
variables except~$x_0$ existentially quantified, $w=m+k+2$:
$$
\begin{array}{l}
  I_1(x_0) \wedge X_1(x_0) \wedge \cdots \wedge X_n(x_0) \;
  \wedge\\[1mm]
  p \; \wedge\\[1mm]
  r(x_{k+1},x_{k+2}) \wedge \cdots \wedge r(x_{w-1},x_w) \; \wedge\\[1mm]
  C_1(z_0^1) \wedge r(z_0^1,z_1^1) \wedge \cdots \wedge
  r(z_{w-1}^1,x_w) \; \wedge\\[1mm]
  \hspace*{3mm}\qquad\qquad\qquad\qquad \cdots \\[1mm]
  C_\ell(z_0^{\ell}) \wedge r(z_0^{\ell},z_1^{\ell}) \wedge \cdots \wedge
  r(z_{w-1}^{\ell},x_w).
\end{array}
$$
\\[2mm]
{\bf Claim 2.} $Q=(\Omc,\Sbf,q)$ is UCQ$_k$-equivalent iff $\varphi$ is
true.
\\[2mm]
In fact, it can be verified that $Q$ is equivalent to $(\Omc,\Sbf,q')$
if $\varphi$ is true, where $q'$ is the CQ that consists of the first
line of the definition of $q$. Conversely, if $\varphi$ is false, then
consider the following CQ, viewed as a database \Dmc:
$$
I_1(x_0) \wedge X_1(x_0) \wedge \cdots \wedge X_n(x_0) \; \wedge p
\wedge B(x_{k+1}).
$$
It is easy to see that \Dmc has tree width $k+1$ and that
$\Dmc \models Q(x_0)$, $x_0$ meaning the constant of the same name in
\Dmc here. It can also be verified that the $k$-unraveling $\Dmc'$ of
\Dmc is such that $\Dmc' \not\models Q(x_0)$. Together with
Corollary~\ref{prop:wasaclaim} and Theorem~\ref{thm:bestapprox},
this implies that $Q$ is not UCQ$_k$-equivalent.
\end{proof}


In the following proof, we are going to make use of complexity
results for OMQ containment from the literature. Recall that, in this
paper, $Q_1 \subseteq Q_2$ with $Q_i = (\Omc_i,\Sbf,q_i)$ if
$Q_1(\Dmc) \subseteq Q_2(\Dmc)$ for all \Sbf-databases \Dmc including
those that are inconsistent with $\Omc_1$ or $\Omc_2$. In the
literature on containment, in contrast, it is common to consider only
those \Dmc that are consistent with both $\Omc_1$ and $\Omc_2$.  We
refer to this as \emph{consistent containment}. In the proof that
follows, the actual queries are UCQs and $\Omc_1$ and
$\Omc_2$ are the same ontology \Omc. In this case, the gap
between containment and consistent containment is unproblematic since
we can reduce containment to consistent containment in polynomial
time, as follows. Let $\Omc'$ be obtained from \Omc by 
\begin{enumerate}

\item replacing every CI $C \sqsubseteq \bot$ with $C \sqsubseteq B$

\item adding $\top \sqsubseteq A$ 

\end{enumerate}
where $A$ and $B$ are fresh concept names. Moreover, if $\xbf=x_1
\cdots x_n$ are the answer variables in $q_1$ and $q_2$ (which we can
w.l.o.g.\ assume to be identical), then let $q'_i$ be obtained from
$q_i$ by adding as an additional disjunct the CQ $A(x_1) \wedge \cdots
\wedge A(x_n) \wedge \exists y \, B(y)$, for $i \in \{1,2\}$. It can
be shown that $Q_1 \subseteq Q_2$ if $(\Omc',\Sbf,q'_1)$ is
consistently contained in $(\Omc',\Sbf,q'_2)$. A crucial observation
is that every $\Sbf$-database is consistent with $\Omc'$.

\thmdecupper*
\begin{proof}\ Points~1 and~2 are proved in a uniform way.  By
  Corollary~\ref{prop:wasaclaim}, it suffices to construct the
  UCQ$_k$-approximation $Q_a=(\Omc,\Sbf,q_a)$ of the input query
  $Q=(\Omc,\Sbf,q)$, and check whether $Q \subseteq Q_a$ (the converse
  containment holds by construction of $Q_a$).  An approach based on
  alternating tree automata has been used in
  \cite{DBLP:conf/ijcai/Bienvenu0LW16} to show that OMQ containment is
  \ExpTime-complete in $(\ELHbot,\text{CQ})$ and \TwoExpTime-complete
  in $(\ELHIbot,\text{CQ})$. It is an easy exercise, and does not
  require any new ideas, to extend this approach from CQs to UCQs, and
  from \ELHbot to \ELpoly. Applying the resulting decision procedures
  for containment as a black box, we obtain a \TwoExpTime upper bound
  for $(\ELpoly,\text{UCQ})$ and a {3\sc ExpTime} upper bound for
  $(\ELHIbot,\text{UCQ})$. To lower these bounds by one exponential,
  we have to address the fact that $q_a$ has exponentially many
  disjuncts (each of polynomial size). This requires another minor
  change in the decision procedure for containment, exploiting that
  every collapsing of a CQ in $q_a$ is also a collapsing of $q$.  We
  give more details in what follows.

  The central relevant statement from
  \cite{DBLP:conf/ijcai/Bienvenu0LW16} is as follows.
\begin{theorem}
\label{thm:fromcontainment}
  For every OMQ $Q=(\Omc,\Sigma,q)$ from
  $(\mathcal{ELHI}_\bot,\text{CQ})$ with $q$ Boolean, there is a
  two-way alternating parity tree automaton $\Amf_Q$ that accepts a
  $(|\Omc| \cdot |q|)$-ary $\Sigma_\varepsilon \cup \Sigma_N$-labeled
  tree $(T, L)$ iff it is proper, $\Dmc_{(T, L)}$ is consistent with
  \Omc\!\!\!, and $\Dmc_{(T, L)} \models Q$.  $\Amf_Q$ has at most
  $2^{p(|q|+\mn{log}(|\Omc|))}$ states, and at most $p(|q|+|\Omc|)$
  states if \Omc is an \ELpoly-ontology, $p$ a polynomial.
  It can be constructed in time polynomial in its size.
\end{theorem}
Here, $\Sigma_\varepsilon \cup \Sigma_N$ are suitable alphabets such
that, among other things, a $\Sigma_\varepsilon \cup \Sigma_N$-labeled
tree $(T, L)$ represents an (almost) tree-shaped database $\Dmc_{(T, L)}$.
The term `proper' refers to a technical condition that need not bother
us here.  The construction of the automaton $\Amf_Q$ from
Theorem~\ref{thm:fromcontainment} relies on the notion of forest
decompositions, which partitions a query into a center part and
several tree-shaped parts.

A \emph{forest decomposition} of $q$ is a tuple $F=(q_0, q_1,x_1,
\ldots, q_k,x_k, \mu)$ where $(q_0, q_1, \ldots, q_k)$ is a partition
of (the atoms of) a contraction of $q$, $x_1,\dots,x_k$ are variables
from $q_0$, and $\mu$ is a mapping from $\mn{var}(q_0)$ to a fixed set
of constants $\Ind_0$ such that the following conditions are satisfied
for $1 \leq i,j \leq k$;
\begin{enumerate}
  \item $q_0$ is non-empty;

  \item $q_i$ is weakly tree-shaped with root $x_i$, that is, $G_q$ is
    a tree (multi-edges allowed);

  \item $\mn{var}(q_i) \cap \mn{var}(q_0)=\{ x_i \}$;

  \item $\mn{var}(q_i) \cap \mn{var}(q_j) \subseteq
  \mn{var}(q_0)$ if $i \neq j$;

  \item $q_i$ contains no atom $A(x_i)$;

  \item $x_i$ has a single successor in $q_i$.

\end{enumerate}
The central property of forest decompositions is then as follows.
\begin{lemma}\label{forest-decomp-lemma}
  Let \Omc be an $\mathcal{ELHI}_\bot$-ontology, $(T,L)$ a proper
  $\Sigma_\varepsilon \cup \Sigma_N$-labeled tree, \Cmc the part of
  $\Dmc_{(T,L)}$ represented by the root vertex of $T$, and $q$ a
  Boolean connected CQ.  Then the following are equivalent:
  \begin{enumerate}

  \item there is a homomorphism $h$ from $q$ to
    $\mn{ch}_\Omc(\Dmc_{(T,L)})$ whose range has a non-empty
    intersection with $\Ind(\Cmc)$;

  \item there is a forest decomposition $F=(q_0, q_1,x_1, \ldots,$
    $q_k,x_k,\mu)$ of $q$ such that
  \begin{itemize}
  \item $\mu$ is a homomorphism from $q_0$ to
    $\mn{ch}_\Omc(\Dmc_{(T,L)})$ whose range falls within
    $\Ind(\Cmc)$;
  \item there is a homomorphism $h_i$ from $q_i$ to
    $\mn{ch}_\Omc(\Dmc_{(T,L)})$ such that $h_i(x_i)=\mu(x_i)$, for $1
    \leq i \leq k$.
  \end{itemize}
  \end{enumerate}
\end{lemma}
In the construction of $\Amf_Q$, the transition relation contains a
disjunction over all forest decompositions of the input query~$q$.
We, however, are not interested in the CQ $q_a$ rather than in $q$.
But this is easy to achieve: instead of using all forest
decompositions of $q$ in the mentioned disjunction, we use all forest
decompositions of a CQ from $q_a$. Because of the use of contractions
in the definition of forest decompositions, each such decomposition is
also a forest decomposition of $q$, and consequently no further
modifications of the construction are required.

\smallskip For Point~3, we again have to check whether $Q \subseteq
Q_a=(\Omc,\Sbf,q_a)$. It was shown in \cite{DBLP:conf/kr/BienvenuLW12}
that containment in $(\text{DL-Lite}_{\mn{horn}},\text{CQ})$ is
$\Pi^p_2$-complete and the proof extends to
$(\text{DL-Lite}^\Rmc_{\mn{horn}},\text{UCQ})$. It thus again remains
to deal with the fact that $q_a$ consists of exponentially many CQs
(of polynomial size). We sketch the proof of a $\Sigma^p_2$ upper
bound for checking non-containment, that is $Q \not\subseteq Q_a$.  We
will make use of the fact that OMQ evaluation in
$(\text{DL-Lite}^\Rmc_{\mn{horn}},\text{UCQ})$ is
\NPclass-complete~\cite{CGLLR07,CombinedApproachRSA:IJCAI15}.

A pair $(\Dmc,\abf)$ with \Dmc an \Sbf-database and \abf a tuple over
$\Ind(\Dmc)$ is a \emph{witness} for $Q \not\subseteq Q_a$ if $\Dmc
\models Q(\abf)$ and $\Dmc \not\models Q_a(\abf)$. It is observed in
\cite{DBLP:conf/kr/BienvenuLW12} that it suffices to consider
witnesses $(\Dmc,\abf)$ where the number of constants in \Dmc is
bounded by $|q| \cdot (|\Sigma|+1)$. To decide whether $Q
\not\subseteq Q_a$, we guess a witness $(\Dmc,\abf)$ of such
dimension. We then verify in \NPclass that $\Dmc \models Q(\abf)$,
co-guess a CQ $p$ from $q_a$, and verify in \coNP that $\Dmc
\not\models (\Omc,\Sbf,p)(\abf)$. With co-guessing a CQ $p'$ in $q_a$,
we mean to co-guess an equivalence relation on $q$ that represents
variable identifications, to then produce $p'$ in polynomial time, and
to verify in polynomial time that it has tree width $k$ (note that $k$
is fixed). The overall algorithm clearly runs in $\Sigma^p_2$, as desired.
\end{proof}


\thmnpcompl*
\begin{proof}\ The \NPclass lower bounds are inherited from the case where the ontology is empty~\cite{DKV02}, while the \ExpTime lower bound is proved by a reduction from the subsumption problem in \ELI, namely given an
\ELI-ontology \Omc and concept names $A,B$, is $A$ subsumed by $B$
w.r.t.\ \Omc (written $\Omc \models A \sqsubseteq B$), that is, is
$A^\Imc \subseteq B^\Imc$ in every model \Imc of~\Omc? This problem is
known to be \ExpTime-complete \cite{DBLP:conf/owled/BaaderLB08}. We
start with the case $k=1$. Thus, let \Omc, $A$, and $B$ be as
stated. Define
  $$
  \begin{array}{rcl}
    \Omc' &=& \Omc \cup \{ B \sqsubseteq \exists r . (B_1
              \sqcap B_2 \sqcap \exists r . \top) \} \\[1mm]
    q(x) &=& \exists y_1 \exists y_2 \exists z \, A(x) \wedge r(x,y_1) \wedge
             r(x,y_2) \wedge B_1(y_1) \; \wedge \\[1mm]
             && \hspace*{14mm} B_2(y_2) \wedge r(y_1,z)
             \wedge r(y_2,z)
  \end{array}
  $$
  where $r$ is a fresh role name, and set
  $Q=(\Omc',\Sbf_{\mn{full}},q)$.  Notice the similarity of this
  construction to the proof of Point~1 of Theorem~\ref{thm:mainhard}.
  \\[2mm]
  {\bf Claim.} $\Omc \models A \sqsubseteq B$ iff $Q$ is
  (U)CQ$_1$-equivalent.
  \\[2mm]
  For the `if' direction, it suffices to note that when
  $\Omc \not\models A \sqsubseteq B$ then $Q$ is a full rewriting of
  itself. For the `only if' direction, from
  $\Omc \models A \sqsubseteq B$ it follows that
  $Q=(\Omc',\Sbf_{\mn{full}},A(x))$ is a full rewriting of $Q$. The
  generalization to the case $k>1$ is as in the proof of Point~1 of
  Theorem~\ref{thm:mainhard}, details are omitted.

\medskip


Let us focus on the upper bounds. We first argue that instead of proving the
results for $(\ELpoly,\text{UCQ})$, $(\ELHIbot,\text{UCQ})$, and $(\text{DL-Lite}_{\mn{horn}}^\Rmc,\text{UCQ})$, it suffices to
establish them for the corresponding OMQ languages based on CQs.  In
fact, we can assume w.l.o.g.\ that, when an OMQ
$Q=(\Omc,\Sbf_{\mn{full}},q)$ is given as input and $q(\xbf)= \bigvee_{1\leq i \leq n} p_i$, then the OMQs $Q_i=(\Omc,\Sbf_{\mn{full}},p_i)$, $1 \leq i \leq n$, are pairwise incomparable regarding containment. The reason is that, when the schema is full, OMQ containment trivially reduces to OMQ evaluation, which means that containment in $(\ELpoly,\text{CQ})$ and $(\text{DL-Lite}_{\mn{horn}}^\Rmc,\text{CQ})$ is in \NPclass, and in \ExpTime in $(\ELHIbot,\text{CQ})$
\cite{DBLP:conf/dlog/KikotKZ11,LutzTomanWolter-IJCAI09,DBLP:conf/kr/OrtizRS10}. We
can thus remove a disjunct $p_i$ from $q$ if there is a $p_j$, $j <
i$, such that $Q_j \subseteq Q_i$. We can also assume that
$\Dmc_{p_i}$ is consistent with \Omc, for $1 \leq i \leq n$, since
otherwise we can remove the disjunct $p_i$ and if all disjuncts are
removed then $Q$ is trivially UCQ$_k$-equivalent; moreover, checking
consistency of $\Dmc_{p_i}$ with \Omc also reduces easily to OMQ
evaluation.
  \\[2mm]
  {\bf Claim.} $Q$ is UCQ$_k$-equivalent iff every $Q_i$ is.
  \\[2mm]
  \emph{Proof of claim.} For the non-trivial `only if' direction,
  assume that $Q$ is UCQ$_k$-equivalent and let
  $Q'=(\Omc,\Sbf_{\mn{full}},q')$ be an equivalent OMQ with $q'$ from
  UCQ$_k$. Consider some $Q_i$. Clearly, $Q_i \subseteq Q'$ implies
  $\Dmc_{p_i} \models Q'(\xbf)$ where the answer variables \xbf of
  $p_i$ are viewed as constants. But then $\Dmc_{p_i} \models Q''$
  where $Q''= (\Omc,\Sbf_{\mn{full}},p')$ for some CQ $p'$ in $q'$,
  and thus $Q_i \subseteq Q''$.  Since $Q'' \subseteq Q$, we can argue
  analogously that $Q'' \subseteq Q_j$ for some $j$. Thus $Q_i
  \subseteq Q_j$ implies that $i=j$ since otherwise $Q_i$ and $Q_j$
  would be comparable regarding containment. But then $Q_i$ is
  equivalent to $Q''$ and thus CQ$_k$-equivalent. This finishes the
  proof of the claim.

   It thus suffices to prove upper bounds for UCQ$_k$-equivalence in
  $(\ELpoly,\text{CQ})$, $(\ELHIbot,\text{CQ})$, and $(\text{DL-Lite}_{\mn{horn}}^\Rmc,\text{CQ})$. All these bounds
  are established in a uniform way.  Assume that the OMQ
  $Q=(\Omc,\Sbf_{\mn{full}},q)$ is given where $q$ is a CQ. We first
  check whether $Q$ is empty (using a containment check) and if it is
  then we return that $Q$ is UCQ$_k$-equivalent. Otherwise, $\Dmc_q$
  must be consistent with \Omc. We then extend $q$ to a CQ $q'$ as
  follows, paralelling Step~2 in the construction of rewritings: for
  each $C \sqsubseteq D \in \Omc$ and $x \in C^{\Dmc_q}$, add a fresh
  copy $q_C$ of $C$ viewed as a CQ and add $q_C$ to $q$, identifying
  $x$ with the root of $q_C$. We then guess a subquery $q''$ of $q'$
  of tree width at most $k$ and check whether $Q$ is equivalent to
  $(\Omc,\Sbf_{\mn{full}},q'')$. Equivalence can be implemented as two
  containment checks; see above for the relevant complexities.

  If we are able to guess correctly, then clearly $Q$ is
  UCQ$_k$-equivalent. Conversely, if $Q$ is UCQ$_k$-equivalent, then
  by Theorem~\ref{thm:fulldatabase} there is a full rewriting
  $Q'=(\Omc,\Sbf_{\mn{full}},p)$ of $Q$ with $p \in \text{CQ}_k$. It is
  easy to verify that $p$ is a subquery of $q'$.

  It can be verified that in all the considered cases, the above procedure yields the stated upper bounds.
  \end{proof}
%

  %

\section*{Proofs for Section~\ref{sec:dl-lite-f}}
\label{app:dl-lite-f}

The result that we can immediately inherit from~\cite{Figueira16} is the following that talks about BCQs:

\begin{theorem}[Figueira]\label{pro:dl-lite-f-ucq-k-equiv-bcqs}
  For OMQs from $(\text{DL-Lite}^{\Fmc}_{=},\text{BCQ})$ based on the full schema, BCQ$_k$-equivalence while preserving the ontology is in~\TwoExpTime, for any $k \geq 1$.
  Moreover, an equivalent OMQ from $(\text{DL-Lite}^{\Fmc}_{=},\text{BCQ}_k)$ can be constructed in double exponential time (if it exists).
\end{theorem}

We need to show that the above result can be stated for UBCQs (Theorem~\ref{pro:dl-lite-f-ucq-k-equiv} in the main body of the paper). To this end, we establish the following technical result.

\begin{lemma}\label{lem:bcqs-to-ubcqs}
Consider an OMQ $Q = (\Omc,\Sbf_{\mn{full}},q)$ from $(\text{DL-Lite}^\Fmc,\text{UBCQ})$. The following are equivalent:
\begin{enumerate}
\item $Q$ is UBCQ$_k$-equivalent while preserving the ontology.
\item For each $q'$ in $q$, (i) $(\Omc,\Sbf_{\mn{full}},q')$ is BCQ$_k$-equivalent while preserving the ontology, or (ii) there exists $q''$ in $q$ such that $(\Omc,\Sbf_{\mn{full}},q') \subseteq (\Omc,\Sbf_{\mn{full}},q'')$.
\end{enumerate}
\end{lemma}

\begin{proof}
$(1) \Rightarrow (2)$. By hypothesis, there exists a UBCQ$_k$ $\hat{q}$ such that $(\Omc,\Sbf_{\mn{full}},q) \equiv (\Omc,\Sbf_{\mn{full}},\hat{q})$. Consider an arbitrary BCQ $q'$ in $q$. Since $(\Omc,\Sbf_{\mn{full}},q) \subseteq (\Omc,\Sbf_{\mn{full}},\hat{q})$, we get that there exists $p$ in $\hat{q}$ such that $(\Omc,\Sbf_{\mn{full}},q') \subseteq (\Omc,\Sbf_{\mn{full}},p)$. But since $(\Omc,\Sbf_{\mn{full}},\hat{q}) \subseteq (\Omc,\Sbf_{\mn{full}},q)$, we get that there exists $p'$ in $q$ such that $(\Omc,\Sbf_{\mn{full}},p) \subseteq (\Omc,\Sbf_{\mn{full}},p')$. Therefore,
\[
(\Omc,\Sbf_{\mn{full}},q')\ \subseteq\ (\Omc,\Sbf_{\mn{full}},p)\ \subseteq\ (\Omc,\Sbf_{\mn{full}},p'),
\]
where $p$ belongs to $\hat{q}$ and $p'$ belongs to $q$. We consider two cases: $q' = p'$, which implies that $(\Omc,\Sbf_{\mn{full}},q') \equiv (\Omc,\Sbf_{\mn{full}},p)$, and condition (i) holds since $p$ is from BCQ$_k$, and $q' \neq p'$ but since $(\Omc,\Sbf_{\mn{full}},q') \subseteq (\Omc,\Sbf_{\mn{full}},p')$ condition (ii) holds.

$(2) \Rightarrow (1)$. This direction is clear.
\end{proof}

It is easy to verify that Theorem~\ref{pro:dl-lite-f-ucq-k-equiv} follows from Theorem~\ref{pro:dl-lite-f-ucq-k-equiv-bcqs} and Lemma~\ref{lem:bcqs-to-ubcqs}.

\lemrewritingucqequiv*

For showing the above lemma, we first need to establish the following technical lemma:

\begin{lemma}\label{lem:rewriting-eval}
  Let $Q = (\Omc,\Sbf_{\mn{full}},q) \in
  (\text{DL-Lite}^\Fmc,\text{UBCQ})$. Then there is a UBCQ
  $\mn{rew}(Q)$ such that for every database $\Dmc$ that
  satisfies~$\Omc^{=}$, $\Dmc \models Q$ iff $\Dmc \models
  \mathsf{rew}(Q)$. Furthermore, if $q \in \text{UBCQ}_k$ for some $k
  \geq 1$, then $\mathsf{rew}(Q) \in \text{UBCQ}_k$.
\end{lemma}
\begin{proof}
  We start with introducing some useful notions. Let $q$ be a BCQ. A
  BCQ $p \subseteq q$ is a \emph{tree in $q$ with root $x$} if $x\in
  \mn{var}(q)$ is an articulation point that separates $q$ into
  components $q',p$ with $p$ of the form $r(x,y) \wedge \varphi(\ybf)$
  where $\varphi(\ybf)$ is a tree with root $y$. Let $\mn{at}$ be an atom of the form $A(x)$, $S(x,z)$, or $S(z,x)$ with $z$ a fresh variable. For an ontology \Omc, we say that \mn{at} \emph{\Omc-generates} $p$ if $\{ \mn{at} \},\Omc \models p(x)$ where the variables in
  $\mn{at}$, including $x$, are viewed as constants in the database
  $\{ \mn{at} \}$. Moreover, $p$ is \emph{\Omc-generatable} if there
  is an atom that \Omc-generates $p$. Further, $\mn{at}$
  \emph{detachedly \Omc-generates} a BCQ $p$ if $\{ \mn{at} \},\Omc
  \models p$ and $p$ is \emph{detached \Omc-generatable} if there is
  an atom that detachedly \Omc-generates $p$.

  Now let $Q = (\Omc,\Sbf_{\mn{full}},q) \in
  (\text{DL-Lite}^\Fmc,\text{UBCQ})$ as in the lemma. We define
  $\mn{rew}(Q)$ to be the disjunction of all BCQs that can be obtained
  in the following way:
  \begin{enumerate}

  \item choose a BCQ $p$ in $q$, a contraction $p'$ of $p$, and a set
    $S$ of trees in $p'$ 
    and are \Omc-generatable, and remove those
    trees from $p'$;

  \item if for the resulting BCQ $p''$, $G_{p''}$ is not a minor of
    $G_p$, then stop;

  \item otherwise, for each $p_t \in S$ choose an atom $\mn{at}$ that
    \Omc-generates $p_t$ and add $\mn{at}$;

  \item choose a set of maximal connected components of the
    resulting BCQ that are detachedly \Omc-generatable, for each
    such component choose an atom $\mn{at}$ that detachedly
    \Omc-generates it, and add $\mn{at}$.

  \end{enumerate}
  With this definition of $\mn{rew}(Q)$, the ``Furthermore'' part of
  the lemma is trivially satisfied since the class of structures of
  tree width $k$ is minor closed \cite{GroheBook} and the additional
  modifications in Steps~3 and~4 clearly cannot increase tree
  width. It thus remains to show the following.
  \\[2mm]
  {\bf Claim.} For every database $\Dmc$ that satisfies~$\Omc^{=}$, $\Dmc
  \models Q$ iff $\Dmc \models \mathsf{rew}(Q)$.
  \\[2mm]
  The `if' direction is easy to show using the construction of
  $\mathsf{rew}(Q)$, we omit details. For the `only if' direction, let
  \Dmc be a database that satisfies~$\Omc^{=}$ and such that $\Dmc
  \models Q$. Then there is a CQ $p$ in $q$ and a homomorphism $h$
  from $p$ to $\mn{ch}_{\Omc}(\Dmc)$. We show how to use $h$ to guide
  the choices in Steps~1 to~4 above so as to obtain a BCQ $\widehat q$
  in $\mn{rew}(Q)$ such that $h$ can be extended to a homomorphism
  from $\widehat q$ to $\mn{ch}_{\Omc}(\Dmc)$.

  Let $V_2$ denote the set of pairs $(x_1,x_2) \in \mn{var}(p)^2$ such
  that $h(x_1) = h(x_2) \in \Ind(\Dmc)$ and, additionally, $p$
  contains atoms
  %
  \begin{equation}
  r_1(y_1, y_2), \dots, r_{n-1}(y_{n-1}, y_n)
\tag{$\dagger$}
  \end{equation}
  %
  such that $y_0=x_1$, $y_n=x_1$, and $h(y_i) \notin \Ind(\Dmc)$ for
  \mbox{$1 < i < n$}. In Step~1, we choose as $p'$ the contraction of
  $p$ that is obtained by identifying $x_1$ and $x_2$ whenever
  $(x_1,x_2) \in V_2$. We further choose as $S$ be the set of all
  trees $p_t(x) = r(x,y) \wedge \varphi(\ybf)$ in $p'$ such that $h(x)
  \in \Ind(\Dmc)$ and $h(y) \notin \Ind(\Dmc)$. We must then have
  $h(z) \notin \Ind(\Dmc)$ for all $z \in \ybf$ as otherwise $x$ would
  have been identified with some variable in \ybf.  By the
  construction of $\mn{ch}_\Omc(\Dmc)$, there must thus be a fact $F$
  in \Dmc such that $\{ F \},\Dmc \models p_t(h(x))$. As a consequence
  of this and the semantics of DL-Lite, we find an atom $\mn{at}$ of
  the form $A(x)$, $S(x,z)$, or $S(z,x)$, with $z$ a fresh variable,
  such that $\mn{at}$ \Omc-generates $p_t$ and $h$ extends to a
  homomorphism from $\{\mn{at}\}$ to \Dmc. By the former, $p_t$ is
  \Omc-generatable.

  Let $p''$ be the result of removing all trees in $S$ from $p'$.  We
  next argue that $p''$ is a minor of $p$ and thus we do not stop in
  Step~2. We work with the definition of minors from \cite{GroheBook},
  that is, an undirected graph $G_1$ is a minor of an undirected graph
  $G_2$ if $G_1$ can be obtained from a (not necessarily induced)
  subgraph of $G_2$ by contracting edges. We start with the subquery
  $p^*$ of $p$ that is the result of
  \begin{enumerate}

  \item taking the restriction of $p$ to all variables $y$ such that
    $h(y) \in \Ind(\Dmc)$ or $y$ is in a maximal connected component
    of $p$ all of whose variables are mapped by $h$ to outside of
    $\Ind(\Dmc)$ and then

  \item adding back a path of the form $(\dagger)$ for all
    $(x_1,x_2) \in V_2$.

  \end{enumerate}
  Clearly, $G_{p^*}$ is a subgraph of $G_{p}$. We can obtain $p''$
  from $p^*$ by contracting all edges in $G_{p^*}$ that are induced by
  the paths that have been added back. Thus $p''$ is a minor of $p$.

  The choices in Steps~3 and~4 can also be guided by $h$. We have
  already argued that all $p_t \in S$ are \Omc-generatable, and that
  we can find replacing atoms $\mn{at}$ that are `compatible with
  $h$', which we choose in Step~3. For Step~4, we choose those maximal
  connected components all of whose variables $h$ maps to outside of
  $\Ind(\Dmc)$. By construction of $\mn{ch}_\Omc(\Dmc)$, for every
  such component there must be an atom $\mn{at}$ of one of the three
  relevant forms that detachedly \Omc-generates it and such that $h$
  extends to a homomorphism from $\{\mn{at}\}$ to \Dmc.

  Let $\widehat p$ be the CQ generated by guiding Steps~1 to~4 with
  $h$ as described above. By construction, we clearly find an
  extension $h'$ of $h$ to the fresh variables in $\widehat p$ such
  that $h'$ is a homomorphism from $\widehat p$ to \Dmc. Thus,
  $\Dmc \models \mn{rew}(Q)$.
\end{proof}

We can now give the proof of Lemma~\ref{lem:rewriting-ucq-equiv}.

\begin{proof}
\underline{$(1) \Rightarrow (2)$.} By hypothesis, there exists an OMQ $Q' = (\Omc,\Sbf_{\mn{full}},q')$ from $(\text{DL-Lite}^\Fmc,\text{UBCQ}_k)$ such that $Q \equiv Q'$. By Lemma~\ref{lem:rewriting-eval}, we can conclude that
\[
(\Omc^=,\Sbf_{\mn{full}},\mathsf{rew}(Q))\ \equiv\ (\Omc^=,\Sbf_{\mn{full}},\mathsf{rew}(Q')).
\]
By Lemma~\ref{lem:rewriting-eval}, $\mathsf{rew}(Q') \in \text{UBCQ}_k$, and (2) follows.

\medskip

\underline{$(2) \Rightarrow (1)$.} By hypothesis, there is $q' \in \text{UBCQ}_k$ such that 
\[
(\Omc^=,\Sbf_{\mn{full}},\mathsf{rew}(Q))\ \equiv\ (\Omc^=,\Sbf_{\mn{full}},q').
\]
By Lemma~\ref{lem:rewriting-eval}, $Q \equiv (\Omc^=,\Sbf_{\mn{full}},\mathsf{rew}(Q))$. It is not difficult to show that $(\Omc^=,\Sbf_{\mn{full}},\mathsf{rew}(Q)) \equiv (\Omc,\Sbf_{\mn{full}},\mathsf{rew}(Q))$, which implies that $Q \equiv (\Omc,\Sbf_{\mn{full}},\mathsf{rew}(Q))$.
It is also easy to verify that $(\Omc,\Sbf_{\mn{full}},\mathsf{rew}(Q)) \equiv (\Omc,\Sbf_{\mn{full}},q')$. Therefore, $Q \equiv (\Omc,\Sbf_{\mn{full}},q')$, and (1) follows since $q' \in \text{UBCQ}_k$.
\end{proof}

\thedllitefacyclic*


Before giving the proof of the above result, we first present the reduction of UBCQ$_1$-equivalence in $(\text{DL-Lite}^\Fmc,\text{UBCQ})$ to UBCQ$_1$-equivalence in $(\text{DL-Lite},\text{UBCQ})$ announced in the paper, where $\text{DL-Lite}$ denotes the DL obtained from $\text{DL-Lite}^\Fmc$ after dropping the functionality assertions.

%
%
For a UBCQ $q$, we use $\mn{id}_F(q)$ to denote the result of
contracting each BCQ in $q$ in a minimal way such that the
functionality assertions in $F$ are respected. It is important to
observe that if $q$ is of tree width~1, then so is $\mn{id}_F(q)$.
This is not the case for any higher tree width.

\begin{lemma}
\label{lem:dlliteFindep}
Let $Q=(\Omc,\Sbf_{\mn{full}},q)$ be an OMQ from
$(\text{DL-Lite}^\Fmc,\text{UBCQ})$. On databases that are consistent
with~$\Omc^=$, $Q$ is equivalent to any of the following:
$(\Omc^\sqsubseteq,\Sbf_{\mn{full}},q)$,
$(\Omc,\Sbf_{\mn{full}},\mn{id}_{\Omc^=}(q))$,
$(\Omc^\sqsubseteq,\Sbf_{\mn{full}},\mn{id}_{\Omc^=}(q))$.
\end{lemma}

\begin{lemma}
\label{lem:dlliteFOrewr}
Let $Q=(\Omc,\Sbf_{\mn{full}},q)$ be an OMQ from
$(\text{DL-Lite}^\Fmc,\text{UBCQ})$. There exists a UBCQ $q'$ such
that
\begin{enumerate}

\item $Q$ is equivalent to $(\Omc^=,\Sbf_{\mn{full}},q')$;

\item $(\emptyset,\Sbf_{\mn{full}},q')$ is equivalent to
  $(\Omc^\sqsubseteq,\Sbf_{\mn{full}},q')$;

\item If $q \in \text{UBCQ}_k$ for some $k \geq 1$, then $q' \in \text{UBCQ}_k$.
\end{enumerate}
\end{lemma}
\begin{proof}
  Let $Q^\sqsubseteq = (\Omc^\sqsubseteq,\Sbf_{\mn{full}},q)$ and assume that $q' = \mn{rew}(Q^\sqsubseteq)$ is the UBCQ provided by Lemma~\ref{lem:rewriting-eval}.

  Then Point~1 is satisfied. Indeed, if
  $\Dmc$ is inconsistent with $\Omc^=$, then $\Dmc \models Q$ and
  $\Dmc \models (\Omc^=,\Sbf_{\mn{full}},q')$. Otherwise, $\Dmc
  \models Q$ iff $\Dmc \models Q^\sqsubseteq$ (by
  Lemma~\ref{lem:dlliteFindep}) iff $\Dmc \models q'$ (by
  Lemma~\ref{lem:rewriting-eval}).

  For Point~2, it suffices to show that
  $(\Omc^\sqsubseteq,\Sbf_{\mn{full}},q') \subseteq
  (\emptyset,\Sbf_{\mn{full}},q')$. Assume that $\Dmc \models
  (\Omc^\sqsubseteq,\Sbf_{\mn{full}},q')$. Then there is a
  homomorphism $h$ from $q'$ to $\mn{ch}_{\Omc^\sqsubseteq}(\Dmc)$.
  Let $\Dmc'$ be the restriction of \Dmc to the constants in the range
  of $h$. By construction, $\Dmc' \models q'$. Thus $\Dmc' \models
  Q^\sqsubseteq$, that is, there is a homomorphism $g$ from $q$ to
  $\mn{ch}_{\Omc^\sqsubseteq}(\Dmc')$. Since $\Dmc'$ is a subset of
  $\mn{ch}_{\Omc^\sqsubseteq}(\Dmc)$, there is also a homomorphism
  from $q$ to $\mn{ch}_{\Omc^\sqsubseteq}(\Dmc)$. We obtain that $\Dmc
  \models Q^\sqsubseteq$, and thus $\Dmc \models (\emptyset,\Sbf_{\mn{full}},q')$.

  Point~3 follows from Lemma~\ref{lem:rewriting-eval}.
\end{proof}

Now for the reduction. Let $Q_1=(\Omc,\Sbf_{\mn{full}},q_1)$ be an OMQ from $(\text{DL-Lite}^\Fmc,\text{UBCQ})$, and let $Q_2=(\Omc^\sqsubseteq,\Sbf_{\mn{full}},q_2)$, where  $q_2=\mn{id}_{\Omc^=}(q)$. We can show the following.
%
%
\begin{lemma}
\label{lem:dlliteToELI}
$Q_1$ is UBCQ$_1$-equivalent iff $Q_2$ is UBCQ$_1$-equivalent.
Moreover, if $Q_2$ is equivalent to $Q'_2 =(\Omc^\sqsubseteq,\Sbf_{\mn{full}},q'_2)$
with $q'_2 \in \text{UBCQ}_1$, then $Q_1 \equiv (\Omc,\Sbf_{\mn{full}},q'_2)$.
\end{lemma}
\begin{proof}
  The `if' direction is implied by the ``Moreover'' part. To prove the
  ``Moreover'' part, assume that $Q_2$ is equivalent to $Q'_2
  =(\Omc^\sqsubseteq,\Sbf_{\mn{full}},q'_2)$ with $q'_2 \in
  \text{UBCQ}_1$. We have to show that $Q_1$ and
  $Q'_1=(\Omc,\Sbf_{\mn{full}},q'_2)$ give the same answers on all databases
  \Dmc. If \Dmc is inconsistent with $\Omc$, then $\Dmc \models Q_1$ and
  $\Dmc \models Q'_1$. If \Dmc is consistent
  with $\Omc^=$, then $\Dmc \models Q_1$ iff $\Dmc \models Q_2$ by
  Lemma~\ref{lem:dlliteFindep}. Moreover, $\Dmc \models Q_2$ iff $\Dmc
  \models Q'_2$ iff $\Dmc \models Q'_1$, the
  latter by Lemma~\ref{lem:dlliteFindep}.

  \smallskip For the `only if' direction, assume that $Q_1$ is
  equivalent to an OMQ $Q'_1$ from
  $(\text{DL-Lite}^\Fmc,\text{UBCQ}_1)$. By
  Lemma~\ref{lem:dlliteFOrewr}, we can assume $Q'_1$ to be
  of the form
  $(\Omc^=,\Sbf_{\mn{full}},q'_1)$ with all BCQs in $q'_1$ of tree
  width~1. We can assume w.l.o.g.\ that
  \begin{enumerate}


  \item[($*$)] $\Dmc_{p}$ satisfies all functionality assertions in $\Omc^=$,
    for each BCQ $p$ in $q_1$.


  \end{enumerate}
  %
%
  In fact, we can achieve ($*$) by replacing each BCQ $p$ in $q'_1$
  with $\mn{id}_{\Omc^=}(p)$ which preserves 
  tree width~1 and, by Lemma~\ref{lem:dlliteFindep}, is also
  equivalence preserving regarding $Q'_1$.

  \smallskip

  We now show that $Q'_2=(\emptyset,\Sbf_{\mn{full}},q'_1)$ is equivalent
  to~$Q_2$.

  \smallskip `$\subseteq$'. Let \Dmc be a database with $\Dmc
  \models Q'_2$.
%
%
  Then there is a homomorphism from some BCQ $p$ in $q'_1$ to
  $\Dmc$. Clearly $\Dmc_{p} \models Q'_2$ and by construction of
  $Q'_2$, this implies $\Dmc_{p} \models Q'_1$ and thus also $\Dmc_p
  \models Q_1$. By ($*$), $\Dmc_{p}$ satisfies all the functionality
  assertions in $\Omc^=$. By Lemma~\ref{lem:dlliteFindep}, we thus
  have $\Dmc_{p} \models Q_2$.
  Since there is a homomorphism from $\Dmc_{p}$ to \Dmc,
  this yields $\Dmc \models Q_2$ as required.
%
%
  %
%
%

  \smallskip `$\supseteq$'. Let $\Dmc$ be a database with $\Dmc
  \models Q_2$.  Then there is a homomorphism from some BCQ $p$ in
  $q_2$ to $\mn{ch}_{\Omc^\sqsubseteq}(\Dmc)$. Clearly, we have
  $\Dmc_p \models Q_2$. By construction of $q_2$, $\Dmc_p$ satisfies all the functionality assertions in $\Omc^=$.  Lemma~\ref{lem:dlliteFindep}
  thus yields $\Dmc_p \models Q_1$ and, consequently, $\Dmc_p \models
  Q'_1$. Since $\Dmc_p$ satisfies all the functionality assertions in
  $\Omc^=$, this means $\Dmc_p \models Q_2'$. The homomorphism from $p$ to
  $\mn{ch}_{\Omc^\sqsubseteq}(\Dmc)$ yields $\Dmc \models
  (\Omc^\sqsubseteq,\Sbf_{\mn{full}},q'_1)$. By Point~2 of
  Lemma~\ref{lem:dlliteFOrewr}, this implies $\Dmc \models
  Q'_2$, as required.
\end{proof}


We can now give the proof of Theorem~\ref{the:main-dl-lite-f-acyclic}.

\begin{proof}
\noindent
Point~1 follows from the ``Moreover'' part of Lemma~\ref{lem:dlliteToELI} and Corollary~\ref{prop:wasaclaim}.

To prove Point~2, assume that an OMQ $Q=(\Omc,\Sbf_{\mn{full}},q)$ from $(\text{DL-Lite}^{\Fmc},\text{UBCQ})^{\equiv}_{\text{UBCQ}_1}$ is given as an input, along with a database \Dmc. We can check in \PTime whether \Dmc is consistent with $\Omc^=$. If it is not, then we answer `true'. Otherwise, by Lemma~\ref{lem:dlliteFindep} it suffices to evaluate $(\Omc^\sqsubseteq,\Sbf_{\mn{full}},q)$ in place of $Q$. This OMQ, however, is from $(\text{DL-Lite},\text{UBCQ})^\equiv_{\text{UBCQ}_1}$ and thus we can
invoke Theorem~\ref{thm:main1}.

The upper bound in Point~3 is an immediate consequence of Lemma~\ref{lem:dlliteToELI} and Point~3 of Theorem~\ref{thm:npcompl}. The lower bound is inherited from the case where the ontology is empty.
\end{proof}

\end{document}